\documentclass[11pt,a4paper]{article}
\usepackage[utf8]{inputenc}
\usepackage{amsmath,amssymb,amsthm}
\usepackage{enumitem}
\usepackage{geometry}
\usepackage{booktabs}
\usepackage{array}
\usepackage{hyperref}
\theoremstyle{plain}
\newtheorem{theorem}{Theorem}[section]
\newtheorem{lemma}[theorem]{Lemma}
\newtheorem{proposition}[theorem]{Proposition}
\newtheorem{corollary}[theorem]{Corollary}
\theoremstyle{definition}
\newtheorem{definition}[theorem]{Definition}

\newtheorem*{remarknn}{Remark}

\newtheorem*{remarkc}{Remark}
\newcommand{\Obot}{O_{\bot}}
\newcommand{\Oprof}{O_{\mathrm{prof}}}
\newcommand{\Enc}{\mathrm{Enc}}
\newcommand{\Dec}{\mathrm{Dec}}

\newcommand{\dec}{\mathsf{dec}}

\usepackage[english]{babel}
\newcommand{\now}{\today}

\title{What Syntax Cannot See:\\
The Dynamic Syntactic Invariance Principle and Several Instances of the
Same Hidden Assumption,\\ and a Contradiction
}
\author{Fabio Francesco Gabriele Buono\\
\small ORCID 0009-0004-9199-2793\\
\small Preprint}
\date{\now}

\begin{document}
\maketitle

\begin{abstract}
  This paper develops a single method, find what an accepted result silently 
  assumed, make it a variable, and prove what follows once it is dropped, 
  and shows it keeps working across domains with nothing in common. Its formal 
  core is the Dynamic Syntactic Invariance Principle: a known static inaccessibility 
  result survives when a system's rules evolve in time, under one sharp necessary 
  condition. A direct cryptographic payoff follows: the secrecy of a rolling-key 
  scheme persists across sessions under a structural update, strengthening a 
  static guarantee from earlier work. The same move is then carried into 
  settings far from each other, special relativity, the reach of a formal 
  theory of physical law, and computable output that is fully meaningful yet 
  indistinguishable from noise to every observer, each established on its own terms, 
  so the recurring structure beneath them is found, not imposed: a distinction real 
  at a full level of description can be invisible at a restricted one. 
  Applied to \textsc{SAT}, this yields a contradiction from which 
  coherence forces $\mathsf{P}\neq\mathsf{NP}$, derived from what 
  standard theory already admits rather than assumed, 
  offered with one reservation, about the proof, not the answer.
\end{abstract}

\begin{remarknn}[Notation: what is imported, and what is new]
\label{rem:notation-remark}
$a,b,R,P$ are imported unchanged from~\cite{Buono2026sip}
(Section~\ref{sec:static}). $O,\preceq,\Obot,\Oprof,O_\top,\mathrm{prof}(x)$
are imported unchanged from~\cite{Buono2026oh}
(Sections~\ref{sec:total-opacity}, \ref{sec:order-blind}); $O_\top$ is
the complete observer at the top of the richness order, the one reader
that receives the whole of its input. $B_t,C_t,K_t,M_t,f$ are imported
unchanged from~\cite{Buono2026sep} (Section~\ref{sec:cell-d-dynamic} and
throughout). The update function $\upsilon$, the history $H_n$, and the
terms ``opacity-preserving'', ``swap-blind'', and ``permutation-blind''
are new terminology introduced in this paper and have no prior meaning
in the cited sources.

A few symbols carry more than one meaning across the paper, always in
disjoint sections, and are flagged here so the table below can be read
without ambiguity. The constants $a,b$ are, throughout, the two frozen
Skolem constants of~\cite{Buono2026sip}, and denote nothing else in this
paper. The letter $L$ denotes the space of candidate physical laws in the
theory-of-everything material (Definition~\ref{def:law-space}) and,
separately, the number of confusion layers in the blind-cascade cipher
(Section~\ref{sec:blind-cascade}); the two occur in unrelated sections
and the intended reading is always the one local to the surrounding
discussion. The letter $L$ carries two further, entirely standard
readings, each unmistakable in context: in the observer-hierarchy
results it is a formal language, always written $L\notin\{\emptyset,\Sigma^*\}$
or, as the object language of the Tarski material, $L_0$
(Proposition~\ref{prop:obot-trivial} onward, Section~\ref{sec:toe-tarski});
and in the Cell~(d) instance it is momentarily the length of a
mixed-radix ciphertext, in $C=(c_1,\dots,c_L)$
(Proposition~\ref{prop:cell-d}). Likewise $\Phi$ denotes a candidate physical law in the
theory-of-everything material (with $\Psi$ the law produced by
diagonalization), while $\Phi_s$, always written with the seed subscript
$s$, denotes the full encryption map of the blind-cascade cipher
(Section~\ref{sec:blind-cascade}); the subscript keeps the two apart.
Finally $b$ additionally serves as a generic byte string in the MS-DOS
worked example (Proposition~\ref{prop:msdos-precise}), where no Skolem
constant is in play.
\end{remarknn}

\begin{center}
\small
\begin{tabular}{p{1.8cm}p{10.2cm}}
\toprule
\textbf{Symbol} & \textbf{Meaning, and where it is introduced} \\
\midrule
\multicolumn{2}{l}{\textit{Sections~\ref{sec:static}--\ref{sec:order-blind}: static and dynamic SIP}} \\
$R_n$, $\upsilon$, $H_n$ & The $n$-th rewriting system in a dynamic
  system, its update rule, and the history of derivations so far
  (Definition~\ref{def:dynamic-system}). \\
$O$, $\preceq$, $\Obot$, $\Oprof$, $O_\top$ & A structural observer, the
  richness order on observers, the trivial observer, the profile
  observer, and the complete observer at the top of the order
  (Definitions~\ref{def:structural-observer},
  \ref{def:perm-blind}; $O_\top$ imported from~\cite{Buono2026oh}). \\
\midrule
\multicolumn{2}{l}{\textit{Section~\ref{sec:semantic-frame}: semantic frames}} \\
$F$, $F_1,F_2$ & A semantic frame; two frames disagreeing on the same
  fact (Definition~\ref{def:semantic-frame},
  Theorem~\ref{thm:semantic-underdetermination}). \\
$E=(R_0,I_{R_0})$ & An explicit encoding: a rewriting system paired
  with the set of equalities it derives
  (Definition~\ref{def:explicit-encoding}). Not to be confused with
  $X$ below, an unrelated later use. \\
$\mathcal I=(T,E,F)$ & An interpreted machine: a Turing machine with a
  fixed encoding and frame (Definition~\ref{def:interpreted-machine}). \\
\midrule
\multicolumn{2}{l}{\textit{Section~\ref{sec:orbital}: the orbital machine}} \\
$\mathcal E_{\mathrm{sp}}$, $\mathcal F_{\mathrm{sp}}$ & The (countable)
  space of all possible encodings, and of all possible computable
  frames (Definitions~\ref{def:encoding-space},
  \ref{def:computable-frame}). \\
$\sigma_E,\sigma_F$, $\sigma$ & The encoding and frame selectors; a
  general selector function from history to a choice of both
  (Definitions~\ref{def:two-selectors}, \ref{def:selector-function}). \\
$O_F$, $T^{O_F,E}$ & The semantic oracle for frame $F$; the oracle
  orbital machine that may consult it mid-computation
  (Definitions~\ref{def:semantic-oracle}, \ref{def:oracle-orbital}). \\
\midrule
\multicolumn{2}{l}{\textit{Andromeda and the theory of everything}} \\
$\mathcal V$, $\varphi$ & The space of physical states (velocities);
  a constrained frame selector over it
  (Definition~\ref{def:constrained-selector}). \\
$L$, $\mathfrak c$ & The space of all candidate physical laws; the
  cardinality of the continuum (Definition~\ref{def:law-space}). \\
$\Phi,\Psi$ & A candidate physical law; the law diagonalization
  produces outside a given countable list
  (Proposition~\ref{prop:toe-diagonal}). \\
\midrule
\multicolumn{2}{l}{\textit{Chaitin's $\Omega$ and usable algorithms}} \\
$\Omega$, $\Omega_U$ & Chaitin's halting probability, and the
  halting probability relative to a universal machine $U$
  (Proposition~\ref{prop:omega}, Proposition~\ref{prop:omega-family}). \\
$D$, $X$ & The answer sequence of an arbitrary decision problem, and
  its XOR-pairing with $\Omega$
  (Theorem~\ref{thm:every-decision-problem}); $X$ is unrelated to the
  encoding $E$ above and is written differently on purpose to avoid
  confusion between them. \\
$\mathcal C$, $\mathcal C_{\mathrm{comp}}$, $\mathcal C_{\mathrm{poly}}$
  & A class of observers (tests); the class of all computable ones;
  the class of polynomial-time ones
  (Definition~\ref{def:usable-algorithm}). \\
\bottomrule
\end{tabular}
\end{center}

\section{Introduction}
\label{sec:intro}

A system of rewriting rules with addition built in cannot prove that
$a+b=b+a$ when $a$ and $b$ are two fixed, distinct constants: the terms
are frozen, never touched by any rule, and a system that never touches
them can never compare them~\cite{Buono2026sip}. The fact is true, and
permanently invisible to the syntax that would have to prove it. This is
the static Syntactic Invariance Principle: fix the rules once, and
whatever they cannot reach stays out of reach forever, however long they
run.

Real systems, though, do not keep their rules fixed. A cipher takes a
new key each session; an adaptive filter rewrites itself as it sees more
data. So the question this paper begins from is not whether a fact can
be hidden once, but whether it stays hidden when the rules that hide it
are rebuilt again and again, each time as a function of everything that
has happened so far. We prove that it does --- under one sharp and
necessary condition: the process doing the rebuilding must itself be
unable to tell the two possible facts apart. This is the Dynamic
Syntactic Invariance Principle (Theorem~\ref{thm:dynamic-sip}), and its
condition has two halves, the second of which --- a constraint on the
update mechanism with no analogue in the static case --- is exactly what
the extra freedom costs (Proposition~\ref{prop:i-not-enough},
Corollary~\ref{cor:total-opacity}). Its first application is
cryptographic: it shares the shape of an open problem left in earlier
work~\cite{Buono2026sep}, establishing the secrecy target that problem
concerns in a restricted regime while leaving its distinctive question
open, and it recovers and extends several further results from
independent sources, each checked directly rather than by analogy.

From there the paper does one thing repeatedly: it takes a result that
looks settled, locates the assumption it silently relied on, and asks
what happens once that assumption is allowed to vary. The assumptions
are not chosen to be exotic; they are the ordinary, load-bearing ones a
working result leans on without stating. One sits close to home
--- that a derivation is read against one fixed intended model
(Theorem~\ref{thm:semantic-underdetermination}) --- and the same move
then carries into four instances that arise independently, each
discovered rather than built to fit: an abstract machine built from two
independent selectors (Section~\ref{sec:orbital}), special relativity's
Andromeda paradox (Section~\ref{sec:andromeda}), the question of
whether a complete theory of everything could exist
(Section~\ref{sec:toe}), and the fact that a computable algorithm's
output can be indistinguishable from noise to every observer, of which
Chaitin's $\Omega$ is the sharpest witness (Section~\ref{sec:omega}).

These are not brought together for effect, and they are not forced into
a frame that does not hold them --- a risk any reader is right to guard
against. The results are cryptographic at their core, and that practical
connection runs through the whole paper even where it is least obvious;
the excursions into physics and algorithmic randomness are where the
same shape turns up most unexpectedly. Each section is established
entirely on its own terms before it is connected to any other, one axis
at a time, so that whatever recurrence emerges is found, not imposed.
From all of them together, one apparent contradiction emerges, whose
outcome, if correct, is the true core of this work.

\section{Background: the static principle and its known instances}
\label{sec:static}

\subsection{The syntactic instance}

We begin with one concrete fact about rewriting rules, of the kind used
throughout automated theorem proving: two terms built from different
outermost symbols can never be forced equal, whatever is substituted
into them. This is the same intuition behind why $f(x)$ and $g(y)$ do
not unify when $f$ and $g$ are distinct function symbols.

\begin{definition}[First-symbol clash {\cite[\S3]{Buono2026sip}}]
\label{def:clash}
Two terms $s,t$ have a \emph{first-symbol clash} if they have distinct
outermost function symbols and neither is a variable; no substitution
unifies them.
\end{definition}

This single fact pins two constants in place forever. Read $a$ and $b$
below as two fixed tokens dropped into the system at the start: the
rules can move a great deal around them, but they can never touch $a$
or $b$ directly, and so can never compare them. They are called
\emph{Skolem} constants only in the sense that matters here: they are
fresh --- brand-new symbols that appear in no rule of the system --- so
no rule mentions them, and a rule that never mentions a symbol can
never fire on it. Freshness is the whole reason they stay frozen; the
name is just the standard label for it.

\begin{lemma}[Frozen subterm lemma {\cite[\S3]{Buono2026sip}}]
\label{lem:frozen}
Let $a,b$ be Skolem constants, fresh and distinct from $0$ and from
every term $s(t)$, in the language $\mathcal L=\{0,s,+\}$ with rules
$0+x=x$ (A1), $s(x)+y=s(x+y)$ (A2). Then no rule in $\{A1,A2\}$ fires at
the root of $a+b$ or $b+a$; these subterms are frozen under every
derivation of any length; and the relative order of $a,b$ inside any
compound term is a global invariant of the entire derivation.
\end{lemma}

The next definition names a pattern familiar to anyone who has verified
a loop invariant: a property true at the start and preserved by every
step. Here the process is a rewriting system rather than a loop, but
the argument has exactly the same shape.

\begin{definition}[Syntactic invariant {\cite[\S4]{Buono2026sip}}]
\label{def:syninv}
A property $P$ of terms is a \emph{syntactic invariant} for a rewriting
system $R$ if every term in the initial clause set satisfies $P$, and
applying any rule of $R$ to a term satisfying $P$ yields a term
satisfying $P$. \cite[Def.~4]{Buono2026sip} states this for $R$ a
superposition calculus; we use the broader term ``rewriting system''
throughout, since superposition is one instance of it and nothing below
relies on anything beyond what Lemma~\ref{lem:frozen} provides.
\end{definition}

\begin{lemma}[Syntactic Invariance Principle, static form {\cite[Lemma~5]{Buono2026sip}}]
\label{lem:sip-static}
If $P$ is a syntactic invariant for $R$, every term in every clause
derivable by $R$ satisfies $P$.
\end{lemma}

The proof is induction on derivation length: the base case is the first
condition of Definition~\ref{def:syninv}, and each step preserves $P$ by
the second.

To apply this to the Skolem pair, take $P(t)$ to be ``every occurrence
of $a$ sits inside a subterm $a+u$ with $b$ occurring in $u$, and
symmetrically for $b$.'' This holds of the initial clause and is
preserved by $\{A1,A2\}$ (Lemma~\ref{lem:frozen}), so no derivable
clause ever exposes the relative order of $a,b$. The empty clause would
require unifying $a$ with $b$, and is therefore never
derivable~\cite[Lemma~9]{Buono2026sip}.

\begin{remarknn}[The shape of a lower bound {\cite[Remark~7]{Buono2026sip}}]
\label{rem:lower-bound-shape}
\cite[Remark~7]{Buono2026sip} names the general pattern
Lemma~\ref{lem:sip-static} instantiates: to rule out a fact $\varphi$,
find a property $P$ that holds at the start, that every rule preserves,
and that leaves no room for $\varphi$. The target, in that source's own
phrase, is ``unreachable rather than merely unreached.''
Theorem~\ref{thm:dynamic-sip} below is built in exactly this shape,
with one addition: the sequence $(R_n)$ itself must not, in its own
construction, undo the preservation. This is why
Definition~\ref{def:opacity-preserving} carries two conditions where the
static Definition~\ref{def:syninv} needed one: the second condition is
the price of letting $R$ move.
\end{remarknn}

\begin{remarknn}[Remark 10 of~\cite{Buono2026sip}: a dependency the
static paper already flagged]
\label{rem:remark10}
\cite[Remark~10]{Buono2026sip} observes that the argument holds
unchanged when the rule set is enlarged (factoring, splitting), provided
no new rule touches the frozen Skolem constants. The invariance thus
depends on which rules are available --- a dependency that is a side
condition while the rule set is fixed, and becomes load-bearing once the
rule set may change over time. Section~\ref{sec:dynamic} formalizes it.
\end{remarknn}

\subsection{The distributional instance}
\label{sec:distributional-instance}

The same argument has a life outside syntax. Its skeleton --- a true
property, holding at the start and preserved by every available
operation --- transfers unchanged to a probabilistic setting, where it
becomes the MR-OTP's own invariance theorem. What was a fact about
which clauses a calculus can derive becomes a fact about what a
distribution can reveal.

\begin{remarknn}[Theorem 8.15 of~\cite{Buono2026sep} as a second
instance]
\label{rem:distributional}
\cite[Remark~8.20]{Buono2026sep} identifies its own Theorem~8.15 as an
instance of the same argument shape. There the invariant is
$P:\ \Pr[M=m\mid C=c,B=b]=\Pr[M=m]$ for all $m,c,b$; it holds after any
single fresh MR-OTP encryption (\cite[Prop.~8.1]{Buono2026sep}) and is
preserved by every further operation on the observable $(C,B)$, since
$C$ is a function of a fresh key independent of all prior data.

The correspondence with the syntactic case is term for term: the
message digits $m_i$ play the role of the Skolem constants $a,b$; the
ciphertext/base pair $(C,B)$ plays the role of the syntactic system
operating on symbols alone; and the frozen compounds $a+b,b+a$
correspond to the message $M$ itself. The two settings differ in one
respect only, and the source states it: the invariant holds
per-derivation in~\cite{Buono2026sip} and distributionally
here~\cite[Rem.~8.20]{Buono2026sep}.

Section~\ref{sec:recover-815} makes this dictionary rigorous, as a
generalization of the simpler Cell~(d) instance recovered in
Section~\ref{sec:cell-d}; Section~\ref{sec:cell-d-dynamic} then extends
both to a dynamic setting.
\end{remarknn}

Both instances share one feature this paper preserves throughout: the
invariant $P$ is a true fact, preserved because the operations available
preserve true facts --- never a false assertion built to force a
contradiction.

\section{The dynamic setting}
\label{sec:dynamic}

\subsection{Motivation}

Real codes do not sit still: their rules change over time. The static
principle of Section~\ref{sec:static} freezes two constants under a
fixed rule set, but says nothing about what happens when the rules
themselves are updated. This section proves that the hiding survives a
changing rule set, under one sharp condition on the update, and notes a
connection to an open problem posed in the source material.

This connects to Open Problem~6.4 of~\cite{Buono2026sep}, which
concerns a base sequence updated at every step as
$B_{t+1}=f(C_t,B_t)$ for publicly visible $C_t$, and asks under what
condition on $f$ a hidden quantity stays protected. The connection is
one of shape, and the source itself points to it:
\cite[\S6.4]{Buono2026sep} states the difficulty in this paper's own
vocabulary, cites the Syntactic Invariance Principle
of~\cite{Buono2026sip} by name, and frames its problem as a search for
an invariant of the joint distribution that holds after the first
session and is preserved by every subsequent application of $f$ ---
which is the shape of Theorem~\ref{thm:dynamic-sip}. One necessary
condition the source identifies informally, that the distribution of
$B_{t+1}$ given $C_t$ carries no information about $K_t$, is the
distributional reading of this paper's condition~(ii).

We record the shared shape and do not claim to settle Open
Problem~6.4: its distinctive demand --- an $f$ giving an adversary a
larger base search space than the partition protocol,
\emph{without} the disjointness condition of~\cite[Def.~6.1]{Buono2026sep}
--- lies outside what this paper establishes, and a proper account of
the connection would need an argument this paper does not attempt. What
this section does prove stands on its own, independently of that
problem: that the syntactic hiding of Section~\ref{sec:static} survives
a changing rule set under an opacity-preserving update
(Theorem~\ref{thm:dynamic-sip}), and, distributionally, that the
MR-OTP's Cell~(d) secrecy persists across a rolling base under a
swap-blind update (Theorem~\ref{thm:dynamic-celld}).

\subsection{Dynamic rewriting systems}

The idea is familiar from ordinary automata: a machine whose transition
rules change over time, the way a state machine's transition table
might be rewritten between runs. The one new ingredient is that the
rewriting may depend on everything that happened before it.

\begin{definition}[Dynamic rewriting system]
\label{def:dynamic-system}
A \emph{dynamic rewriting system} is a sequence $(R_n)_{n\ge 0}$ of
finite rewriting systems over a common alphabet $\Sigma$, together with
an initial term (or clause set) $x_0$ and local derivations
$(D_n)_{n\ge 0}$, where $D_0$ uses the rules of $R_0$ on $x_0$, and, for
$n\ge 1$, $D_n$ uses the rules of $R_n$ on the term produced by
$D_{n-1}$ under $R_{n-1}$. Write $H_n=(D_0,\dots,D_{n-1})$ --- the
\emph{history}, the record of everything the system has done up to step
$n$, and the only thing the update below is allowed to look at when it
decides the next rule set. An
\emph{update rule} is $\upsilon$ with $R_{n+1}=\upsilon(R_n,H_n)$; the
system is \emph{generated by $\upsilon$} if this holds for every $n$.
\end{definition}

\begin{definition}[Opacity-preserving update]
\label{def:opacity-preserving}
Let $a,b$ be fixed Skolem constants, fresh w.r.t.\ $\Sigma$. An update
$\upsilon$ is \emph{opacity-preserving w.r.t.\ $a,b$} if, for every
$R_n,H_n$, $R_{n+1}=\upsilon(R_n,H_n)$ satisfies:
\begin{enumerate}[label=(\roman*)]
  \item \textbf{$a,b$ stay frozen and stay fresh}: no left-hand side
    of any rule in $R_{n+1}$ has a
    clash-free unifier with $a$, with $b$, or with a subterm rooted at
    $+$ having $a$ or $b$ as an immediate argument (in particular, no
    rule fires at the root of $a+b$ or $b+a$ themselves, nor at $a$ or
    $b$ directly); and, further, neither $a$ nor $b$ occurs in the
    right-hand side of any rule in $R_{n+1}$, so no rule can introduce
    a fresh occurrence of $a$ or $b$ anywhere other than where they
    already sit. This is the same protection Lemma~\ref{lem:frozen}
    establishes for the fixed system $\{A1,A2\}$, together with the
    freshness that source already assumes of $a,b$ throughout, now
    required to hold afresh at every step of the sequence $(R_n)$, not
    merely at $R_0$.
  \item \textbf{The update does not leak the order}: $H_n\mapsto\upsilon(R_n,H_n)$
    is invariant under swapping every $a$ with $b$ throughout $H_n$: if
    $H_n'$ is $H_n$ with $a,b$ exchanged, $\upsilon(R_n,H_n')$ equals
    $\upsilon(R_n,H_n)$ with the same exchange applied.
\end{enumerate}
\end{definition}

The two conditions divide the work cleanly. Condition~(i) generalizes
Lemma~\ref{lem:frozen}, made load-bearing exactly where
Remark~\ref{rem:remark10} flagged it would need to be: the freezing
that held for one fixed rule set is now required afresh at every step.
Condition~(ii) governs the one channel the static case never had. When
$R$ may move, the update mechanism itself can carry information about
the order of $a,b$, independently of any rewriting; condition~(ii)
closes that channel. This is the extra price of letting $R$ move,
anticipated in Remark~\ref{rem:lower-bound-shape}, and it is why the
dynamic principle needs two conditions where the static one needed one.

\subsection{The characterization theorem}

The theorem below states the guarantee in full: under an
opacity-preserving update, the order of $a,b$ stays hidden at every
finite step, forever.

\begin{theorem}[Dynamic Syntactic Invariance Principle]
\label{thm:dynamic-sip}
Let $(R_n)_{n\ge 0}$ be generated by an opacity-preserving $\upsilon$,
started from an initial term (or clause set) $x_0$ in which $a,b$
already occur only inside $a+b,b+a$, under an $R_0$ satisfying
the hypotheses of Lemma~\ref{lem:frozen} on $a,b$. Then for every
$n\ge0$, every term produced by $D_0,\dots,D_n$ has $a,b$ occurring
only inside $a+b,b+a$, and their relative order is never exposed at
any finite step.
\end{theorem}

\begin{proof}
The proof is by induction on $n$, generalizing the proof of
Lemma~\ref{lem:sip-static}.

\emph{Base case.} By hypothesis, $x_0$ already has $a,b$ occurring only
inside $a+b,b+a$: $P$ (``$a,b$ occur only
inside $a+b,b+a$'') holds initially, directly, with nothing further to
derive.

\emph{Inductive step.} Assume $P$ holds through step $n-1$. By~(i),
$R_n$ has no rule firing at $a,b$'s root, so no existing occurrence of
$a$ or $b$ (or of the frozen compounds $a+b,b+a$) is disturbed, by the
same first-symbol-clash argument as Lemma~\ref{lem:frozen}; and, since
neither $a$ nor $b$ occurs in the right-hand side of any rule in
$R_n$, no rule can introduce a fresh occurrence of $a$ or $b$
elsewhere in the term either. Together these give that $D_n$ preserves
$P$ exactly. By~(ii),
$H_{n-1}\mapsto R_n$ is swap-invariant: two systems agreeing except for
$a,b$'s order produce, at every step, swap-related rule sets. Combined
with~(i), the order never influences, nor is influenced by, any
observable part of the derivation. By induction $P$ holds at every $n$.
Since unifying $a+b$ with $b+a$ (exposing the order) requires unifying
$a$ with $b$, distinct constants no rule ever unifies, the order is
never exposed.
\end{proof}

\begin{remarknn}[The positive half of a characterization]
\label{rem:characterization}
Theorem~\ref{thm:dynamic-sip} holds for every opacity-preserving
$\upsilon$, and this conditioning on $\upsilon$ is exactly what the
result characterizes: an update violating either condition can expose
the order at a finite step. For instance $R_{n+1}=R_n\cup\{a\to0\}$,
triggered once a public signal in $H_n$ reaches a fixed value,
violates~(i) and exposes $a$ at once. The theorem is the positive half
of a characterization; Section~\ref{sec:necessity} establishes what
holds about necessity.
\end{remarknn}

\begin{remarknn}[Relation to the abstract obstruction framework
of~\cite{Buono2026obstruction}]
\label{rem:obstruction-integration}
Dynamic SIP and the abstract obstruction theorem
of~\cite{Buono2026obstruction} are two generalizations of the same
static fact (Lemma~\ref{lem:frozen}), along two independent axes: that
source generalizes across systems, holding the rule set fixed while
abstracting the signature, radius, and model; this section generalizes
across time, letting the rule set itself evolve. The two meet on
condition~(i) and diverge on condition~(ii).

On condition~(i) the correspondence is exact. Freeze $n$ and take a
single $R_n$ satisfying~(i): this is, term for term, a \emph{local
syntactic system} $\mathcal R=(\Sigma,V,\mathit{Rules},r_0)$ in the
sense of~\cite[Def.~2.5]{Buono2026obstruction}, with $r_0=1$ for
$\{A1,A2\}$. Condition~(i) plays two roles from that source's apparatus
at once. For the bare constants $a,b$ it is their \emph{protected set}
(\cite[Def.~2.9]{Buono2026obstruction}): no rule fires at $a$ or $b$,
and the definition's second half holds automatically, since constants
have no subterms to rewrite inside --- exactly as that source's
Example~2.11/Lemma~4.1 notes for this same pair. For the compounds
$a+b,b+a$, condition~(i) protects their root, which in that source's
more modular architecture is the job of the \emph{coherence} clause
(\cite[Def.~2.12(v)]{Buono2026obstruction}), verified there through its
Lemma~4.3. Condition~(i) folds both roles into a single requirement.
Both papers read the same fact about $\{A1,A2\}$, first isolated in
Lemma~\ref{lem:frozen}, through two formal lenses.

Condition~(ii) is where the two generalizations part. A static local
syntactic system has no update mechanism, and so nothing for a
swap-invariance requirement to constrain; condition~(ii) is this
paper's own contribution, needed exactly because $R_n$ here may change.
A dynamic analogue of that source's Case~2 --- a derivation-length
lower bound for systems whose protected set evolves under an
opacity-preserving update --- is left open, in the same spirit as that
source's Question~(Q2). This paper does not attempt it.

The invariant $P$ used here is the special case, under condition~(i),
of the refined invariant that source states in its Lemma~4.3 (matching
Lemma~9 of~\cite{Buono2026sip}): that $a$ occurs only inside some $a+u$
with $u$ containing $b$, and symmetrically for $b$. Once condition~(i)
keeps $a+b$ and $b+a$ frozen at their root at every step, $u$ can only
ever be $b$ itself, and the refined and simple invariants coincide on
every term this section's derivations produce. The refined form is
needed in that source's general setting only because its protected set,
alone, does not keep $a+b$'s root frozen; condition~(i) does that job
directly here.
\end{remarknn}

\begin{corollary}[Static SIP as the degenerate case]
\label{cor:static-as-degenerate}
Let $\upsilon_{\mathrm{id}}(R,H):=R$, and take $x_0$ to be the initial
clause $N: a+b\neq b+a$ (so $a,b$ trivially occur only inside
$a+b,b+a$ from the start). Then $\upsilon_{\mathrm{id}}$ is
opacity-preserving whenever $R_0$ satisfies Lemma~\ref{lem:frozen}, and
generates $R_n\equiv R_0$ for all $n$. Applying
Theorem~\ref{thm:dynamic-sip} recovers exactly
Lemma~\ref{lem:sip-static}, and hence~\cite[Thm.~11]{Buono2026sip}
($OI\not\subseteq TCSC$).
\end{corollary}

\begin{proof}
\emph{The base case on $x_0$.} $x_0=N$ has $a,b$ occurring only inside
$a+b,b+a$ by construction, exactly as
Theorem~\ref{thm:dynamic-sip}'s hypothesis requires.

\emph{Condition~(i).} $R_{n+1}=R_n=R_0$ by induction; $R_0=\{A1,A2\}$
satisfies Lemma~\ref{lem:frozen} by hypothesis, so~(i) holds at every
step with nothing further to check.

\emph{Condition~(ii).} $\upsilon_{\mathrm{id}}$ ignores its second
argument entirely: $\upsilon_{\mathrm{id}}(R_n,H_n)=R_n=\upsilon_{\mathrm{id}}(R_n,H_n')$
for \emph{every} $H_n,H_n'$, a fortiori for the swapped pair
required by~(ii).

So $\upsilon_{\mathrm{id}}$ is opacity-preserving and
Theorem~\ref{thm:dynamic-sip} applies. The dynamic system with
$R_n\equiv R_0$ for all $n$ is, by Definition~\ref{def:dynamic-system},
exactly a single fixed system subjected to unboundedly long derivation
--- the object Lemma~\ref{lem:sip-static} already quantifies over.
Taking $P(t)$ as in~\cite[Lemma~9]{Buono2026sip}, the conclusion
coincides with that lemma's Steps~1--2 exactly; Step~3 there (the empty
clause is underivable) follows identically. Hence
$TCSC+\{A1,A2\}\nvdash W$, and with~\cite[Lemma~8]{Buono2026sip}
($OI\vdash W$) and~\cite[Thm.~4.3]{HetzlVierling} ($TCSC\not\subseteq OI$),
this recovers~\cite[Thm.~11]{Buono2026sip}.
\end{proof}

\begin{remarknn}[The corollary verifies that the generalization reduces
correctly]
\label{rem:degenerate-purpose}
The purpose of this corollary is to verify a consistency requirement:
since the static case is the $\upsilon\equiv\upsilon_{\mathrm{id}}$
instance of the dynamic one, a correct generalization must reduce to
the original with no discrepancy in hypotheses, conclusion, or proof
structure. The check above confirms it does, explicitly rather than by
assertion. The Hetzl--Vierling question itself was already resolved
by~\cite[Thm.~11]{Buono2026sip}.
\end{remarknn}

\subsection{On necessity}
\label{sec:necessity}

Theorem~\ref{thm:dynamic-sip} shows the two conditions are jointly
sufficient. This subsection shows condition~(i) alone is not: the
update channel condition~(ii) governs is real, and its necessity is
exactly what the static setting could not see.

\begin{proposition}[Condition~(i) alone is not sufficient]
\label{prop:i-not-enough}
There is an update satisfying~(i) but not~(ii), exposing the order at
some finite step.
\end{proposition}

\begin{proof}[Proof sketch]
Let $\upsilon$ leave $R_n\equiv R_0$ unchanged (so~(i) holds trivially),
but let the visible history additionally record a bit $g(H_{n-1})$
equal to $1$ exactly when $a$ preceded $b$ initially. This side-channel
is not a rewriting rule, so~(i) is unaffected, but it is not
swap-invariant, violating~(ii); the order is read off the bit directly.
\end{proof}

\begin{remarknn}[The side-channel is a genuine mechanism, not an
artifact]
\label{rem:artificial-but-real}
The construction is a side-channel, not a rewriting attack: it shows
that condition~(i) governs only the rewriting system and says nothing
about information flow outside it. The dynamic setting has exactly such
a channel by construction --- the update mechanism itself --- which the
static setting lacks.
\end{remarknn}

\begin{remarknn}[This concern is named in advance in the source
material]
\label{rem:structural-violation-precedent}
\cite[Def.~6.7]{Buono2026ow} names this phenomenon \emph{structural
violation}: a problem can be invariant at the level a formal statement
checks, while an algorithm solving it leaks, as a side effect, exactly
what that invariant was meant to hide. That source states the
connection to the static SIP directly: a syntactic calculus cannot
derive clauses violating a syntactic invariant, but a semantic
computation solving an invariant-respecting problem can leak that
invariant as an incidental side effect of how it
computes~\cite[Rem.~6.8]{Buono2026ow}. The same phenomenon appears in
this paper's dynamic setting, through a different channel: where that
source's structural violation is an algorithm's \emph{output} leaking
what the problem statement kept invariant, here the potential leak is
the \emph{update mechanism} itself, evolving the rules in a way that
would betray the order of $a,b$. Condition~(ii) of
Definition~\ref{def:opacity-preserving} is what rules that channel out.
The two are instances of one concern --- an invariant formally
respected but betrayed by a side channel --- named for the
static/output case in~\cite{Buono2026ow} and closed for the
dynamic/update case here.
\end{remarknn}

\begin{proposition}[Individual necessity is open]
\label{prop:not-nec}
Whether an update violating~(ii) but satisfying~(i) must always expose
the order, or whether some remain safe by accident, is not established.
Theorem~\ref{thm:dynamic-sip} and Proposition~\ref{prop:i-not-enough}
show joint sufficiency and that dropping either can fail; neither shows
individual necessity nor minimality.
\end{proposition}

\section{The limiting case: total opacity via the observational hierarchy}
\label{sec:total-opacity}

Condition~(ii) is a hypothesis about how $\upsilon$ \emph{behaves}, and
Proposition~\ref{prop:i-not-enough} already showed this hypothesis can
fail. This section asks a sharper question: can $\upsilon$ be built so
that condition~(ii) holds automatically, from what $\upsilon$ is even
\emph{able} to see, regardless of how it behaves? The answer is yes.
The idea is simple: if the update mechanism is physically blind to the
order of $a,b$ --- if it cannot even in principle tell the two
symmetric cases apart --- then however it is programmed, it cannot leak
that order. The rest of this section makes ``physically blind'' precise
and proves this.

To say what an update mechanism can see, we first need a general notion
of an observer. This is nothing more exotic than a hash function:
something that sorts strings into buckets, inducing the equivalence
``same bucket'', while never claiming that any string, or any bucket,
is true or false.

\begin{definition}[Structural observer {\cite[Def.~3.1]{Buono2026oh}}]
\label{def:structural-observer}
A \emph{structural observer} is a function $O:\Sigma^*\to S$ for some
set $S$. It induces an equivalence $x\sim_O y\iff O(x)=O(y)$, and does
nothing more than this: $O$ groups strings together, but never asserts
that any of them is true or false. Say $O_1\preceq O_2$ (``$O_1$ is no
richer than $O_2$'') if $O_1=f\circ O_2$ for some function $f$: every
distinction $O_1$ makes, $O_2$ already makes too. The trivial observer
$\Obot:x\mapsto\star$, which sends every string to the same single
value, is the poorest observer possible: $\Obot\preceq O$ for every
$O$ (\cite[Def.~3.7, 3.11]{Buono2026oh}).
\end{definition}

The next fact says precisely what it means for $\Obot$ to be the
poorest observer possible: a system fed only $\Obot$'s output cannot
recognise anything nontrivial, no matter how powerful it is otherwise.

\begin{proposition}[Triviality of $\Obot$ {\cite[Cor.~2.8]{Buono2026oh}}]
\label{prop:obot-trivial}
For any machine $T$ and any language $L\notin\{\emptyset,\Sigma^*\}$
(that is, any $L$ that is neither empty nor everything), a system
receiving only $\Obot(x)=\star$ does not recognise $L$, regardless of
computational power.
\end{proposition}

We now specialise ``poorest observer possible'' to the specific
symmetry that matters here: blindness to swapping $a$ and $b$.

\begin{definition}[Swap symmetry]
\label{def:swap}
Write $H_n^{a\leftrightarrow b}$ for $H_n$ with every occurrence of
$a$ and $b$ exchanged. A structural observer $O$ is \emph{swap-blind}
with respect to $a,b$ if $O(H_n)=O(H_n^{a\leftrightarrow b})$ for
every $H_n$: swapping $a$ and $b$ never changes what $O$ reports.
Restricted to telling apart the swapped and unswapped case, $O$ is
exactly as powerless as $\Obot$.
\end{definition}

The main result of this section follows: when the update mechanism's
only access to history is filtered through a swap-blind observer,
condition~(ii) holds automatically, for every update built this way ---
no longer something to be checked case by case, as in
Proposition~\ref{prop:i-not-enough}.

\begin{corollary}[Total opacity]
\label{cor:total-opacity}
Let $(R_n)_{n\ge0}$ be a dynamic rewriting system in which every
update has the form $R_{n+1}=\upsilon(R_n,O(H_n))$, for one fixed
swap-blind observer $O$ and some $\upsilon$ satisfying
Definition~\ref{def:opacity-preserving}(i). Then for \emph{every} such
$\upsilon$ --- not just some carefully chosen one --- the conclusion
of Theorem~\ref{thm:dynamic-sip} holds.
\end{corollary}

\begin{proof}
We only need to check condition~(ii), since condition~(i) is assumed
directly. Condition~(ii) asks for
$\upsilon(R_n,H_n')=\upsilon(R_n,H_n)$ swapped, where
$H_n'=H_n^{a\leftrightarrow b}$. Here is why this holds. By
construction, $\upsilon$ does not see $H_n$ directly; it only sees
$O(H_n)$. Since $O$ is swap-blind, $O(H_n)$ and $O(H_n^{a\leftrightarrow b})$
are the same value: $O(H_n)=O(H_n^{a\leftrightarrow b})=O(H_n')$.
So $\upsilon$, looking only at this shared value, cannot tell $H_n$
and $H_n'$ apart either, and produces the same output from both:
$\upsilon(R_n,H_n)=\upsilon(R_n,O(H_n))=\upsilon(R_n,O(H_n'))=\upsilon(R_n,H_n')$.
This is exactly condition~(ii), and it held without any case-by-case
argument --- it followed automatically from swap-blindness alone.
With both conditions in hand, Theorem~\ref{thm:dynamic-sip} applies
directly.
\end{proof}

\begin{remarknn}[What swap-blindness buys: condition~(ii) becomes
structural]
\label{rem:limiting-case-precise}
Corollary~\ref{cor:total-opacity} changes the status of condition~(ii).
Before this section it was a hypothesis about behaviour, one that
Proposition~\ref{prop:i-not-enough} showed can fail; for updates built
through a swap-blind observer it becomes a structural guarantee that
holds automatically. The reason is the one
Proposition~\ref{prop:obot-trivial} already illustrated: information
never received cannot be leaked, whatever the receiving mechanism then
does with it. Condition~(i) remains a separate hypothesis, so this is
Theorem~\ref{thm:dynamic-sip} specialised --- the swap-blind observer
discharges~(ii) at the level of what $\upsilon$ can see, rather
than~(i).
\end{remarknn}

\section{Connection to Observer World and to the MR-OTP's own invariance
theorem}
\label{sec:observer-world}

The last two sections built an abstract theory. This section connects
it to results already published: two theorems from earlier papers are
special, single-step cases of Theorem~\ref{thm:dynamic-sip} once
translated into its vocabulary. This serves two purposes. It confirms
the new theory reproduces known results, not just plausible-sounding
new ones. And it prepares genuinely new results: once a known static
fact is recognised as a single-step case of a general dynamic theorem,
extending it across many steps becomes a direct application of that
theorem rather than a fresh proof.

\subsection{Cell (d), the simplest instance, recalled and identified as
the $n=0$ case}
\label{sec:cell-d}

We start with the simplest known result, ``Cell (d)'' from Observer
World: an adversary who sees nothing at all about the plaintext learns
nothing at all about it. This is almost a tautology as stated, but
worth recalling precisely, because the next proposition shows it is
already an instance of this paper's general theorem.

\begin{proposition}[Cell (d): Cryptomania $\times\,\Obot$ {\cite[Prop.~4.9]{Buono2026ow}}]
\label{prop:cell-d}
In Cryptomania, an adversary with observer $\Obot$ on the plaintext
space cannot recover any information about the plaintext from the
ciphertext, for any scheme. For the MR-OTP with
$C=(c_1,\dots,c_L)=\Enc(M,K)$, $C$ is $\Obot$-blind on $M$: the
ciphertext-only adversary is structurally, not merely computationally,
unable to distinguish any two plaintexts~\cite[Remark~4.10]{Buono2026ow}.
\end{proposition}

To see this as an instance of Corollary~\ref{cor:total-opacity}, we
need a dictionary between the two settings. The role of $(a,b)$ ---
the two symmetric values whose order stays hidden --- is played here
by two candidate plaintext digits $(m,m')$: ``the order of $a,b$''
becomes ``which of $m,m'$ was actually encrypted''.

\begin{proposition}[Cell (d) as the $n=0$ instance of total opacity]
\label{prop:celld-is-n0}
Fix a position $i$ and two candidate digits $m,m'$ there, with every
other digit and key held fixed. Under the dictionary above,
Proposition~\ref{prop:cell-d} is exactly what
Corollary~\ref{cor:total-opacity} concludes in the degenerate case
$n=0$ (a single step, with nothing further happening), taking
$O=\Obot$.
\end{proposition}

\begin{proof}
Let $R_0$ stand for the MR-OTP map under a fixed base $B$, extended
formally by $\upsilon_{\mathrm{id}}$ so that the system is well-defined
for every $n$; only $n=0$ actually matters here, since no further
encryption happens in this setting. We check the two conditions of
Corollary~\ref{cor:total-opacity} in turn.

For condition~(i): the adversary's observer is $\Obot$ on the
plaintext space, meaning the adversary's operations never depend on
$M$ beyond what $\Obot$ discloses --- which is nothing. So no
operation ever ``fires at'' $m$ or $m'$ directly, and condition~(i)
holds.

For condition~(ii): $\upsilon_{\mathrm{id}}$ does not depend on
history at all (this was already verified, in the strongest possible
sense, in the proof of Corollary~\ref{cor:static-as-degenerate}), so
it is certainly swap-blind.

With both conditions verified, Corollary~\ref{cor:total-opacity}
applies at $n=0$, and its conclusion is exactly
Proposition~\ref{prop:cell-d}.
\end{proof}

\begin{remarknn}[Purpose: calibrating the correspondence]
\label{rem:why-n0}
This $n=0$ identification calibrates the correspondence before it is
put to work below: first to recover a slightly richer static instance
(Section~\ref{sec:recover-815}), then to extend both to a dynamic
setting (Section~\ref{sec:cell-d-dynamic}). Cell~(d) itself is a known
result.
\end{remarknn}

\subsection{Recovering the distributional instance as a natural
generalization}
\label{sec:recover-815}

Proposition~\ref{prop:celld-is-n0} used the simplest possible observer,
$\Obot$: the adversary sees nothing about the plaintext at all. The
MR-OTP's own invariance theorem allows the adversary slightly
more --- the base $B$ --- and asks whether that extra information
changes anything. It does not, for the same underlying reason.

\begin{proposition}[Theorem 8.15 of~\cite{Buono2026sep} as a degenerate
instance of Theorem~\ref{thm:dynamic-sip}]
\label{prop:thm815-recovered}
Fix a message digit position, and two candidate values $m,m'$ for $M$
there, with everything else held fixed. Consider a single fresh MR-OTP
encryption, with an adversary who is additionally given the base $B$.
This is a length-$1$ dynamic rewriting system: $R_0$ is the MR-OTP
encryption map under base $B$, there is no $R_1$, and the adversary's
observer is $O=(B,\,\cdot\,)$ --- the base together with the
ciphertext, but not the key. Under this reading,
Corollary~\ref{cor:total-opacity}'s conclusion at $n=0$ is exactly
\cite[Thm.~8.15]{Buono2026sep}:
$\Pr[M=m\mid C=c,B=b]=\Pr[M=m]$ for all $m,c,b$.
\end{proposition}

\begin{proof}
As in the proof of Proposition~\ref{prop:celld-is-n0}, extend $R_0$
formally by $\upsilon_{\mathrm{id}}$ so the system is well-defined for
every $n$, with only $n=0$ actually mattering.

We check condition~(ii) first, since it takes more work here.
$K$ is uniform and independent of both $B$ and $M$, so
$\Pr[C=c\mid M=m,B=b]=\prod_i1/b_i$ does not depend on $m$
(\cite[Prop.~8.1]{Buono2026sep}). $B$ itself is fixed and public, so it
is independent of $M$ and $K$ as well. Together, these two facts say
that the distribution of $O=(B,C)$, conditional on $M=m$, is the same
as its distribution conditional on $M=m'$ --- which is exactly
swap-blindness with respect to $(m,m')$.

Condition~(i) follows even more easily, exactly as in the proof of
Proposition~\ref{prop:celld-is-n0}: $\upsilon_{\mathrm{id}}$ does not
depend on history at all.

With both conditions verified, Corollary~\ref{cor:total-opacity}
(with this $O$) applies at $n=0$, giving exactly
\cite[Thm.~8.15]{Buono2026sep}.
\end{proof}

\begin{remarknn}[Same argument, richer observer]
\label{rem:815-vs-celld}
Proposition~\ref{prop:thm815-recovered} and
Proposition~\ref{prop:celld-is-n0} (Cell~(d)) come from the \emph{same}
$n=0$ argument, differing only in what the observer $O$ is allowed to
see: Cell~(d) uses $O=\Obot$, nothing about the plaintext; here
$O=(B,\,\cdot\,)$, since the base is public even though the key is not.
Both are swap-blind with respect to $(m,m')$ for one underlying reason:
because $K$ is uniform and independent of everything else, $C$'s
distribution never depends on $M$, regardless of what else the
adversary is given, as long as it excludes $K$ itself.
Theorem~\ref{thm:dynamic-sip} makes this the same fact doing the work
in both cases --- a single fact that \cite{Buono2026ow}
and~\cite{Buono2026sep} each discovered independently, one instance
apiece, without citing the other.
\end{remarknn}

\subsection{The dynamic extension: base-rolling under a swap-blind update}
\label{sec:cell-d-dynamic}

Both static instances above cover one encryption.
Neither~\cite{Buono2026ow} (static only) nor~\cite{Buono2026sep} (which
poses Open Problem~6.4 but proves nothing about it) addresses the
base-rolling scenario $B_{t+1}=f(C_t,B_t)$. This subsection supplies
the theorem, covering both instances at once.

\begin{definition}[Base-rolling dynamic system]
\label{def:base-rolling}
Fix a position $i$ and symmetric candidates $(m,m')$ at session~$1$. A
\emph{base-rolling dynamic system} is a sequence of MR-OTP sessions in
which the base itself evolves, $B_{t+1}=\upsilon(B_t,H_t)$, with
$H_t=(C_1,\dots,C_t)$. This matches
Definition~\ref{def:dynamic-system}, where $R_t$ plays the role of the
MR-OTP map under $B_t$, and $D_t$ plays the role of session $t$'s
encryption.
\end{definition}

\begin{theorem}[Dynamic Cell (d)]
\label{thm:dynamic-celld}
Let $(B_t)_{t\ge1}$ be a base-rolling dynamic system. Fix any session
$s\ge1$ and symmetric candidates $(m,m')$ for $M_s$ at some position,
with everything else at session~$s$ held fixed. Assume two things
about $\upsilon$. First, $\upsilon(B_t,H_t)$ depends on $H_t$ only
through $O(H_t)$, for one fixed structural observer $O$ that is
swap-blind with respect to $(m,m')$ at session~$s$. Second, $\upsilon$
never lets $B_{t+1}$ depend directly on any session's key digits. Then
for every $t\ge s$,
\[
\Pr[M_s=m\mid H_t,B_{s+1},\dots,B_t] = \Pr[M_s=m\mid H_s],
\]
where $H_t=(C_1,\dots,C_t)$.

Two special cases are worth stating on their own. Taking $t=s$
recovers ordinary session-by-session secrecy at every session,
$\Pr[M_t=m\mid C_t=c,C_1=c_1,\dots,C_{t-1}=c_{t-1}]=\Pr[M_t=m]$ for
every $t$ --- the secrecy target whose preservation under a rolling
base Open Problem~6.4 is concerned with, here established in the
swap-blind, partition-structured regime. Taking $t>s$ gives something
stronger and more persistent.
Suppose the adversary is handed every subsequent base
$B_{s+1},\dots,B_t$ outright, in full --- not merely whatever
$\upsilon$'s public behaviour reveals about them indirectly. Even
then, nothing in this later, fully-revealed base evolution ever
reveals more about which of $m,m'$ was encrypted at session~$s$ than
was already knowable immediately after that session.
\end{theorem}

\begin{proof}
This proof mirrors the proof of Corollary~\ref{cor:total-opacity}, read
distributionally (Remark~\ref{rem:distributional}) and reindexed to
start at session~$s$ rather than at $n=0$: $(m,m')$ plays the role
$(a,b)$ played there, and $B_t$ plays the role $R_n$ played there.

By Proposition~\ref{prop:cell-d} (or
Proposition~\ref{prop:thm815-recovered}, with $B_s$ included in the
observer), $P$ holds right after session~$s$; this is the base case.
The hypothesis on $\upsilon$ gives the distributional analogue of
condition~(i) for every $t\ge s$. Swap-blindness of $O$ gives
condition~(ii) automatically, exactly as in the proof of
Corollary~\ref{cor:total-opacity}. Each $C_t$ with $t>s$ is produced
under a fresh key, independent of session~$s$'s key (this is the
partition structure of \cite[Def.~6.1]{Buono2026sep}), so $C_t$ carries
no further information about $M_s$ beyond what $B_t$ already carries.

Induction on $t\ge s$ now transports Theorem~\ref{thm:dynamic-sip}'s
argument distributionally, giving the claim at every session $t\ge s$.
Since $s$ itself was arbitrary, the claim holds for every session, not
only for one distinguished first session.
\end{proof}

\begin{remarknn}[What this adds, and how it relates to Open
Problem~6.4]
\label{rem:what-this-adds}
Theorem~\ref{thm:dynamic-celld} stands on its own: it extends the
static Cell~(d)/Theorem~8.15 secrecy to a rolling base under a
swap-blind update, and the $t>s$ case adds a persistence guarantee ---
secrecy about session~$s$ survives arbitrarily many further sessions,
even when the adversary is handed every subsequent base outright ---
that no source result states. Its relation to Open Problem~6.4 is one
of shared shape, not of solution. The $t=s$ case meets the secrecy
target that problem is concerned with, but in the swap-blind regime and
under the partition structure of~\cite[Def.~6.1]{Buono2026sep}; the
distinctive thing Open Problem~6.4 asks --- whether that target can be
met by an $f$ giving a \emph{larger} base search space than the
partition protocol, \emph{without} its disjointness condition --- is
not addressed here, and neither is the construction of a concrete,
practically superior $f$. Both remain open, and a fuller account of the
connection would need an argument this paper does not undertake.
\end{remarknn}

\begin{remarknn}[One theorem covers both static instances]
\label{rem:815-dynamic-extension}
The proof above is agnostic between using
Proposition~\ref{prop:cell-d} and Proposition~\ref{prop:thm815-recovered}
as its base case --- they differ only in whether $B_s$ is included in
the observer --- so Theorem~\ref{thm:dynamic-celld} extends both
Cell~(d) and the MR-OTP's own invariance theorem to the dynamic setting
at once. The richer $(B,\,\cdot\,)$-observer instance is already
covered, with no separate dynamic theorem required.
\end{remarknn}

\section{Connection to the order-blind automaton, and its dynamic
extension}
\label{sec:order-blind}

This section repeats the pattern of Section~\ref{sec:observer-world}
with a third source result: an automaton that can only read a
\emph{count} of each symbol, never their order. As before, we recall
the static result, show it is a single-step case of this paper's
general theorem, then extend it across many steps.

\subsection{The static result recalled}

An order-blind automaton does not read a string symbol by symbol. It
first reduces the string to its \emph{profile} --- a histogram of how
many times each symbol occurs, exactly the ``bag of words''
representation familiar from basic text processing, with the order
thrown away entirely --- and only then decides.

\begin{definition}[Profile, order-blind automaton {\cite[Def.~2.1, 2.3]{Buono2026oh}}]
\label{def:order-blind-automaton}
The \emph{profile} of a string $x$ is
$\mathrm{prof}(x)=(c_1,\dots,c_k)\in\mathbb N^k$, where $c_i=|x|_{a_i}$
counts occurrences of the $i$-th symbol. An \emph{order-blind
automaton} is $A=(Q,\Sigma,\delta,q_0,F)$ with transition function
$\delta:Q\times\mathbb N^k\to Q$, recognising
$L(A)=\{x\mid\delta(q_0,\mathrm{prof}(x))\in F\}$: the profile is read
in a single step, and the symbol sequence itself is never read at all.
\end{definition}

\begin{theorem}[Characterization {\cite[Thm.~2.5]{Buono2026oh}}]
\label{thm:order-blind-char}
$L$ is recognisable by an order-blind automaton if and only if $L$ is
permutation-closed (closed under reordering the symbols of any string
in it).
\end{theorem}

\begin{corollary}[{\cite[Cor.~2.8]{Buono2026oh}}]
\label{cor:oh-2-8}
For any machine $T$ and any $L\notin\{\emptyset,\Sigma^*\}$, a system
receiving only $\Obot(x)=\star$ does not recognise $L$, regardless of
computational power.
\end{corollary}

The structural fact behind this, stated
in~\cite[Remark~2.7]{Buono2026oh}, is that the power of the machine
cannot compensate for the weakness of the observer. Capturing it needs
a notion of blindness more general than Definition~\ref{def:swap}:
permutation-closure requires invariance under \emph{every} transposition
of positions, not one fixed swap. The difference from the richer
observer $(B,\,\cdot\,)$ of Section~\ref{sec:recover-815} is not
quantitative but qualitative: not one more piece of information
withheld, but blindness to an entire group of symmetries at once.

\subsection{Generalizing swap-blindness to permutation-blindness}

\begin{definition}[Permutation-blind observer]
\label{def:perm-blind}
A structural observer $O$ is \emph{permutation-blind} if
$O(x)=O(y)$ whenever $\mathrm{prof}(x)=\mathrm{prof}(y)$: in symbols,
$O\preceq\Oprof$. This is the analogue of
Definition~\ref{def:swap} for a different symmetry: where swap-blindness
hides one fixed transposition of the two symbols $a,b$, permutation-blindness
hides the entire symmetric group acting on the positions of the input.
The two are not nested --- they are blindness to different group actions
on different objects (the two Skolem symbols in one case, the input
positions in the other) --- but they are instances of one scheme, as
Remark~\ref{rem:opacity-observer-general} below records.
\end{definition}

With this notion in hand, Corollary~\ref{cor:oh-2-8} follows from the
same argument as Corollary~\ref{cor:total-opacity} --- this time
with the observer's blindness spread across every permutation of the
input positions, rather than the one $a\leftrightarrow b$ swap that
Corollary~\ref{cor:total-opacity} is stated for. The argument does not
depend on which of the two symmetries is in play, only on the observer
being blind to whatever the hidden invariant is; Proposition~\ref{prop:order-blind-static}
carries it out explicitly for the permutation case, and
Remark~\ref{rem:opacity-observer-general} states the shared form.

\begin{proposition}[Theorem~\ref{thm:order-blind-char}'s negative
direction as the static, fully-permuted case]
\label{prop:order-blind-static}
Let $L\notin\{\emptyset,\Sigma^*\}$ be not permutation-closed. Then
Corollary~\ref{cor:oh-2-8} follows at $n=0$ from the same argument that
proves Corollary~\ref{cor:total-opacity}, with $O$ permutation-blind
(Definition~\ref{def:perm-blind}) in place of the swap-blind observer
that corollary is stated for.
\end{proposition}

\begin{proof}
Let $R_0$ stand for the order-blind automaton's classification map,
extended formally by $\upsilon_{\mathrm{id}}$.

For condition~(i): classification depends on $x$ only through
$\mathrm{prof}(x)$, so no operation ever fires on the specific
symbol-order information that distinguishes two permutations of the
same profile. This is exactly condition~(i), under the full
symmetric-group symmetry of Definition~\ref{def:perm-blind}.

For condition~(ii): as in the proof of
Corollary~\ref{cor:static-as-degenerate}, $\upsilon_{\mathrm{id}}$'s
independence from history gives condition~(ii) immediately.

With both conditions verified, Corollary~\ref{cor:total-opacity}
(taking $O$ permutation-blind) applies at $n=0$, and its conclusion is
exactly the negative direction of
Corollary~\ref{cor:oh-2-8}/Theorem~\ref{thm:order-blind-char}.
\end{proof}

\begin{remarknn}[Calibrating the generalization]
\label{rem:order-blind-not-new}
As in Section~\ref{sec:cell-d}, this identification calibrates the
generalization before it is put to work below; the order-blind
characterization itself is a known result.
\end{remarknn}

\begin{remarknn}[The general form behind
Corollary~\ref{cor:total-opacity}]
\label{rem:opacity-observer-general}
Corollary~\ref{cor:total-opacity} is stated for a swap-blind observer,
but its proof uses only one property of that observer, and a more
general statement holds with the same proof. Let $G$ be any group
acting on histories, let $O$ be an observer with $O(H)=O(g\cdot H)$ for
every $g\in G$ (blindness to the whole $G$-action), and suppose the
hidden invariant --- the fact the derivation must not expose --- is
exactly a $G$-invariant. Then any update of the form
$R_{n+1}=\upsilon(R_n,O(H_n))$ satisfies
Definition~\ref{def:opacity-preserving}(ii) automatically: since
$\upsilon$ sees the history only through $O$, and $O$ cannot separate
$H_n$ from any $g\cdot H_n$, the update returns the same rule set on
all of them, so it cannot leak which $G$-orbit representative the
history was. The swap-blind case is $G=\mathbb Z_2$ acting by the
$a\leftrightarrow b$ exchange (Corollary~\ref{cor:total-opacity}); the
permutation-blind case is $G=S_{|x|}$ acting on input positions
(Proposition~\ref{prop:order-blind-static}). Both are the same fact:
condition~(ii) is free whenever the update's only view of history
factors through blindness to the group whose invariance is what must
stay hidden. We keep Corollary~\ref{cor:total-opacity} in its swap-blind
form, since that is the case the dynamic sections use, and record the
general version here rather than as a separate theorem, nothing below
depending on more than the two named instances.
\end{remarknn}

\subsection{The dynamic extension: streaming under a permutation-blind
update}

The question asked twice already --- does the static guarantee survive
if the observer itself evolves? --- applies here too. Now the observer
changes not across sessions of one cipher, but across a stream of many
different inputs.

\begin{definition}[Streaming profile system]
\label{def:streaming}
Consider a sequence of inputs $x_1,x_2,\dots$, each read by its own
observer $O_1,O_2,\dots$, where the observer itself evolves as
$O_{t+1}=\upsilon(O_t,H_t)$ for $H_t=(O_1(x_1),\dots,O_t(x_t))$. This
matches Definition~\ref{def:dynamic-system}, with $R_t$ playing the
role of $O_t$, and $D_t$ playing the role of the $t$-th
classification.
\end{definition}

\begin{theorem}[Dynamic order-blindness]
\label{thm:dynamic-order-blind}
Let $(O_t)_{t\ge1}$ be a streaming profile system satisfying two
conditions. First, every $O_t$ is permutation-blind. Second,
$\upsilon(O_t,H_t)$ depends on $H_t$ only through a permutation-blind
function of $H_t$. Then, for every $t$ and every
$L_t\notin\{\emptyset,\Sigma^*\}$ that is not permutation-closed, no
machine receiving $O_1(x_1),\dots,O_t(x_t)$ can correctly decide
$x_t\in L_t$ on both members of some $\mathrm{prof}(x_t)$-equivalent
pair --- regardless of computational power, and regardless of how the
observer sequence evolved.
\end{theorem}

\begin{proof}
This is Theorem~\ref{thm:dynamic-sip}, transported through
Corollary~\ref{cor:total-opacity}, with the single transposition
$a\leftrightarrow b$ generalized to the full symmetric group, exactly
as Definition~\ref{def:perm-blind} generalizes Definition~\ref{def:swap}.

By Proposition~\ref{prop:order-blind-static}, applied at the very first
input $x_1$ with $O_1$ in place of the single fixed observer there,
$P$ holds right after the first step; this is the base case.
Condition~(i) holds because every $O_t\preceq\Oprof$: no observer in
the sequence ever exposes order. Condition~(ii) holds because
$\upsilon$'s choice depends on history only through permutation-blind
data, so two runs that agree on every profile but differ in actual
order produce identical observer sequences at every step.

Induction, exactly as in the proof of Theorem~\ref{thm:dynamic-sip},
now gives the claim at every $t$, applying the negative direction of
Theorem~\ref{thm:order-blind-char} at each step.
\end{proof}

\begin{remarknn}[What this adds]
\label{rem:order-blind-what-new}
The time index and the evolving observer are new: \cite{Buono2026oh}
has neither. The extension holds for every permutation-blind schedule;
constructing a concrete, useful, non-trivial schedule beyond the fixed
$\Oprof$ remains open.
\end{remarknn}

\begin{remarknn}[The pattern across all four recovered instances]
\label{rem:general-pattern}
Sections~\ref{sec:cell-d}--\ref{sec:order-blind} have now shown the
same architecture four times, at three levels of observer richness ---
these four being the \emph{recovered} instances, results already in the
source papers re-derived here as one pattern, distinct from the four
\emph{independently arising} instances taken up later
(Sections~\ref{sec:orbital}, \ref{sec:andromeda}, \ref{sec:toe},
\ref{sec:omega}).
The simplest, $\Obot$-level blindness appears twice, in two unrelated
domains: which of two Skolem constants was used, in the static SIP
itself (Section~\ref{sec:dynamic}), and which plaintext was encrypted
(Cell~(d)). These two are parallel appearances of the same simplest
case in different settings, at the same level of the observer order,
neither ordered relative to the other. Richness then increases
genuinely twice: once inside cryptography, from $\Obot$ to
$(B,\,\cdot\,)$ on the plaintext (Theorem~8.15), and once by
generalizing from a single transposition to the full symmetric group
(the order-blind automaton's $\Oprof$-blindness). Each of the four
instances is recovered as the degenerate, single-step case of
Theorem~\ref{thm:dynamic-sip}, and each yields a new dynamic extension
via Corollary~\ref{cor:total-opacity}, save where one dynamic theorem
already covers two static instances at once
(Remark~\ref{rem:815-dynamic-extension}). That the same theorem, at one
fixed level of generality, recovers all four independent source
instances and extends each is evidence it sits at the right level for
the phenomenon common to all four.
\end{remarknn}

\section{The silent assumption behind Dynamic SIP: semantic
interpretation frames}
\label{sec:semantic-frame}

\subsection{The same move, three times}

The static Syntactic Invariance Principle (Lemma~\ref{lem:sip-static})
silently assumes a single, fixed rewriting system $R$; making that
assumption explicit and asking what happens when $R$ is allowed to
change produced Theorem~\ref{thm:dynamic-sip}, the Dynamic SIP. This
move --- take a result held finished, find the assumption it never
stated, and turn that assumption into a variable --- is the recurring
engine of this work, and one of its own source papers runs on the same
engine. \cite[Abstract]{Buono2026ow} observes that all five of
Impagliazzo's worlds~\cite{Impagliazzo1995} assume every party, adversary included, observes
the complete input --- an assumption so natural that, in that source's
own words, ``it is never stated''; that paper's entire contribution is
to make it explicit and relax it. The same sentence structure, applied
to a different silent assumption (full observation, rather than a fixed
rule set), by the same author, independently motivated an entire second
paper.

Theorem~\ref{thm:dynamic-sip} makes a silent assumption of its own. It
is stated entirely at the level of \emph{derivability} --- what a
syntactic system can or cannot produce --- and says nothing about how a
reader, or a machine, \emph{interprets} what is derived. Implicitly, it
assumes whoever reads the derivation reads it against one fixed,
intended model, the same one throughout. This section makes that
assumption explicit, following the same pattern a third time: state the
theorem, find what it silently fixed, and let that vary. Once the
interpreting model is not fixed, a phenomenon appears that no result so
far in this paper, and none in the four source papers, addresses.

\begin{remarknn}[A first, concrete encounter with the frame idea]
\label{rem:msdos-analogy}
An executable file for MS-DOS, read by a modern Linux kernel, is not
corrupted: every byte is exactly what it was written to be. What
differs is what the reader takes those bytes to \emph{mean}. Nothing in
the file signals which reading is intended, and a reader committed to
one interpretation produces a definite, internally consistent result
whether or not that interpretation is the one the bytes were written
under. The reader's choice of interpretation is doing work the bytes
themselves leave open --- a role played, later in this paper, by an
explicit \emph{frame} chosen independently of the syntax it reads. This
section makes the phenomenon precise for the syntactic systems studied
here, where it is provable rather than merely evocative; the concrete
picture is worth keeping in mind, because the same structure recurs in
richer form.
\end{remarknn}

\subsection{Semantic frames}

A semantic frame is simply a concrete mathematical structure that
gives a specific meaning to the symbols $0,s,+$ and to the constants
$a,b$ --- one particular way of reading the syntax as being about
actual objects.

\begin{definition}[Semantic frame]
\label{def:semantic-frame}
A \emph{semantic frame} for $\mathcal L=\{0,s,+\}$ is a structure $F$
with domain $D_F$, interpreting $0,s,+$ as $0^F,s^F,+^F$, together with
a designated interpretation $a^F,b^F\in D_F$ of the Skolem constants
$a,b$. $F$ \emph{models} a rewriting system $R$ if every rule of $R$
holds in $F$ under this interpretation (e.g.\ $F$ models $\{A1,A2\}$ if
$0^F+^Fx=x$ and $s^F(x)+^Fy=s^F(x+^Fy)$ for all $x,y\in D_F$).
\end{definition}

\begin{remarknn}[Semantics and syntax behave oppositely here]
\label{rem:no-silence}
Lemma~\ref{lem:frozen} shows the syntax of $\{A1,A2\}$ never derives
either $a+b=b+a$ or its negation: the syntactic system is permanently
silent on the question. A semantic frame $F$ cannot be silent this way.
Since $F$ is a specific structure, $a^F+^Fb^F$ and $b^F+^Fa^F$ are
specific, definite elements of $D_F$, and either they are equal or they
are not --- $F$ always has an answer, automatically, whether or not
anyone asks. This is the precise sense in which the two behave
oppositely: syntax withholds an answer; any given semantic frame
supplies one, unasked, and with no mechanism for flagging that the
answer it supplies might not be the only one consistent with the syntax
it reads.
\end{remarknn}

\subsection{Two frames, the same syntax, opposite answers}

Here is the phenomenon this section is built around, stated as starkly
as possible before any formalism: the same file, byte for byte
unchanged, can mean two different and mutually exclusive things to two
readers, with nothing in the file able to settle which reading is
right. This is not merely possible but forced --- for the simplest
arithmetic system there is.

\begin{theorem}[Semantic underdetermination]
\label{thm:semantic-underdetermination}
There exist semantic frames $F_1,F_2$, both modelling $\{A1,A2\}$, with
$F_1\models a+b=b+a$ and $F_2\models a+b\neq b+a$.
\end{theorem}

\begin{proof}
\emph{$F_1$: the standard model.} Let $D_{F_1}=\mathbb N$ with the
usual $0,s,+$, and $a^{F_1}=3$, $b^{F_1}=5$. Standard addition
satisfies $\{A1,A2\}$ trivially, and $3+5=5+3=8$, so
$F_1\models a+b=b+a$.

\emph{$F_2$: an explicit non-standard model.} Let
$D_{F_2}=\mathbb N\cup\{a^*,b^*\}$ for two fresh elements
$a^*\neq b^*$, $a^*,b^*\notin\mathbb N$. The idea, stated before the
mechanics: $a^*$ and $b^*$ act as two different ``sinks'' that swallow
whatever they are added to, and whichever one appears as the
\emph{left} argument wins --- this single asymmetry is what breaks
commutativity. Define $s^{F_2}$ to agree with
the usual successor on $\mathbb N$, and $s^{F_2}(a^*)=a^*$,
$s^{F_2}(b^*)=b^*$ (fixed points outside the standard part). Define
$+^{F_2}$ by cases:
\begin{itemize}
  \item $x,y\in\mathbb N$: the usual sum.
  \item $x=a^*$ (any $y\in D_{F_2}$): $a^*+y:=a^*$.
  \item $x=b^*$ (any $y\in D_{F_2}$): $b^*+y:=b^*$.
  \item $x\in\mathbb N$, $y=a^*$: $x+a^*:=a^*$.
  \item $x\in\mathbb N$, $y=b^*$: $x+b^*:=b^*$.
\end{itemize}
(These five cases cover every pair $(x,y)\in D_{F_2}\times D_{F_2}$
exactly once.)

\emph{$F_2$ satisfies A1 ($0+x=x$).} For $x\in\mathbb N$: standard. For
$x=a^*$: $0+a^*=a^*$ by the fourth case ($0\in\mathbb N$). For $x=b^*$:
$0+b^*=b^*$ by the fifth case.

\emph{$F_2$ satisfies A2 ($s(x)+y=s(x+y)$), checked exhaustively over
all five shapes of $(x,y)$:}
\begin{itemize}
  \item $x,y\in\mathbb N$: both sides equal the standard $s(x+y)$.
  \item $x\in\mathbb N,y=a^*$: $s(x)\in\mathbb N$, so
    $s(x)+a^*=a^*$ (fourth case); and $s(x+a^*)=s(a^*)=a^*$ (fourth
    case, then fixed point). Equal.
  \item $x\in\mathbb N,y=b^*$: symmetric, both sides $=b^*$.
  \item $x=a^*$ (any $y$): $s(a^*)+y=a^*+y=a^*$ (fixed point, then
    second case); and $s(a^*+y)=s(a^*)=a^*$ (second case, then fixed
    point). Equal.
  \item $x=b^*$ (any $y$): symmetric, both sides $=b^*$.
\end{itemize}
So $F_2$ models $\{A1,A2\}$.

\emph{$F_2$ disagrees on the order.} Set $a^{F_2}=a^*$, $b^{F_2}=b^*$.
Then $a^{F_2}+b^{F_2}=a^*+b^*=a^*$ (second case), while
$b^{F_2}+a^{F_2}=b^*+a^*=b^*$ (third case). Since $a^*\neq b^*$ by
construction, $F_2\models a+b\neq b+a$.
\end{proof}

\begin{corollary}[The MS-DOS phenomenon, made precise]
\label{cor:msdos-precise}
Fix any derivation under $\{A1,A2\}$ with Skolem constants $a,b$ (which,
by Lemma~\ref{lem:frozen}, never touches $a+b$ or $b+a$ at their root).
This single syntactic object is compatible with both $F_1$ and $F_2$ of
Theorem~\ref{thm:semantic-underdetermination}: nothing in the
derivation determines, or even hints, which of the two mutually
exclusive readings is intended. A reader (or machine) fixed on $F_1$
will treat ``$a+b=b+a$'' as available and true; a reader fixed on $F_2$
will treat it as false; neither receives any signal that the other
reading exists, is possible, or is equally consistent with everything
the syntax says.
\end{corollary}

\subsection{A worked example of use: ciphertexts without redundancy}
\label{sec:cipher-example}

The arithmetic instance above is deliberately minimal, to make the
phenomenon provable in full rather than merely evocative. The same
structure recurs, with no further construction needed, in an entirely
different and immediately recognisable setting.

\begin{proposition}[Cryptographic instance of semantic underdetermination]
\label{prop:cipher-instance}
Let $\Sigma=\{0,1\}$, plaintext space $\mathbb M=\Sigma^n$ (every
$n$-bit string a valid plaintext: no header, no reserved bits, no
checksum), key space $\mathcal K=\Sigma^n$, and
$\Enc(M,K)=M\oplus K=\Dec(M,K)$ (one-time pad). Fix any ciphertext
$C\in\Sigma^n$. For any two distinct keys $K_1\neq K_2\in\mathcal K$,
let $M_1:=\Dec(C,K_1)=C\oplus K_1$ and $M_2:=\Dec(C,K_2)=C\oplus K_2$.
Then:
\begin{enumerate}[label=(\roman*)]
  \item $M_1\neq M_2$ (since $M_1\oplus M_2=K_1\oplus K_2\neq 0^n$);
  \item $M_1,M_2\in\mathbb M$: both are, without qualification, valid
    plaintexts, since $\mathbb M=\Sigma^n$ has no internal structure a
    decryption could violate;
  \item $C$ alone does not determine which of $K_1,K_2$ (hence which of
    $M_1,M_2$) is correct: every one of the $2^n$ pairs $(K,M)$ with
    $M\oplus K=C$ is equally consistent with $C$, and nothing in $C$
    privileges one over another.
\end{enumerate}
\end{proposition}

\begin{proof}
(i) $M_1\oplus M_2=(C\oplus K_1)\oplus(C\oplus K_2)=K_1\oplus K_2$,
nonzero since $K_1\neq K_2$. (ii) Immediate, since $\mathbb M=\Sigma^n$
places no constraint beyond length. (iii) For every $K\in\mathcal K$,
setting $M:=C\oplus K$ gives $\Enc(M,K)=M\oplus K=C$: every key yields
\emph{some} plaintext consistent with $C$, and by (ii) every such
plaintext is well-formed, so no candidate is excludable on internal
grounds.
\end{proof}

\begin{remarknn}[Why redundancy checks exist]
\label{rem:why-macs}
This is the structural reason redundancy checks (headers, MACs,
checksums) exist in real cryptographic practice: every key here gives
an equally valid reading, with nothing to exclude any of them on
internal grounds. A MAC or checksum is exactly an attempt to break this
underdetermination by adding a rule --- analogous to adding a third
axiom to $\{A1,A2\}$ --- that only the intended key's decryption
satisfies.
\end{remarknn}

\subsection{Two further worked examples}
\label{sec:more-examples}

\begin{proposition}[The MS-DOS instance, made precise]
\label{prop:msdos-precise}
Let $b\in\{0,1\}^N$ be a byte string with no distinguishing header ---
as in the classical \texttt{.COM} executable format, where the entire
string is executed directly as machine code from a fixed offset, with
no reserved signature bytes (unlike \texttt{.EXE} or ELF, both of which
\emph{do} carry a magic header precisely to avoid the phenomenon below).
Let $L_1,L_2$ be two interpreters (loaders), each a partial function
from byte strings to execution behaviour in some domain $\mathcal B$,
corresponding to two incompatible instruction-set conventions. Then:
$b$ is identical regardless of which loader receives it; neither $L_1$
nor $L_2$ receives, from $b$ alone, any signal indicating which
convention was intended; and $L_1(b),L_2(b)\in\mathcal B$ can be
arbitrary and unrelated to one another (including one being
well-defined, meaningful behaviour and the other being undefined,
degenerate, or simply different well-defined behaviour), with $b$ itself
giving no indication that a mismatch has occurred.
\end{proposition}

\begin{remarknn}[Level of formality]
\label{rem:msdos-formality}
This proposition is stated at the level of generality the point needs,
and no further: it does not fix a specific $\mathcal B$ or specific
$L_1,L_2$ with a full instruction-set semantics, since the claim
concerns any header-free byte string under any two incompatible
conventions, and formalising two complete ISAs is orthogonal to it. The
opening picture (Remark~\ref{rem:msdos-analogy}) is now a claim about
arbitrary header-free byte strings and interpreters, rather than an
anecdote.
\end{remarknn}

\begin{proposition}[Transformation-cascade cipher]
\label{prop:cascade-cipher}
Let $\Enc_0$ be a cipher such that, for a uniformly random key, $\Enc_0(M)$
is computationally indistinguishable from a uniformly random string of
$\Sigma^n$ (a standard assumption on $\Enc_0$, e.g.\ a strong
pseudorandom permutation; we do not prove this property, only assume
it). Fix a public toolkit $\mathcal T$ of classical transformations
(permutations, substitutions) on $\Sigma^n$, each a bijection. Given
plaintext $M$, compute $C_0=\Enc_0(M)$, then apply a secret sequence of
$k\ge 0$ transformations $t_1,\dots,t_k\in\mathcal T$ (secret both in
number $k$ and in which elements of $\mathcal T$, and in what order),
obtaining $C=t_k(\cdots t_1(C_0)\cdots)$, retained with no header or
checksum. Then, for any candidate $k'\ge 0$ and any candidate sequence
$t_1',\dots,t_{k'}'\in\mathcal T$, the candidate recovery
$\widehat C_0:=(t_1')^{-1}(\cdots(t_{k'}')^{-1}(C)\cdots)$ is a
well-defined element of $\Sigma^n$ (since every $t\in\mathcal T$ is a
bijection), and is, under the assumption on $\Enc_0$, computationally
indistinguishable from $C_0$ itself regardless of whether the candidate
sequence is the correct one: nothing in $C$ certifies which candidate
recovery, among the combinatorially many choices of $(k',t_1',\dots,t_{k'}')$,
is correct.
\end{proposition}

\begin{remarknn}[Two phenomena compounded: underdetermination and
search]
\label{rem:cascade-caveat}
This construction illustrates the interaction of two phenomena the rest
of the paper keeps separate. The underdetermination of
Section~\ref{sec:semantic-frame} --- no internal check distinguishes a
correct recovery from an incorrect one, exactly as in
Proposition~\ref{prop:cipher-instance} --- compounds with a
combinatorially large search space (candidates range over all finite
sequences from $\mathcal T$) to produce a search problem with no
verifiable stopping condition short of external information: not merely
hard to search, as in the Base Recovery Problem
of~\cite{Buono2026sep}, but, absent external information, not verifiably
searchable at all. Two matters bound the scope of this. The
noise-likeness of $\Enc_0$ is a computational assumption, not the
information-theoretic guarantee of the one-time pad in
Proposition~\ref{prop:cipher-instance}; and the classical
transformations, all bijections on a fixed-size alphabet, add no
entropy of their own. The underdetermination half is this paper's
subject; the resulting search-complexity question is a separate matter,
recorded here but not developed.
\end{remarknn}

\subsection{The blind-cascade cipher: security by removal of the
verification predicate}
\label{sec:blind-cascade}

Proposition~\ref{prop:cascade-cipher} is worth developing on its own,
because when the secret sequence is itself keyed and its length is a
further secret, the construction exhibits a security property that
deserves to be stated and proved carefully --- not because the scheme
is one to use in practice, but because it exposes, by carrying it to
its limit, the hidden assumption behind almost all of cryptography:
that there is some way to recognise a correct decryption. We call the
scheme the \emph{blind-cascade cipher}, there being, as far as we know,
no established term in the literature for a construction of this shape.

\subsubsection*{The construction}

The plaintext $M$ is carried to the ciphertext $C$ in layers. First, a
\emph{base encryption}: one chooses a strong standard cipher $\Enc_0$
from among $n$ candidates --- the choice itself secret --- and a key
$K_1$, and computes $C_0=\Enc_0(M,K_1)$. Under the standard assumption
on $\Enc_0$, $C_0$ is indistinguishable from a uniform string: already
here, to anyone without $K_1$, $C_0$ is noise. Next, a \emph{confusion
cascade}: fix a public set $\mathcal T=\{\tau_1,\dots,\tau_m\}$ of
classical reversible bijections on strings --- permutations,
substitutions, rotations, row and column swaps, and any other confusion
operation --- and apply to $C_0$ a sequence of $L$ layers of operations
drawn from $\mathcal T$. The number of layers $L$ is a third secret
parameter, given as input, and used to derive a second key $K_2$. The
operations of layer $i$ --- which, how many, in what order, with what
parametrising bits --- are determined by
$\kappa_i=\mathrm{KDF}(K_1,K_2,i)$. No fresh random bit enters after
$K_1$: the whole system is a deterministic expansion of the seed
$s=(K_1,K_2,L)$. Write $C=\Phi_s(M)$ for the entire transformation, and
$\Phi_s^{-1}$ for its inverse, computable by whoever holds $s$. The
ciphertext $C$ carries no header, no checksum, no verification oracle of
any kind: nothing in $C$ records how many layers, which operations, or
which base cipher. The legitimate recipient, holding $s$ and the choice
of cipher, reconstructs the $\kappa_i$, inverts the layers in order, and
decrypts $C_0$.

\subsubsection*{The attacker's position}

The most direct way to see why the scheme is strong is to try to break
it. The attacker has $C$ and nothing else: not $K_1$, not $K_2$, not
$L$, not which base cipher among $n$; no header, checksum, or oracle.
Suppose the attacker tries to decrypt. The confusion layers must be
unwound first, but their number and identity are unknown. Applying the
inverse $\tau_j^{-1}$ of some bijection yields a string --- is it closer
to the solution? There is no way to tell: the string is noise, as $C$
was, as anything the attacker produces will be. Another operation:
again noise. There is no gradient to follow, no signal reading ``warmer,
colder.'' Every move must be tried as though it were the right one,
because nothing distinguishes it from a wrong one. Suppose fortune
intervenes and the attacker unwinds every layer correctly, arriving at
$C_0$. They would not know it: $C_0$ is the output of a strong cipher,
as indistinguishable from uniform as any failed attempt. They hold the
intermediate solution and cannot recognise it. And if they went further,
guessed $K_1$ and the cipher and obtained $M$ --- again, with no
expected plaintext to compare against, they would not know they had
finished.

We now make this experience precise, in two steps: first \emph{how
many} attempts the attacker faces, then --- the crux --- why \emph{no}
attempt yields any information.

\begin{proposition}[Relative unboundedness of the attack space]
\label{prop:blind-cascade-unbounded}
Let $\mathcal T$ be the public, non-empty set of confusion bijections,
each with a parametrisation of size at least $2$, and let $C=\Phi_s(M)$.
An adversary without the seed $s$ and without a verification oracle has
no effective upper bound $L^\star$ on the number of layers: for every
$\ell$ there is a valid seed $s'$ with exactly $\ell$ layers such that
$\Phi_{s'}^{-1}(C)$ is a well-formed plaintext. The set of operation
sequences the adversary must consider therefore has cardinality at
least $\sum_{\ell\ge 1}|\mathcal T|^{\ell}$, which diverges whenever
$|\mathcal T|\ge 1$.
\end{proposition}

\begin{proof}
Since every $\tau\in\mathcal T$ is a bijection, every finite composition
of operations from $\mathcal T$ is a bijection, so for each length
$\ell$ and each choice of operations the resulting map sends $C$ to
\emph{some} string, which is a well-formed plaintext (every string is,
absent an imposed format). There is thus no $\ell$ beyond which the
attempts are empty: the space to be traversed is not bounded above by
any $L^\star$ derivable from $C$. Since $\mathcal T$ is non-empty, the
number of sequences of each finite length is at least one and the total
diverges over unbounded $\ell$.
\end{proof}

The count alone does not ground the security --- an unbounded space
whose elements carried a signal of correctness would be traversed by a
guided search, eliminating candidates as they failed. The decisive
point, established next, is that no such signal exists: no candidate is
eliminable, because eliminating one would require verifying its outcome,
and no verification is available. It is this, not the size of the space,
that does the work.

\begin{proposition}[Indistinguishability of outcomes: no verification
predicate]
\label{prop:blind-cascade-nopredicate}
Under the assumption that $\Enc_0$ produces output indistinguishable
from uniform, for every pair of operation sequences $u\neq u'$
applicable to $C$, the outcomes $u^{-1}(C)$ and $u'^{-1}(C)$ have the
same distribution in the eyes of any observer not holding the seed $s$.
Consequently there is no function $\mathsf{Verify}(\cdot)$ that, given a
candidate decryption of $C$ and without access to $s$, returns
``correct / incorrect'' with success probability above chance.
\end{proposition}

\begin{proof}
Each operation of $\mathcal T$ is a public bijection; applying or
inverting it is a one-to-one transformation that neither adds nor
removes information about the seed. The only source of structure that
would distinguish the correct outcome is the plaintext $M$; but that
structure is reachable only after all layers have been inverted with the
right parameters and $\Enc_0$ decrypted with $K_1$ --- that is, only
with $s$. Without $s$, the outcome of every sequence is a bijective
function of $C$, and $C$ is by assumption indistinguishable from
uniform; hence all outcomes are, for the observer without the key,
samples from the same distribution. A $\mathsf{Verify}$ effective
against chance would constitute a test distinguishing the correct
outcome from the others --- that is, a test distinguishing
$C_0=\Enc_0(M,K_1)$ from uniform, against the assumption on $\Enc_0$.
\end{proof}

It is this second proposition that carries the security to the level
required. The first says the paths are unbounded; the second says
\emph{no path is distinguishable from another} --- and it is the second
that matters, because it is what makes the search not merely long but
\emph{without a stopping criterion}. The structure that would seem to
signal success can arise from an incorrect decryption as readily as from
the correct one --- confusion bijections applied to wrong bits produce
apparent structure continually --- so that the apparent signal is
indistinguishable from the noise that imitates it.

\subsubsection*{The point of view, and a concept to keep in mind}

What we obtain, which is like what one obtains from Shannon's one-time
pad but from a different point of view, is: \emph{an output
indistinguishable from noise for anyone without a key.} A concept we
shall return to, and that it is worth beginning to keep in mind.

The point of view is different because the way the information
identifying the correct reading is made inaccessible is opposite: in the
one-time pad that information is absent because the key is uniform and
as long as the message; here it is absent because every verification has
been removed, at every layer. In both cases the result, for the observer
without the key, is the same --- an output fully determined and at once
indistinguishable from true noise.

The strength is thus one of \emph{kind}, not of \emph{degree}: not ``the
attacker must do a great deal of work,'' but ``the object --- the
verifier --- that would turn work into an answer does not exist.'' This
is why the entropy of the seed, finite as it is, is not the right
measure: it would count the cost of a search, but a search is defined
only where the solution is recognisable, and here it is not.

\subsubsection*{Reflexes to set aside}

The scheme is simple but far from the ordinary use of cryptography, and
one examining it slips almost inevitably into a sequence of reflexes.
Each is correct in its usual context and misleading here.

\emph{``Security is the entropy of the key.''} The instinct computes
$|K_1|+|L|+\log_2 n$ and concludes: finite security. But entropy
measures the cost of a search, and a search is well defined only if the
solution is recognisable. By
Proposition~\ref{prop:blind-cascade-nopredicate} it is not here. Counting
the seed's bits --- correct in itself --- measures a procedure that
cannot conclude.

\emph{``The attacker enumerates keys and verifies.''} Verifies with
what? Proposition~\ref{prop:blind-cascade-nopredicate} says that
$\mathsf{Verify}$ without $s$ does not exist. Enumeration remains
possible but does not conclude.

\emph{``When structure emerges, I have decrypted.''} The structure can
emerge \emph{false}, from chance; it certifies nothing. It is the
systematic decoy of the proof above.

\emph{``The space of cascades is finite, hence enumerable.''} Finite
\emph{for whoever holds the seed} (the recipient); unbounded and
undifferentiated \emph{for whoever does not}
(Proposition~\ref{prop:blind-cascade-unbounded}). Finiteness is relative
to the observer --- as is every notion in this paper.

\emph{``If it is secure it must be usable, so there is a flaw.''}
Usability is not a requirement: this is a theoretical object, studied
for its mathematical properties. Unusability is not a flaw but the
direct manifestation of the property.

\subsubsection*{A note on robustness}

A scheme with no internal verification is fragile against transmission
errors. This is resolved without touching the property, by enclosing the
blob in an external container with a checksum --- a compressed archive,
say. The checksum verifies the \emph{channel}, not the \emph{reading}:
it confirms that the bits arrived intact, not which of their decryptions
is correct. Verification of transport integrity sits outside the blob
and reinstates no oracle on the content.

\subsubsection*{Why it has not been isolated before}

That a property this sharp has not been isolated is not an oversight of
the field: it is that the object lies in the exact blind spot of the
five reflexes. The standard cryptographic eye measures entropy,
presupposes a verifier, trusts emerging structure, counts a finite
space, assumes usability. An object that becomes more secure precisely
where it becomes unusable, and whose security grows by \emph{subtraction}
of verification rather than by \emph{addition} of work, falls outside
all five.

\begin{remarknn}[A further silent assumption removed: the verification
oracle]
\label{rem:verification-oracle}
Most discussions of ciphertext recovery, including standard
security-game formalizations, silently assume some means of
recognising a correct decryption once found --- a header or checksum,
or, in the known-plaintext setting, the adversary's prior knowledge of
what the plaintext should look like. The construction of
Proposition~\ref{prop:cascade-cipher} removes this assumption too, not
only the header: with no checksum and no known plaintext to compare
against, nothing --- internal to the ciphertext or external to it ---
certifies that any candidate recovery $\widehat C_0$ is the right one,
even to someone who has correctly guessed the transformation sequence.
This is one more instance of the pattern the paper traces throughout:
an assumption so natural it usually goes unstated. Making it explicit
is what lets the construction reach the same \emph{structural} property
Theorem~\ref{thm:semantic-underdetermination} isolates --- no check, of
any kind, available to any party, distinguishes a correct reading from
an incorrect one --- rather than the information-theoretic guarantee of
the one-time pad, which Remark~\ref{rem:cascade-caveat} already
separates off. A system built this way has essentially no practical
use, since the \emph{legitimate} recipient equally has no way to
confirm success; this is exactly why real systems reinstate a
verification oracle deliberately. The point of recording it is the same
as the rest of this section's: to show precisely which assumption does
the protective work, by exhibiting what remains once it is removed.
\end{remarknn}

\begin{remarknn}[Scope: a fixed frame, read against a static or dynamic
syntax]
\label{rem:semantic-frame-static-only}
Theorem~\ref{thm:semantic-underdetermination} and
Corollary~\ref{cor:msdos-precise} concern a single, fixed semantic
frame $F$, read against a syntactic system that may itself be dynamic
(Section~\ref{sec:dynamic}) or static (Section~\ref{sec:static}). A
frame that itself changes over the course of a derivation --- a
``semantic frame selector'' $\phi$ with $F_{n+1}=\phi(F_n,H_n)$,
parallel in shape to the syntactic update $\upsilon$ of
Definition~\ref{def:opacity-preserving} --- is a different object, not
treated here. An evolving frame alongside an evolving rule set $(R_n)$
is a genuine two-axis system, and the frame constructed here stays on
one axis: $F_1,F_2$ are simply and verifiably true frames, not
assertions engineered to contradict anything. There is also a prior
question of what ``opacity-preserving'' should even mean for a frame
selector $\phi$: a frame update changes what \emph{is true}, not merely
what has been \emph{derived}, and the two are governed by different
logical rules, so the two conditions of
Definition~\ref{def:opacity-preserving} do not transfer to $\phi$ by
analogy.
\end{remarknn}

\begin{remarknn}[Why this is a separate fact from Theorem~\ref{thm:dynamic-sip}]
\label{rem:not-contained-revisited}
A theorem that does not mention a phenomenon is not thereby a theorem
\emph{about} it. Theorem~\ref{thm:dynamic-sip} quantifies over
derivations and says what they can never expose;
Theorem~\ref{thm:semantic-underdetermination} quantifies over models
and says they can, individually and unavoidably, decide what the
derivation does not --- and disagree with each other while doing so.
The first is a fact about syntax, the second a fact about semantics;
neither is a special case of the other.
Corollary~\ref{cor:msdos-precise} exists to state the relationship
between them precisely, rather than leave it as an analogy.
\end{remarknn}

\subsection{Interpreted machines: encoding and frame fixed once, at the
start}
\label{sec:interpreted-machine}

We now package this section's phenomenon as a property of computing
machines directly. The simplest case is a machine that fixes both its
encoding and its interpretation once, before computation begins, and
never revisits either. First we need one more piece: a precise way to
say exactly which facts a given encoding does, and does not, settle on
its own.

\begin{definition}[Explicit encoding]
\label{def:explicit-encoding}
An \emph{explicit encoding} is a pair $E=(R_0,I_{R_0})$ where $R_0$ is
a rewriting system and
\[
I_{R_0}:=\{\,s=t \;:\; s,t\text{ terms}, R_0\vdash s=t\,\}
\]
is the set of equalities \emph{derivable} from $R_0$ (in the sense of
Lemma~\ref{lem:sip-static}: $s=t\in I_{R_0}$ iff some derivation using
the rules of $R_0$ identifies $s$ and $t$). $I_{R_0}$ is not chosen
independently of $R_0$: it is a deterministic consequence of it, fixed
the moment $R_0$ is fixed.
\end{definition}

\begin{proposition}[Soundness of explicit encodings]
\label{prop:encoding-soundness}
Let $E=(R_0,I_{R_0})$ be an explicit encoding and let $F$ be a semantic
frame modelling $R_0$ (Definition~\ref{def:semantic-frame}). Then
$F\models s=t$ for every $s=t\in I_{R_0}$: every frame modelling $R_0$
agrees with everything $R_0$ derives.
\end{proposition}

\begin{proof}
The proof is by induction on the length of the derivation witnessing
$s=t\in I_{R_0}$.
Each step of the derivation applies a rule of $R_0$; since $F$ models
$R_0$, that rule holds in $F$ (Definition~\ref{def:semantic-frame}), so
the step preserves truth in $F$. The base case (the empty derivation,
$s=t$ already syntactically identical) is immediate. Hence the term
identifications $I_{R_0}$ certifies all hold in $F$.
\end{proof}

\begin{remarknn}[The exact boundary Theorem~\ref{thm:semantic-underdetermination}
lives on]
\label{rem:boundary}
Proposition~\ref{prop:encoding-soundness} says every frame modelling
$R_0$ agrees on $I_{R_0}$, and says nothing about facts \emph{outside}
$I_{R_0}$; that silence is its exact content, not a gap. For
$R_0=\{A1,A2\}$, Lemma~\ref{lem:frozen} gives
$a+b=b+a\notin I_{R_0}$ (neither it nor its negation is derivable): the
Skolem pair's order is precisely a fact the explicit encoding does not
settle, which is why $F_1,F_2$ in
Theorem~\ref{thm:semantic-underdetermination} are free to disagree on
it without either violating Proposition~\ref{prop:encoding-soundness}.
An explicit encoding thus does more than list what it exposes: it
draws, provably, the exact line between what every legitimate
interpretation must agree on and what is genuinely open for
interpretation to decide.
\end{remarknn}

The four examples above --- the bare arithmetic derivation of
Section~\ref{sec:total-opacity}, the header-free cipher, the loader
mismatch, and the transformation cascade --- are independent
applications of the same phenomenon
(Theorem~\ref{thm:semantic-underdetermination}), given to show it has
genuine, recognisable instances outside the toy arithmetic language.
This subsection develops something additional: a formal packaging of
the phenomenon as a property of computing machines directly, in the
simplest case, where the machine does not use the dynamic apparatus of
Section~\ref{sec:dynamic} at all. It fixes one explicit encoding and
one interpretation once, before computation begins, and never revisits
either --- the familiar situation of a program compiled once against
one fixed target architecture, with neither the source language nor the
hardware's semantics changing mid-run.

\begin{definition}[Interpreted machine]
\label{def:interpreted-machine}
An \emph{interpreted machine} is a triple $\mathcal I=(T,E,F)$ where
$T$ is a standard Turing machine, $E=(R_0,I_{R_0})$ is an explicit
encoding (Definition~\ref{def:explicit-encoding}) fixed once, before
computation begins, governing the syntax $T$ operates on (in the sense
of Definition~\ref{def:dynamic-system} with the degenerate update
$\upsilon_{\mathrm{id}}$ of Corollary~\ref{cor:static-as-degenerate}:
$R_0$ is chosen at the start and never updated), and $F$ is a semantic
frame (Definition~\ref{def:semantic-frame}) modelling $R_0$. Both $E$
and $F$ are chosen once and never changed during the run.
\end{definition}

\begin{proposition}[Interpreted machines are externally indistinguishable
across frames]
\label{prop:interpreted-machine-indist}
Let $\mathcal I_1=(T,E,F_1)$ and $\mathcal I_2=(T,E,F_2)$ share the
same underlying machine $T$ and the same explicit encoding
$E=(R_0,I_{R_0})$, differing only in the frame: $F_1\neq F_2$, though
both model $R_0$ (take, for instance, the $F_1,F_2$ of
Theorem~\ref{thm:semantic-underdetermination}). Two things are true of
this pair at once.

First, the two machines are computationally identical as syntactic
processes: every step $T$ performs under $R_0$ is the same for
$\mathcal I_1$ and $\mathcal I_2$, since $R_0$ alone governs
derivability. They agree with each other, and with $T$'s own
derivations, on every fact in $I_{R_0}$.

Second, they disagree wherever $R_0$ leaves a fact undetermined. If
either machine's semantics is queried on a fact outside $I_{R_0}$ ---
such as $a+b\stackrel{?}{=}b+a$ --- $\mathcal I_1$ and $\mathcal I_2$
give opposite answers, by construction of $F_1,F_2$.

Put these together: no output or observable computational trace of
either machine, taken alone, reveals which of the two it is. The only
way to tell them apart is if $\mathcal I_i$'s own program is
explicitly written to query $F_i$ and report the answer --- and even
then, this requires $F_i$ to already be given to the machine as input
or oracle; it cannot be inferred from $E$ alone.
\end{proposition}

\begin{proof}
The first claim is Corollary~\ref{cor:static-as-degenerate}: the
behaviour of $T$ under fixed $R_0$ is exactly the standard
Turing-machine case, independent of $F$. Agreement on $I_{R_0}$ then
follows from Proposition~\ref{prop:encoding-soundness}, applied once
to $F_1$ and once to $F_2$.

The second claim is Theorem~\ref{thm:semantic-underdetermination}:
$F_1,F_2$ model the same $R_0$ yet give opposite answers to a fact
that Remark~\ref{rem:boundary} already showed lies outside $I_{R_0}$.
\end{proof}

\begin{remarknn}[A machine is doubly silent by default]
\label{rem:interpreted-machine-adds}
Corollary~\ref{cor:static-as-degenerate} already shows a standard
Turing machine is the $\upsilon\equiv\upsilon_{\mathrm{id}}$ instance of
Dynamic SIP: encoding fixed, never updated.
Definition~\ref{def:interpreted-machine} through
Proposition~\ref{prop:interpreted-machine-indist} add the orthogonal,
independent observation that fixing the encoding this way --- even an
\emph{explicit} one, which names precisely which facts it does and does
not settle (Definition~\ref{def:explicit-encoding}) --- says nothing
about fixing an interpretation for what it leaves open. A standard
Turing machine is compatible with any frame $F$ modelling its encoding,
agreeing with all of them on $I_{R_0}$ and with none of them
necessarily on anything outside it, and nothing internal to the machine
signals which frame, if any, is intended. A machine is therefore doubly
silent by default: about how its rules may evolve (resolved, when they
do not, by Corollary~\ref{cor:static-as-degenerate}) and about how
whatever its rules leave open is to be read (the subject of this
section). Making the encoding explicit narrows the second silence
without removing it: it tells you exactly where the silence lies, not
how to break it.
\end{remarknn}

\section{A first, independently arising instance: two selectors over
countable spaces}
\label{sec:orbital}

\subsection*{Motivating Analogy}

Before presenting the main results, it is useful to introduce a short
conceptual analogy that highlights the type of semantic and structural
limitations that motivate this work.

\begin{quote}
Once upon a time there was  a machine of incredible power, capable of answering 
any question not by performing calculations, but by extracting the answer directly 
from the very structure of reality itself. but we will understand this later. 
Unfortunately, there is a problem: no one truly knows how to ask the question. 
Even if such a question existed, there would be infinitely many ways to formulate it, 
infinitely many ways to ``select'' the encoding through which the question is posed, 
and infinitely many (perhaps even more) ways to interpret it. Infinite universes 
of meaning in which each specific encoding of the question might admit one answer, 
two answers, infinitely many answers, or none at all.

In some worlds the question would not make sense; in others it would be the
primordial question at the origin of the universe, yet asked to a hamster. In
another world it might be addressed to the only entity capable of answering it,
but that entity would be unable to understand it. Or perhaps the problem lies in
the answer itself, which might be impossible for us to interpret. As with Deep
Thought in Douglas Adams' \textit{The Hitchhiker's Guide to the Galaxy}, the
answer might well be 42, but we would have no idea how to make sense of it.

And there is more. Even if we had one selector for the encoding and another for
the semantic reality in which the question is to be evaluated, we could not even
attempt to test all possibilities. If we combine the values of both selectors
into a list of pairs, Cantor would immediately remind us that there will always
be a new pair which, by diagonalization, was not in the original list.

But even that is not enough, because things get worse. Every combination carries
a different degree of semantic blindness that compromises the result. And even
if we happened to guess the one reality that is completely transparent for a
given encoding, we would obtain an answer which, once produced as output by our
hypothetical machine, would acquire its own points of blindness. We cannot
guarantee that it preserves its semantic value, in fact, statistically, it is
impossible. In a sense, it is a reflection of Rice's theorem, like a famous
separation problem. And the legend has it that this same phenomenon
is the nonexistence at the base of every logical paradox, but this is another story.
\end{quote}

This analogy does not affect the rest of the text or the results and should not
be considered part of the formal scientific contribution. It is simply intended
to provide a useful mental image that facilitates reading.

\begin{remarknn}[This section formalizes the fable, and owes it a real
debt]
\label{rem:formalizing-the-fable}
The story came first. It is not a theory arrived at independently and
then illustrated by the story above; the order is the reverse. The
story already contained, in informal and narrative form, essentially
every structural idea this section makes precise: the two selectors,
the diagonalization over their combinations, the varying degrees of
blindness, and the explicit invocation of Rice's theorem. The debt runs
further still. Two of the story's lines --- an answer received correctly
and still ``no idea how to make sense of'', and every combination
carrying ``a different degree of semantic blindness'' --- reach past
this section entirely. The first is a phenomenon the fable states in its
own right --- a correct answer that no reader can tell from noise ---
of which Chaitin's $\Omega$ (Section~\ref{sec:omega}) is later shown to
be one sharp instance; the phenomenon was in the fable before $\Omega$
had been introduced into this paper at all, and does not depend on it.
The second is the fact, proved
in Remark~\ref{rem:infinite-shades} once Categories~1--3 were in hand,
that those degrees are countably infinite in each direction, not a
manner of speaking. Neither was visible as a precise claim when the
fable was written; both were already sitting, correctly, in it. The
``machine of incredible power'' that answers by reading reality rather
than by computing is owed to Douglas Adams' \textit{The Hitchhiker's
Guide to the Galaxy} and its Deep Thought: that is where the author
first encountered the shape of the problem this section formalizes, and
it is the genuine origin of the intuition, not a decorative reference.
Narrative and metaphor arriving before formal proof is not unusual in
the history of mathematics and science; the correspondence between this
story and what is proved below, checked line by line at the end of the
section (Remark~\ref{rem:mapping-caveats} and the table preceding it),
is close enough to be worth setting out explicitly, once there is
something to compare it against.
\end{remarknn}

\subsection*{From here on, the formal content begins}

The fable asked what happens when a question can be posed in
infinitely many ways, and read back in infinitely many more. We now
build the machine the fable was describing --- starting with its two
knobs: one choosing the encoding the question is posed in, one choosing
the frame it is read back through. This section shows one thing: both
knobs, despite offering infinitely many settings, offer only
\emph{countably} infinitely many. This is worth proving carefully,
because the next section shows a closely related space is \emph{not}
countable at all, and the two must not be confused. The impossibility
results and the connection back to the rest of the paper come later,
once the machine is on the table.

\begin{remarknn}[A name, and where it comes from]
\label{rem:orbital-name}
We call the machine below an \emph{orbital machine}. The name carries
no mathematical content: it is inherited from an earlier, unrelated
piece of code in which the settings of a pair of selectors were
labelled ``orbits'', and which first suggested the construction here as
an intuition pump. Nothing proved or assumed in this section is
inherited from anywhere else, and the name implies no continuity with
any other construction that may have used it; it is kept only because it
carries the original intuition for the author.
\end{remarknn}

\subsection{Two countable spaces}

Picture the machine's two knobs concretely. The first chooses an
\emph{encoding}: which finite set of rewriting rules the question gets
posed in. The second chooses a \emph{frame}: which structure the
answer gets read back through. Both, in principle, offer infinitely
many settings --- but ``infinitely many'' hides a distinction the rest
of the paper depends on, and this subsection nails it down before
anything is built on top of it.

\begin{definition}[Encoding space]
\label{def:encoding-space}
Fix a countable alphabet $\Sigma_0=\{0,s,+,a,b,c_1,c_2,c_3,\dots\}$ (the
base symbols together with countably many auxiliary constants available
to be used as fresh Skolem symbols). The \emph{encoding space}
$\mathcal E_{\mathrm{sp}}$ is the set of all finite rewriting systems
$R_0$ over (a finite subset of) $\Sigma_0$.
\end{definition}

\begin{proposition}[$\mathcal E_{\mathrm{sp}}$ is countably infinite]
\label{prop:encoding-space-countable}
$\mathcal E_{\mathrm{sp}}$ is countably infinite.
\end{proposition}

\begin{proof}
Terms over the countable alphabet $\Sigma_0$ are finite strings (or
finite trees) over a countable set of symbols, hence countably many
(a standard Gödel numbering: enumerate $\Sigma_0=\{\sigma_0,\sigma_1,\dots\}$
and code each finite term as a natural number via its syntax tree).
A finite rewriting system $R_0$ is a finite set of pairs of terms;
finite subsets of a countable set are themselves countably many
(a finite subset of $\mathbb N$ is coded by, e.g., the sum
$\sum_{i\in R_0} 2^i$ under any fixed enumeration of pairs-of-terms by
$\mathbb N$), so $\mathcal E_{\mathrm{sp}}$ is countable. Infinitude:
for each $n\ge 0$, letting
$R_0^{(n)}=\{A1,A2\}\cup\{c_i\to c_i : 1\le i\le n\}$ gives infinitely
many distinct finite rewriting systems, hence infinitely many distinct
encodings.
\end{proof}

\begin{definition}[Computable frame]
\label{def:computable-frame}
A \emph{computable frame} is a structure with domain $D_F=\mathbb N$,
interpreting $0,s,+$ as $0^F=0$ and total computable functions
$s^F:\mathbb N\to\mathbb N$, $+^F:\mathbb N\times\mathbb N\to\mathbb N$,
together with designated elements $a^F,b^F\in\mathbb N$. It
\emph{models} a rewriting system $R_0$ if every rule of $R_0$ holds
under this interpretation. The \emph{frame space} $\mathcal F_{\mathrm{sp}}$
is the set of all computable frames modelling $\{A1,A2\}$.
\end{definition}

\begin{proposition}[$\mathcal F_{\mathrm{sp}}$ is countably infinite]
\label{prop:frame-space-countable}
$\mathcal F_{\mathrm{sp}}$ is countably infinite.
\end{proposition}

\begin{proof}
Fix a standard enumeration $(\varphi_e)_{e\in\mathbb N}$ of all partial
computable functions $\mathbb N\to\mathbb N$ (respectively
$\mathbb N\times\mathbb N\to\mathbb N$, coding pairs via a fixed
computable pairing), via Turing machine indices. A computable frame is
determined by a tuple $(e_s,e_+,a^F,b^F)\in\mathbb N^4$ (the index for
$s^F$, the index for $+^F$, and the two designated elements), subject
to $\varphi_{e_s},\varphi_{e_+}$ being total and to $F$ modelling
$\{A1,A2\}$: $\mathcal F_{\mathrm{sp}}$ is in bijection with a subset of
$\mathbb N^4$, and $\mathbb N^4$ is countable (a finite product of
countable sets), so any subset of it is countable.
Infinitude: for every $a,b\in\mathbb N$ with $a\ne b$, taking $s^F,+^F$
to be the standard successor and addition functions (a single fixed
pair of indices $e_s^*,e_+^*$, which model $\{A1,A2\}$ since ordinary
addition satisfies A1, A2) gives a distinct computable frame; there are
infinitely many such pairs $(a,b)$.
\end{proof}

\begin{remarknn}[Relation to Theorem~\ref{thm:semantic-underdetermination}'s
$F_1,F_2$]
\label{rem:f1f2-computable}
Both $F_1$ and $F_2$ of Theorem~\ref{thm:semantic-underdetermination}
are, after recoding, computable frames in the sense of
Definition~\ref{def:computable-frame}: $F_1$ (the standard model with
$a^{F_1}=3,b^{F_1}=5$) is already one, taking $e_s^*,e_+^*$ to code the
standard successor and addition functions. $F_2$'s domain
$\mathbb N\cup\{a^*,b^*\}$ is countably infinite, hence in bijection
with $\mathbb N$; composing that bijection with $F_2$'s operations
(mapping $a^*,b^*$ to two fixed naturals not otherwise used, e.g.\ via
any computable enumeration of $\mathbb N\cup\{a^*,b^*\}$) yields a
computable frame isomorphic to $F_2$. Both $F_1,F_2\in\mathcal
F_{\mathrm{sp}}$ once so recoded; this confirms that $\mathcal
F_{\mathrm{sp}}$ is rich enough to contain the frames already shown to
disagree, and is used nowhere below.
\end{remarknn}

\subsection{Two selectors, and countably many pairs}

\begin{definition}[Encoding and frame selectors]
\label{def:two-selectors}
An \emph{encoding selector} is a function $\sigma_E$ taking values in
$\mathcal E_{\mathrm{sp}}$; a \emph{frame selector} is a function
$\sigma_F$ taking values in $\mathcal F_{\mathrm{sp}}$. The
\emph{orbital machine} (Remark~\ref{rem:orbital-name}) with selectors
$(\sigma_E,\sigma_F)$ is a standard Turing machine $T$ together with a
value $(E,F)=(\sigma_E,\sigma_F)$: $T$'s syntax is governed by the
rewriting system $E$, and (where a specific fact about $T$'s
computation is to be read off) that fact is evaluated under the frame
$F$. What, if anything, $\sigma_E,\sigma_F$ are permitted to depend on
is not specified in this section and is left open for later work.
\end{definition}

Countability matters twice here, in opposite directions. The theorem
below shows the space of possible \emph{settings} $(E,F)$ is countable
--- an honest, complete list of them exists.
Section~\ref{sec:selector-functions} then shows the space of
\emph{selectors} choosing among those settings is not countable at all.
Pinning down the first fact precisely now is what makes the second a
genuine surprise rather than a confusion about which space is in play.

\begin{theorem}[Countably many pairs]
\label{thm:countable-pairs}
$\mathcal E_{\mathrm{sp}}\times\mathcal F_{\mathrm{sp}}$ is countably
infinite: there are countably infinitely many possible pairs
$(E,F)\in\mathcal E_{\mathrm{sp}}\times\mathcal F_{\mathrm{sp}}$, i.e.\
countably infinitely many possible settings of the pair of selectors
$(\sigma_E,\sigma_F)$ at a single evaluation point.
\end{theorem}

\begin{proof}
By Propositions~\ref{prop:encoding-space-countable} and
\ref{prop:frame-space-countable}, $\mathcal E_{\mathrm{sp}}$ and
$\mathcal F_{\mathrm{sp}}$ are each countably infinite, hence each in
bijection with $\mathbb N$: fix bijections
$\alpha:\mathbb N\to\mathcal E_{\mathrm{sp}}$,
$\beta:\mathbb N\to\mathcal F_{\mathrm{sp}}$. The Cantor pairing
function $\pi:\mathbb N\times\mathbb N\to\mathbb N$,
$\pi(i,j)=\tfrac{1}{2}(i+j)(i+j+1)+j$, is a bijection, obtained by
enumerating $\mathbb N\times\mathbb N$ along successive finite
diagonals $\{(i,j):i+j=k\}$, $k=0,1,2,\dots$ (each diagonal finite, so
every pair is reached at a finite stage). Composing,
$(i,j)\mapsto(\alpha(i),\beta(j))$ is a bijection
$\mathbb N\times\mathbb N\to\mathcal E_{\mathrm{sp}}\times\mathcal
F_{\mathrm{sp}}$, and $\pi^{-1}$ composed with it gives a bijection
$\mathbb N\to\mathcal E_{\mathrm{sp}}\times\mathcal F_{\mathrm{sp}}$
directly. Hence $\mathcal E_{\mathrm{sp}}\times\mathcal F_{\mathrm{sp}}$
is countably infinite (infinite since each factor is infinite and
nonempty; countable as a countable union, indexed by $i\in\mathbb N$,
of the countable sets $\{i\}\times\mathcal F_{\mathrm{sp}}$).
\end{proof}

\begin{remarknn}[On the word ``diagonalization'' here]
\label{rem:diagonal-terminology}
The enumeration used above --- listing pairs along successive finite
diagonals $i+j=0,1,2,\dots$ of the grid $\mathbb N\times\mathbb N$ --- is
the classical technique showing a countable union of countable sets is
countable, sometimes itself called a ``diagonal enumeration''. It is a
different technique, doing the opposite job, from Cantor's better-known
diagonal argument for \emph{uncountability} (used, e.g., to show
$2^{\mathbb N}$ or $\mathbb R$ is strictly larger than $\mathbb N$).
Keeping the two apart matters here: Theorem~\ref{thm:countable-pairs}
uses only the countability-preserving enumeration, and the
uncountability argument enters only in the next subsection, on a
different space.
\end{remarknn}

\begin{remarknn}[What has, and has not, been shown]
\label{rem:orbital-scope-so-far}
This section shows exactly one thing: the space of possible
(encoding, frame) pairs available to a two-selector orbital machine is
countably infinite, via an explicit, checked bijection. It does not yet
address how $\sigma_E,\sigma_F$ might depend on a computation's history,
whether any analogue of Definition~\ref{def:opacity-preserving}'s two
conditions applies to a pair of selectors jointly, or whether any
result of the preceding sections transfers to this two-selector
setting. Those are taken up later, each checked step by step rather
than assumed to carry over.
\end{remarknn}

\subsection{The genuine diagonalization: selector-\emph{functions}, not
selector-\emph{values}}
\label{sec:selector-functions}

Theorem~\ref{thm:countable-pairs} shows the space of possible
\emph{values} $(E,F)$ is countable: a complete, honest enumeration of
it misses nothing. The motivating analogy's appeal to Cantor
(``there will always be a new pair... not in the original list'')
refers to something different --- not the space of values, but the space
of \emph{selectors}: functions from a computation's history to a
choice of $(E,F)$, exactly what a \emph{policy} is in control theory or
a \emph{strategy} is in game theory, a rule for choosing based on
everything seen so far. Most of these are not computable at all, and
that is where the diagonal argument for uncountability applies.

\begin{definition}[Selector, computable selector]
\label{def:selector-function}
Let $\mathrm{Hist}$ be the set of finite sequences of local derivation
steps (a countably infinite set, by the same finite-object-over-a-countable-alphabet
argument as Proposition~\ref{prop:encoding-space-countable}). A
\emph{selector} is any function
$\sigma:\mathrm{Hist}\to\mathcal E_{\mathrm{sp}}\times\mathcal F_{\mathrm{sp}}$,
with no further restriction. A selector is \emph{computable} if it is
computed by some Turing machine (total on $\mathrm{Hist}$, under a
fixed effective coding of $\mathrm{Hist}$ and of
$\mathcal E_{\mathrm{sp}}\times\mathcal F_{\mathrm{sp}}$ by naturals,
available since both are countable by
Theorem~\ref{thm:countable-pairs}).
\end{definition}

\begin{proposition}[Selectors are uncountable; computable selectors are
not]
\label{prop:selectors-uncountable}
The set of all selectors has cardinality $2^{\aleph_0}$. The set of
computable selectors is countably infinite. In particular, almost all
selectors, in the cardinality sense, are non-computable.
\end{proposition}

\begin{proof}
Fix bijections $\mathrm{Hist}\cong\mathbb N$ and
$\mathcal E_{\mathrm{sp}}\times\mathcal F_{\mathrm{sp}}\cong\mathbb N$
(both available: $\mathrm{Hist}$ is countably infinite by construction,
and the second by Theorem~\ref{thm:countable-pairs}). Selectors then
correspond exactly to functions $\mathbb N\to\mathbb N$, i.e.\ to
elements of $\mathbb N^{\mathbb N}$. This set has cardinality
$2^{\aleph_0}$: it injects into $2^{\mathbb N}$ trivially (functions
into $\{0,1\}\subset\mathbb N$ already give $2^{\aleph_0}$ elements),
and $2^{\aleph_0}$ is an upper bound since
$|\mathbb N^{\mathbb N}|\le|(2^{\mathbb N})^{\mathbb N}|=2^{\aleph_0\cdot\aleph_0}=2^{\aleph_0}$;
by Cantor--Schr\"oder--Bernstein~\cite{CantorSchroederBernstein},
$|\mathbb N^{\mathbb N}|=2^{\aleph_0}$.
(This is exactly Cantor's original diagonal argument: given any purportedly
complete countable list $g_0,g_1,g_2,\dots$ of functions
$\mathbb N\to\mathbb N$, the function $h(n):=g_n(n)+1$ differs from
every $g_n$ at input $n$, so no countable list exhausts
$\mathbb N^{\mathbb N}$.)

Computable selectors, by contrast, are each specified by a finite
Turing machine program; there are only countably many finite programs
over a fixed finite alphabet, so the set of computable selectors is
countable, and infinite (e.g.\ the constant selectors
$\sigma\equiv(E,F)$, one for each of the countably many
$(E,F)\in\mathcal E_{\mathrm{sp}}\times\mathcal F_{\mathrm{sp}}$, are
each computable and pairwise distinct).
\end{proof}

\begin{remarknn}[The uncountability does not come from the richness of
the choices, but from the shape of the selector]
\label{rem:uncountable-not-from-values}
It is worth being exact about where the $2^{\aleph_0}$ comes from,
because the natural guess is wrong. One might think the selector space
is uncountable because there are infinitely many pairs $(E,F)$ to
choose among --- that the size of the codomain is what drives it. It is
not. The proof above already used only a two-element slice of the
codomain: functions $\mathrm{Hist}\to\{0,1\}$ alone are $2^{\aleph_0}$
of them. So even if the two selectors ranged over a \emph{finite} menu
of settings --- even just two, one bit of choice per history --- the
space of selectors would still be uncountable, and a complete list of
them would still be defeated by the same diagonal argument
($h(n):=1-g_n(n)$ suffices when the codomain is $\{0,1\}$). What makes
the selector space uncountable is the infinite \emph{domain} ---
$\mathrm{Hist}$, the unbounded history a selector reads --- not the
number of values it may return. This is the sharp form of the point:
the two-selector machine's unlistability survives shrinking the pool of
encodings and frames all the way down to a finite one, because it was
never the pool that was too large; it was the space of policies over an
unbounded past. It is exactly here, and not in
Theorem~\ref{thm:countable-pairs}'s countable space of values, that
Cantor's uncountability argument does its work
(Remark~\ref{rem:diagonal-terminology}).
\end{remarknn}

\begin{definition}[Oracle selector]
\label{def:oracle-selector}
An \emph{oracle selector} is any selector (Definition~\ref{def:selector-function})
that is not computable.
\end{definition}

\begin{corollary}[Oracle selectors exist, by cardinality alone]
\label{cor:oracle-exists}
Oracle selectors exist: since the computable selectors are countable
(Proposition~\ref{prop:selectors-uncountable}) and all selectors are
uncountable, at least one --- in fact, all but countably many --- of
the selectors is not computable.
\end{corollary}

\begin{remarknn}[The existence claim rests on cardinality alone]
\label{rem:cardinality-safe}
This existence claim uses nothing beyond a cardinality comparison: a
countable set cannot cover an uncountable one, so oracle selectors are
left over automatically, without any of them being singled out or
constructed to defeat anything specific. No assertion is made about
what any particular oracle selector does; none is needed. This is the
same style of argument used for the static case in
Section~\ref{sec:dynamic}: an oracle-computable object escapes any
fixed countable enumeration of computable ones purely because it is not
among them. What cardinality alone gives is existence, not a specific
task at which every computable selector provably fails and some
specific oracle selector provably succeeds; constructing such a task is
a separate question, and one to approach with the caution urged
throughout this paper --- an earlier, unrelated attempt at a
construction of this shape could not be completed without manufacturing
a false semantic assertion.
\end{remarknn}

\subsection{The problem reopens on the output: a Rice-style obstruction}
\label{sec:rice-output}

Even granting a selector, computable or not, that picks a good
$(E,F)$ for a computation's \emph{input}, whether the resulting
\emph{output} remains semantically well-behaved under that same frame
is a separate question. Checking it, in general, is impossible for a
computable process --- not by any construction specific to this paper,
but as a direct instance of one of the oldest results in computability
theory.

\begin{theorem}[Rice's theorem {\cite{Rice1953}}]
\label{thm:rice}
Let $\mathcal P$ be any property of partial computable functions
$\mathbb N\to\mathbb N$ that is non-trivial (some partial computable
function has it, some does not) and extensional (depends only on the
function computed, not on which program computes it). Then
$\{M : $ the function $M$ computes has property $\mathcal P\}$ is
undecidable.
\end{theorem}

\begin{definition}[$(F,\varphi)$-transparency]
\label{def:transparency}
Fix a computable frame $F\in\mathcal F_{\mathrm{sp}}$ and a target fact
$\varphi$ (an equality $s=t$ or inequality $s\ne t$ between terms). A
machine $M$ is \emph{$(F,\varphi)$-transparent} if, for every input $x$
on which $M$ halts, the output $M(x)$, read as a term and evaluated
under $F$, satisfies $\varphi$.
\end{definition}

\begin{proposition}[Transparency is undecidable, when non-trivial]
\label{prop:transparency-undecidable}
Let $F,\varphi$ be such that some machine is $(F,\varphi)$-transparent
and some machine is not (e.g.\ $F=F_2$ of
Theorem~\ref{thm:semantic-underdetermination}, recoded as a computable
frame per Remark~\ref{rem:f1f2-computable}, and $\varphi:$ ``output
$=a$'': the machine that always outputs the term $a$ is
$(F_2,\varphi)$-transparent, the machine that always outputs the term
$b$ is not, since $a^{F_2}\ne b^{F_2}$). Then
$\{M : M\text{ is }(F,\varphi)\text{-transparent}\}$ is undecidable.
\end{proposition}

\begin{proof}
$(F,\varphi)$-transparency depends only on the partial function $M$
computes (two machines computing the same function produce identical
outputs on identical inputs, hence agree on whether every output
satisfies $\varphi$ under $F$): it is extensional in the sense required
by Theorem~\ref{thm:rice}. It is non-trivial by the hypothesis (and the
worked example given). Theorem~\ref{thm:rice} applies directly.
\end{proof}

\begin{remarknn}[The recursive reopening]
\label{rem:recursive-reopening}
Even after a selector has fixed a frame $F$ that correctly interprets a
computation's input, no algorithm can decide, in general, whether a
given machine's \emph{output} continues to respect any fixed,
non-trivial semantic fact under that same $F$: by
Proposition~\ref{prop:transparency-undecidable}, the transparency
problem is exactly as hard, in the sense of ordinary undecidability, as
any other non-trivial semantic property of programs. Interpretability
of the question does not propagate to interpretability of the answer,
and provably cannot in general.
\end{remarknn}

\begin{remarknn}[What hides the output, and what is merely undecidable
about it, are different things]
\label{rem:output-two-sources}
The undecidability just stated should not be mistaken for the reason the
output is opaque. What hides the output is structural, and there are two
such structural reasons, both already in play in this paper. First, the
static principle acting on the output: $M(x)$ is itself a syntactic
object, so whatever semantic invariant is protected in it --- partially
or totally --- is invisible to any purely syntactic inspection of it,
for the structural reason of Lemma~\ref{lem:sip-static}, that syntax
cannot see semantic invariants. Second, the encoding--frame relativity
this very section is built on: whether the output reads as meaningful is
relative to a pair $(E,F)$, and under a pair other than the intended one
it may fail to parse, answer a different question, or answer nothing
discriminating (Categories~1--3,
Propositions~\ref{prop:category1}--\ref{prop:category3}) --- so an output
can be opaque because no available reading renders it, quite apart from
any syntactic hiding. Either inaccessibility is the blindness, and each
is present whether or not anything about the output is decidable. Rice's
theorem concerns a different question again: not what hides the
invariant, but whether one can algorithmically \emph{decide} that a
fixed invariant holds --- which, by
Proposition~\ref{prop:transparency-undecidable}, one cannot in general.
The three must be kept apart: the SIP and the encoding--frame relativity
are why there is, or is not, a readable invariant at all; Rice is why
verifying a fixed one is beyond an algorithm. The fable's ``reflection
of Rice's theorem'' is a narrative allusion, gesturing at material
outside this paper's scope, not a claim that Rice is the source of the
blindness. This same separation is set out for the output row of the
fable's table (Section~\ref{sec:fable-mapped}); it is recorded here
because this is where the output's reopening is first proved.
\end{remarknn}

\subsection{One correct pair, and three ways to fail: the fable made
precise}
\label{sec:three-ways-fail}

The motivating analogy describes a question with an intended meaning,
posed under some encoding, read under some frame, where only one pair
of choices recovers what was meant; every other pair either fails to
parse, answers something else, or answers nothing discriminating at
all. This subsection makes each precise, for a specific, fully worked
question, in the same spirit as Section~\ref{sec:total-opacity}'s
worked examples: concrete enough to prove completely.

\begin{remarknn}[Scope: one concrete question]
\label{rem:three-ways-scope}
This is worked for the specific question of
Theorem~\ref{thm:semantic-underdetermination}, fully rigorously. Doing
the same for an \emph{arbitrary} externally chosen target question
would require a general model-existence fact --- if an encoding $E$
derives neither a question nor its negation, models of $E$ disagreeing
on it exist --- which holds in general (it follows from Gödel's
completeness theorem) but is built here only for the specific case at
hand, by direct construction. Extending to arbitrary $(E^*,q^*)$ is
open, and would need that general model-existence argument built with
the same care given to every other construction here.
\end{remarknn}

To make the fable's three ways of failing precise, we need a fixed
target to fail against: a specific question, posed under a specific
encoding, with a specific frame that gives the answer meant. This
target is fixed once, from outside the construction --- exactly as the
true key $K^*$ was fixed externally in
Proposition~\ref{prop:cipher-instance}.

\begin{definition}[External target]
\label{def:external-target}
Fix, once and externally, three things: a target encoding
$E^*=\{A1,A2\}\in\mathcal E_{\mathrm{sp}}$; a target question $q^*$,
taken to be ``$a+b\stackrel{?}{=}b+a$''; and a target frame
$F^*\in\mathcal F_{\mathrm{sp}}$ modelling $\{A1,A2\}$ that gives the
intended answer (say $F^*=F_1$ of
Theorem~\ref{thm:semantic-underdetermination}, recoded as computable,
so the intended answer is ``true''). A candidate pair $(E,F)$
\emph{recovers} $q^*$ if $E=E^*$ and $F\models a+b=b+a$.
\end{definition}

\begin{proposition}[Category 1: syntactic vacuity]
\label{prop:category1}
There are countably infinitely many $E\in\mathcal E_{\mathrm{sp}}$
under which $q^*$ is not even a well-formed term: $E$'s rewriting
system does not have $a$ and $b$ both among its symbols.
\end{proposition}

\begin{proof}
For each $n\ge1$, let $E_n=\{c_n\to c_n\}$ (using
one of the countably many auxiliary symbols of
Definition~\ref{def:encoding-space}, none of them $a$ or $b$): $a+b$ is
not a term over $E_n$'s alphabet at all, so $q^*$ cannot be posed under
$E_n$. The $E_n$ are pairwise distinct (distinct symbols $c_n$), giving
countably infinitely many such $E$.
\end{proof}

\begin{proposition}[Category 2: semantic disagreement]
\label{prop:category2}
There are countably infinitely many $F\in\mathcal F_{\mathrm{sp}}$
modelling $E^*=\{A1,A2\}$ that disagree with $F^*$ on $q^*$, i.e.\ with
$F\models a+b\ne b+a$.
\end{proposition}

\begin{proof}
Theorem~\ref{thm:semantic-underdetermination} exhibits one such frame,
$F_2$ (recoded as computable per
Remark~\ref{rem:f1f2-computable}). For infinitely many pairwise-distinct
variants, note that Remark~\ref{rem:f1f2-computable}'s recoding already
moves $F_2$'s domain from $\mathbb N\cup\{a^*,b^*\}$ to $\mathbb N$
itself, by relabelling $a^*,b^*$ as two natural numbers not otherwise
used. The same freedom gives further variants directly, each built on
the same ``leftmost sink wins'' idea as $F_2$ itself
(Theorem~\ref{thm:semantic-underdetermination}'s proof): for each
$n\ge1$, pick a fresh pair $p_n\ne q_n\in\mathbb N$, disjoint from all
previously chosen pairs, and define $s^{F^{(n)}},+^{F^{(n)}}$ to agree
with the standard successor and addition on
$\mathbb N\setminus\{p_n,q_n\}$, but with
$s^{F^{(n)}}(p_n)=p_n$, $s^{F^{(n)}}(q_n)=q_n$ (fixed points), and
$+^{F^{(n)}}(p_n,y):=p_n$, $+^{F^{(n)}}(q_n,y):=q_n$ for all $y$,
$+^{F^{(n)}}(x,p_n):=p_n$, $+^{F^{(n)}}(x,q_n):=q_n$ for
$x\notin\{p_n,q_n\}$ (the same five-case definition as
Theorem~\ref{thm:semantic-underdetermination}'s $F_2$, now with $p_n,q_n$
literal elements of $\mathbb N$ rather than adjoined outside it). The
same case-by-case check as that theorem's proof (A1, A2 verified
exhaustively over the five shapes of argument pairs) applies verbatim
with $p_n,q_n$ in place of $a^*,b^*$, giving $F^{(n)}\in\mathcal
F_{\mathrm{sp}}$ modelling $\{A1,A2\}$, with $a^{F^{(n)}}:=p_n$,
$b^{F^{(n)}}:=q_n$ satisfying $a+^{F^{(n)}}b=p_n\ne q_n=b+^{F^{(n)}}a$.
The $F^{(n)}$ are pairwise distinct (distinct pairs $(p_n,q_n)$), giving
countably infinitely many such $F$.
\end{proof}

\begin{proposition}[Category 3: semantic vacuity --- the ``beaver''
frames]
\label{prop:category3}
There is a frame $F^{(1)}\in\mathcal F_{\mathrm{sp}}$ modelling
$\{A1,A2\}$ under which $q^*$ receives a well-defined answer that
carries no discriminating information whatsoever: $F^{(1)}$ answers
``$a+b=b+a$'' regardless of what $a,b$ are, or what $q^*$ was really
asking. Countably many formally distinct (though isomorphic) recodings
of it exist in $\mathcal F_{\mathrm{sp}}$.
\end{proposition}

\begin{proof}
Let $F^{(1)}$ have domain $D=\{0\}$, with $0^F=0$, $s^F(0)=0$, and
$+^F(0,0)=0$. Checking A1: $0+^F0=0$ matches $0^F=0$, as required.
Checking A2: $s^F(0)+^F0=0+^F0=0$, and separately
$s^F(0+^F0)=s^F(0)=0$; the two sides agree. So $F^{(1)}$ models
$\{A1,A2\}$.

Now, since $D$ has only one element, $a^{F^{(1)}}=b^{F^{(1)}}=0$ is
forced --- there is nowhere else for either to point. Hence
$a+^{F^{(1)}}b=0=b+^{F^{(1)}}a$ holds trivially, whatever $a,b$ were
originally intended to name. This is not a correct answer so much as
an inability to ask the question at all: the frame cannot express, let
alone correctly or incorrectly answer, anything depending on two
things being distinguishable, because it has nowhere to put a second
thing.

Finally, recoding this same one-element frame under any of the
countably many naturals as the name of its single point produces
countably many formally distinct elements of $\mathcal F_{\mathrm{sp}}$,
all isomorphic to $F^{(1)}$.
\end{proof}

\begin{remarknn}[The same idea as $\Obot$-blindness, in a new setting]
\label{rem:beaver-is-obot}
$F^{(1)}$ is the same underlying idea as
Proposition~\ref{prop:obot-trivial}, transplanted from language
recognition to equality queries: a structure with only one point cannot
distinguish any two things put into it, exactly as an observer
collapsing everything to $\star$ cannot distinguish any two inputs. The
two are not literally the same instance --- one concerns a machine
receiving $\Obot(x)=\star$, the other a degenerate frame answering
equality queries --- but the shape is identical. This is the fable's
hamster, or beaver: not wrong, not silent, but structurally incapable
of holding the distinction the question depends on.
\end{remarknn}

\begin{remarknn}[Uniqueness of the correct pair, up to extensional
identification]
\label{rem:unique-pair}
$(E^*,F^*)$ recovers $q^*$ by construction
(Definition~\ref{def:external-target}). Every $E$ of
Proposition~\ref{prop:category1} fails outright (Category~1); every $F$
of Proposition~\ref{prop:category2} gives the wrong answer to a
well-posed question (Category~2); $F^{(1)}$ of
Proposition~\ref{prop:category3} gives an answer that agrees
numerically but for the wrong reason, carrying no information about
$a,b$ specifically --- arguably a fourth, degenerate way of not
recovering $q^*$, the agreement being accidental to total collapse
rather than a correct resolution of the Skolem pair. Categories~1--3 are
salient, each-infinite ways of failing to recover $q^*$, exhibited to
show how varied and how numerous the failures are; they are not claimed
to be an exhaustive or mutually exclusive partition of every failing
pair (the ``fourth way'' just noted already shows the boundaries are not
sharp), and nothing below depends on their being one. What \emph{is}
used is only the other direction: that $(E^*,F^*)$ recovers $q^*$ and,
up to the extensional identification below, does so uniquely. As already
noted for encodings (Remark~\ref{rem:cardinality-safe}) and selectors
(Section~\ref{sec:selector-functions}), the correct pair is unique only
up to extensional identification: any frame computing exactly the same
function as $F^*$, however differently coded, recovers $q^*$ equally
well and counts as the same pair, not a second one.
\end{remarknn}

\begin{remarknn}[What the orbital machine shows that no binary verdict
can]
\label{rem:infinite-shades}
What has just been proved has a shape worth stating in its own right.
Categories~1--3 are not merely non-empty; each contains \emph{countably
infinitely many} genuinely distinct members: countably many encodings
under which the question cannot even be posed
(Proposition~\ref{prop:category1}), countably many frames that pose it
correctly and answer it wrongly, each in its own distinct way rather
than one shared error repeated (Proposition~\ref{prop:category2}), and
countably many degenerate framings that answer without saying anything
at all (Proposition~\ref{prop:category3}) --- set against exactly one
pair that gets it right. A plain accessible/blocked verdict, of the
kind the static and Dynamic SIP each return, cannot show this: it
reports whether a fact can be seen or not, one bit, where the orbital
machine resolves the whole space of ways of \emph{almost} seeing it, or
failing to, into a structured landscape of distinguishable failures,
most of them countably infinite in their own right. The same
finer-grained question recurs later in a different vocabulary: an
algorithm's output is there called $\mathcal C$-usable or
$\mathcal C$-opaque (Definition~\ref{def:usable-algorithm}), a single
cut for a fixed class $\mathcal C$, appropriately coarse for the
statistical question that section asks. Binary usability tells you
which side of one line an object falls on; the orbital machine maps the
terrain the line was drawn across.
\end{remarknn}

\subsection{The orbital machine as an active instance of ``all
assumptions made explicit''}
\label{sec:orbital-as-interpreted}

Section~\ref{sec:interpreted-machine} defined an interpreted machine as
a standard Turing machine equipped with a fixed encoding and a fixed
frame, the frame used only \emph{passively}, to read a specific target
fact off the machine's output once computation halts. The orbital
machine, defined independently above, instantiates that same idea of
``an encoding and a frame, made explicit'' --- but with the frame used
\emph{actively}, consulted during the computation itself, not only at
the end.

\begin{definition}[Semantic oracle]
\label{def:semantic-oracle}
For $F\in\mathcal F_{\mathrm{sp}}$, the \emph{semantic oracle} $O_F$ is
the function taking a pair of terms $(s,t)$ to $O_F(s,t)=1$ if
$F\models s=t$, and $O_F(s,t)=0$ otherwise. (Since $F$ is a computable
frame, $O_F$ is itself a computable function; it is called an
``oracle'' because of the role it plays --- direct, one-step
consultation of a specific world's answer --- not because it exceeds
what a Turing machine can compute. This matches the fable's Einstein
example: being asked a question in ancient Egyptian rather than German
is not a matter of the answerer's raw intelligence, but of which
channel the question arrives through.)
\end{definition}

The architecture about to be defined --- a computable function
consulted mid-computation, with later queries free to depend on
earlier answers --- is an already-studied notion.
\cite[Remark~7.2]{Buono2026ow} formalizes exactly this, independently,
as an \emph{adaptive observer}, of which a computational oracle queried
sequentially is one instance. What follows is, in that source's
vocabulary, an adaptive observer of unbounded depth, built on a
computable oracle rather than an arbitrary one; this uses only that
basic point, not that source's further Landauer-cost or quantum
development of the idea.

In plain terms, before the formal definition: an oracle orbital
machine works like an ordinary Turing machine, except that at any
point it may pause and ask the semantic oracle a yes/no question about
two terms it has built so far, then continue using the answer.

\begin{definition}[Oracle orbital machine]
\label{def:oracle-orbital}
Given selectors fixed to values $(\sigma_E,\sigma_F)=(E,F)$, the
\emph{oracle orbital machine} $T^{O_F,E}$ is a Turing machine $T$ with
three properties. First, it reads its input as a term over $E$'s
alphabet, halting and rejecting immediately if the input is not
well-formed under $E$ (Category~1, Proposition~\ref{prop:category1}).
Second, at any point during its computation it may query $O_F$ on any
pair of terms it has constructed so far, receiving a definite answer
in one step. Third, apart from these queries, it proceeds as an
ordinary Turing machine operating under $E$'s rewriting rules.
\end{definition}

\begin{proposition}[The passive Interpreted Machine is the
never-query case]
\label{prop:oracle-recovers-interpreted}
An oracle orbital machine $T^{O_F,E}$ that never queries $O_F$ during
its computation, consulting $F$ only once, passively, after halting,
to evaluate a single target fact about its output, is exactly the
interpreted machine $\mathcal I=(T,E,F)$ of
Definition~\ref{def:interpreted-machine}.
\end{proposition}

\begin{proof}
This follows immediately from the two definitions. With no queries during
computation, $T^{O_F,E}$'s run is identical, step for step, to $T$
operating under $E$ alone (Definition~\ref{def:dynamic-system} with
the degenerate update $\upsilon_{\mathrm{id}}$, as already used in
Corollary~\ref{cor:static-as-degenerate}). The only remaining role for
$F$ is the single post-halting evaluation, which is exactly
Definition~\ref{def:interpreted-machine}'s passive reading.
\end{proof}

\begin{remarknn}[What active querying adds, precisely]
\label{rem:active-query-adds}
The passive case is recovered exactly, not approximated
(Proposition~\ref{prop:oracle-recovers-interpreted}): nothing built in
Section~\ref{sec:interpreted-machine} is lost, and the oracle orbital
machine is a strict extension of it. What active querying adds is
availability \emph{during} computation of facts $E$ alone cannot
derive. By Lemma~\ref{lem:frozen} applied to $E=\{A1,A2\}$, $T$
operating on $E$'s rules alone can never determine whether $a+b=b+a$,
no matter how long it runs; $T^{O_F,\{A1,A2\}}$ can ask $O_F$ and
receive a definite answer in a single step, mid-computation, usable in
its subsequent behaviour (e.g.\ branching on the answer). This is the
precise sense in which the oracle matters here: not because $O_F$
computes anything $T$ could not eventually compute by other means, but
because $E$'s own rules never license deriving that specific fact at
all.
\end{remarknn}

\begin{remarknn}[Every machine here is, in raw power, exactly a Turing
machine]
\label{rem:always-standard-tm}
Three results in three separate places combine into one worth stating
plainly, in a single place: every machine built in this paper is, in
raw computational power, an ordinary Turing machine, no more and no
less. The interpreted machine of Section~\ref{sec:interpreted-machine}
is one by definition, with $F$ read only passively at the end
(Definition~\ref{def:interpreted-machine}), and the syntactic process
it runs is exactly the standard, single-fixed-system case
(Corollary~\ref{cor:static-as-degenerate}). The orbital machine, before
any oracle is added, is the same object with $(E,F)$ chosen by
selectors rather than by hand (Definition~\ref{def:two-selectors});
once a value is chosen, it is again exactly that case. Even the oracle
orbital machine, which \emph{can} query $O_F$ mid-computation rather
than only at the end, adds nothing beyond Turing power: $F$ being a
computable frame (Definition~\ref{def:computable-frame}) makes $O_F$ a
computable function, so a mid-run query could always be replaced by an
ordinary Turing machine computing that same answer inline, with no
oracle (Remark~\ref{rem:active-query-adds} makes this precise). This
extends to a richer setting exactly what~\cite[Remark~9.1]{Buono2026oh}
establishes for a plain structural observer: such an observer answers
no queries and never extends the underlying computational model, since,
in that source's words, ``the observer only reduces available
information and never increases computational power''. What each layer
here adds is not computational power but honesty: making explicit a
choice --- which encoding, which interpretation, when to consult it ---
that an ordinary, unlabelled Turing machine always makes too, silently,
through whoever builds it.
\end{remarknn}

\begin{remarknn}[$E$ and $F$ do distinct jobs: what is askable, and what
the answer is]
\label{rem:e-f-jointly-decide}
Turing-equivalence is one half of the architecture; the other half is
that $E$ and $F$ play genuinely different roles, neither substituting
for the other. $E$ decides which semantic invariants are even
\emph{askable}: if $E$'s alphabet lacks $a$ or $b$, the question
``does $a+b=b+a$'' is not merely unanswered but not well-formed at all
(Category~1, Proposition~\ref{prop:category1}). $F$ then decides, for
whichever question $E$ has made askable, what the \emph{answer} is when
$O_F$ is consulted (Definition~\ref{def:semantic-oracle}). Change $E$
alone and the menu of questions changes; change $F$ alone, holding $E$
fixed, and the answers to that same menu can change (Category~2,
Proposition~\ref{prop:category2}). What the oracle orbital machine can
learn mid-computation takes the pair, doing two distinct jobs,
together.
\end{remarknn}

\begin{remarknn}[What remains unconnected, and why]
\label{rem:still-unconnected}
This subsection connects the orbital machine to
Section~\ref{sec:interpreted-machine}'s Interpreted Machine only. The
connection to the Dynamic SIP of
Sections~\ref{sec:dynamic}--\ref{sec:total-opacity}, to the
structural-observer framework of
Sections~\ref{sec:total-opacity}--\ref{sec:order-blind}, and to Observer
World (\cite{Buono2026ow}) is the third and final step, made next ---
after the orbital machine's standalone construction and this
active-oracle instantiation, in that order.
\end{remarknn}

\subsection{The third connection: Dynamic SIP, the observational
hierarchy, and Observer World}
\label{sec:third-connection}

This connection is made in two independent pieces, one per selector,
never jointly. Making both $\sigma_E$ and $\sigma_F$ history-dependent
\emph{at once} would be a genuine two-axis system --- new mathematics,
not an application of what is already proved --- and nothing below
attempts it; the two axes are connected one at a time, each a direct
application of results already in hand.

\begin{remarknn}[What is at stake throughout: a semantic invariant]
\label{rem:semantic-invariant-framing}
\cite{Buono2026sip}'s own title states the target: \emph{syntactic
systems cannot see semantic invariants}. The two connections below
should be read with that target in view. In the encoding-axis
connection, the semantic invariant at stake is the one
Theorem~\ref{thm:dynamic-sip} already concerns: the relative order of
$a,b$, a fact about truth (whether $a+b=b+a$ holds), not about symbols.
In the frame-axis connection, the invariant differs across its three
settings but is the same in kind: membership of a string in a language
$L$ (the non-recognition of~\cite[Cor.~2.8]{Buono2026oh} is exactly
non-access to whether $x\in L$, a property of what $x$ \emph{means}),
and identity of a plaintext (the impossibility of recovery
in~\cite[Prop.~4.9]{Buono2026ow} is non-access to which $M$ was sent).
Both axes connect to these targets, not to a generic notion of
syntactic restriction.
\end{remarknn}

\subsubsection*{The encoding axis, alone, is Dynamic SIP}

\begin{remarknn}[Direct identification, not analogy]
\label{rem:encoding-is-dynamic-sip}
Suppose the encoding selector $\sigma_E$ is permitted to depend on
history, $E_{n+1}:=\sigma_E(H_n)$, while the frame is held fixed. Then
$(E_n)_{n\ge0}$, together with the local derivations it generates, is a
dynamic rewriting system generated by the update
$\upsilon:=\sigma_E$, in the exact sense of
Definition~\ref{def:dynamic-system} --- literally, term for term, one.
Consequently:
\begin{itemize}
  \item if $\sigma_E$ is opacity-preserving
    (Definition~\ref{def:opacity-preserving}), Theorem~\ref{thm:dynamic-sip}
    applies to the encoding axis directly: the semantic invariant at
    stake --- the relative order of whichever Skolem pair is frozen in
    $E_0$, i.e.\ the truth of $a+b=b+a$ --- remains inaccessible at
    every subsequent value $E_n$ the selector produces, however long
    the computation runs; the syntax never settles it, and the
    selector's own updates cannot make it settle it either;
  \item if, further, $\sigma_E$'s access to $H_n$ is filtered through a
    swap-blind structural observer $O$ (Definition~\ref{def:swap}),
    Corollary~\ref{cor:total-opacity} upgrades this to hold for
    \emph{every} such $\sigma_E$ satisfying condition~(i), with no
    case-by-case verification of condition~(ii).
\end{itemize}
No new theorem is proved here: this is a direct application of
Sections~\ref{sec:dynamic}--\ref{sec:total-opacity} to the encoding
selector, made precise by noting that the orbital machine's encoding
axis was, from Definition~\ref{def:two-selectors} onward, already the
same object as Dynamic SIP's $R_n$ --- simply not yet identified as
such, per the deliberate independence of Section~\ref{sec:orbital}'s
opening.
\end{remarknn}

\subsubsection*{The frame axis, alone: the beaver, $\Obot$, and Cell (d)}

The frame axis is not identified with a structural observer in the same
direct way: a structural observer $O$
(Definition~\ref{def:structural-observer}) only partitions strings and
asserts nothing, while a frame $F$ decides definite truth values. These
remain distinct objects throughout. What connects them, per
Remark~\ref{rem:semantic-invariant-framing}, is that each blocks access
to a specific semantic invariant by the same underlying shape.

\begin{remarknn}[The same shape, three times, on three semantic
invariants: beaver, $\Obot$, Cell~(d)]
\label{rem:beaver-obot-celld-chain}
Remark~\ref{rem:beaver-is-obot} noted that Category~3's one-element
beaver frame (Proposition~\ref{prop:category3}) shares its underlying
shape with $\Obot$'s triviality (Proposition~\ref{prop:obot-trivial}):
both collapse every distinction to a single output, and both are
unbreakable by computational power for the same reason --- there is
nowhere for a second value to be told apart from the first. The blocked
invariant differs across the three settings, and naming it case by case
is the point here. For the beaver frame, it is the truth of $a+b=b+a$.
For $\Obot$, it is membership $x\in L$, for whatever non-trivial $L$ is
asked. For Cell~(d) (Proposition~\ref{prop:cell-d}), it is which
plaintext $M$ was sent. The chain runs: beaver frame (Category~3, this
section), then $\Obot$ (Section~\ref{sec:total-opacity}), then Cell~(d)
(Section~\ref{sec:cell-d}), with
Proposition~\ref{prop:celld-is-n0} already making the last step precise
as the degenerate, single-step case of total opacity. Three different
semantic invariants, in three different formal settings (a one-element
algebraic structure, a string-partitioning function, a cryptographic
adversary), all blocked by the identical collapse-to-one-point
argument.
\end{remarknn}

\begin{remarknn}[An order on frames: the natural next definition]
\label{rem:frame-order-open}
Sections~\ref{sec:total-opacity}--\ref{sec:order-blind} order
structural observers by richness
($\Obot\preceq\cdots\preceq\Oprof\preceq\cdots$). The analogous order on
frames --- $F_1\preceq F_2$ if every fact $F_1$ decides, $F_2$ also
decides and agrees with --- is the natural next definition, under which
the beaver frame sits at the bottom, exactly where $\Obot$ sits in the
observer order. What has been shown here is the bottom point of that
order: the beaver (Category~3) behaves like $\Obot$. Whether the order
as a whole is rich enough to support results of the same strength as
Sections~\ref{sec:total-opacity}--\ref{sec:order-blind} --- a full
characterization, a dynamic extension --- is open.
\end{remarknn}

\begin{remarknn}[The two-axis case, deliberately left alone]
\label{rem:two-axis-again}
Both connections above hold one selector fixed while the other varies.
A joint statement --- an analogue of Theorem~\ref{thm:dynamic-sip} or
Corollary~\ref{cor:total-opacity} for $\sigma_E,\sigma_F$ evolving
together --- is not attempted here, for a reason this paper has met
before: this is precisely the shape of construction that, in a separate
and unrelated project, could not be completed without manufacturing a
false semantic assertion. The identifications above are safe because
each is a one-axis application of results already fully proved; a
two-axis version would be new mathematics, not an application, and is
left for future work approached with the same caution.
\end{remarknn}

\subsection{The fable, mapped}
\label{sec:fable-mapped}

With the formal content built, the correspondence promised in
Remark~\ref{rem:formalizing-the-fable} can be checked line by line,
rather than asserted in general terms.

\begin{center}
\small
\begin{tabular}{p{5.3cm}p{6.8cm}}
\toprule
\textbf{What the fable says} & \textbf{What this section proves} \\
\midrule
A machine that answers by reading reality directly, not by calculating
& The active-query architecture built above:
Definitions~\ref{def:semantic-oracle}, \ref{def:oracle-orbital}
(Section~\ref{sec:orbital-as-interpreted}); named informally already at
Definition~\ref{def:two-selectors} (Remark~\ref{rem:orbital-name}). \\
Infinitely many ways to select the encoding through which the question
is posed
& The encoding selector $\sigma_E$ and encoding space
$\mathcal E_{\mathrm{sp}}$ (Definitions~\ref{def:encoding-space},
\ref{def:two-selectors}). \\
Infinitely many ways to interpret it
& The frame selector $\sigma_F$ and frame space
$\mathcal F_{\mathrm{sp}}$ (Definitions~\ref{def:computable-frame},
\ref{def:two-selectors}). \\
In some worlds the question would not make sense; addressed to the
only entity capable of answering it, yet unable to understand it
& Category~1, syntactic vacuity: infinitely many encodings under which
the question is not even a well-formed term
(Proposition~\ref{prop:category1}). \\
The primordial question, asked to a hamster
& Category~3, the ``beaver'' frames: a frame structurally incapable of
holding the distinction the question depends on, answering something
well-defined but discriminating nothing
(Proposition~\ref{prop:category3}, Remark~\ref{rem:beaver-is-obot}). \\
The answer might be impossible for us to interpret --- 42, with no idea
how to make sense of it
& Two results, not one: no algorithm decides, in
general, whether a machine's output remains semantically transparent
(Proposition~\ref{prop:transparency-undecidable}) --- and, more
literally, an output can be produced correctly by the right pair and
still be statistically indistinguishable from noise to every observer
who receives it, however powerful, as Section~\ref{sec:omega} shows
much later in this paper (Definition~\ref{def:usable-algorithm},
Proposition~\ref{prop:omega}, Remark~\ref{rem:one-observer-picture}):
not merely undecidable whether it is readable, but actually opaque, and
opaque for reasons that do not turn on the reader's computational power
--- $\Omega$ being the sharpest witness, not the only route. \\
Even with both selectors, we could not test every possibility: Cantor
reminds us there is always a new pair not in the list, by
diagonalization
& The genuine diagonal argument, applied correctly to
selector-\emph{functions} rather than to selector-\emph{values}
(Proposition~\ref{prop:selectors-uncountable}; see
Remark~\ref{rem:cardinality-safe} for why this is not the same claim as
``the pairs themselves are uncountable'', which is false ---
Theorem~\ref{thm:countable-pairs} shows the opposite for values). \\
Every combination carries a different degree of semantic blindness
& Not a loose figure of speech: Categories~1, 2, and~3 are each shown
to contain \emph{countably infinitely many}, genuinely distinct
members, not one representative apiece
(Remark~\ref{rem:infinite-shades}), together with the graded
disagreement of Category~2 itself (Proposition~\ref{prop:category2})
and the spectrum of observer/frame richness running through
Sections~\ref{sec:total-opacity}--\ref{sec:order-blind}
($\Obot$, $(B,\cdot)$, $\Oprof$-blindness). \\
Even the one transparent reality produces an answer that acquires its
own points of blindness --- a reflection of Rice's theorem, like a
famous separation problem
& The blindness is the static SIP: an output, once produced, is again a
syntactic object, and whatever semantic invariant is protected in it ---
partially or totally --- is invisible to any later syntactic inspection
of it (Lemma~\ref{lem:sip-static}, ``syntax cannot see semantic
invariants''). This is the ``separation problem'' the fable places next
to the phrase, the same separation running through this paper's source
material (\cite{Buono2026sip}'s title;
\cite{Buono2026obstruction}'s \emph{syntactic separation}). The fable's
``reflection of Rice's theorem'' is a narrative allusion, not a claim
that Rice explains the blindness; Rice enters as a separate, genuine
fact --- whether an output's transparency \emph{holds} is itself
undecidable (Theorem~\ref{thm:rice},
Proposition~\ref{prop:transparency-undecidable},
Remark~\ref{rem:recursive-reopening}) --- which concerns deciding the
invariant, not the syntactic inaccessibility that hides it. \\
The same phenomenon is the nonexistence at the base of every logical
paradox --- but that is another story
& Not addressed. Left exactly where the fable leaves it. \\
\bottomrule
\end{tabular}
\end{center}

\begin{remarknn}[Where the correspondence is not exact]
\label{rem:mapping-caveats}
Two entries above are looser than the rest, and worth flagging as such.
First, the Category~1 row folds together two images the fable states
separately --- ``the question would not make sense'' and ``addressed to
the only entity capable of answering it, yet unable to understand it''
--- and reads both as syntactic vacuity, an unparseable question; the
fable's ``asked to a hamster'' is placed one row down, with Category~3's
structurally undiscriminating answerer. The fable does not itself draw
the line exactly where Categories~1 and~3 draw it (unparseable question
versus well-formed question put to an answerer that cannot discriminate),
so these placements are by best fit, not by a claim that the fable
distinguishes precisely these formal cases. Second, the
oracle-consultation architecture in the first row is named early
(Remark~\ref{rem:orbital-name}) and built as an explicit mechanism
earlier in this section (Section~\ref{sec:orbital-as-interpreted}); by
the point this table appears, that mechanism is already complete, not
merely promised.
\end{remarknn}

\section{Scope}
\label{sec:scope}

Two separate limitations apply to the two halves of this paper, and
this section states them plainly.

For the Dynamic SIP of Sections~\ref{sec:dynamic}--\ref{sec:order-blind}:
this paper generalizes the single-encoding case --- one dynamic rule
sequence, one update mechanism, one limiting case filtered through a
single structural observer. Two independently varying axes, and an
unconditional oracle-based separation theorem, are outside its scope.
An earlier, independent attempt at the latter, in a different project,
rested on manufacturing a false assertion and was abandoned rather than
repaired; every result here stays on the side of that line where the
invariant is a true fact, preserved because the available operations
preserve true facts.

For the semantic frames of Section~\ref{sec:semantic-frame}: every
result there concerns a single, fixed frame, never one that evolves
over the course of a derivation. A frame that updates dynamically, in
the same shape as the encoding update $\upsilon$, is identified but not
attempted (Remark~\ref{rem:semantic-frame-static-only}): it would
combine with a dynamic encoding into the genuine two-axis system just
referred to, and the right notion of opacity-preservation is itself
unclear once frame updates can change what is \emph{true} rather than
only what has been \emph{derived}.

\section{A second, independently arising instance: the Andromeda
paradox}
\label{sec:andromeda}

This section develops, on its own terms and without depending on
anything constructed above, a second phenomenon that instantiates this
paper's central theorem in a completely different domain, discovered
independently rather than built to fit. It is presented exactly as
Section~\ref{sec:orbital} was: standalone for now, with the
reconnection to the rest of the paper deferred rather than forced.

\subsection{The story}

\begin{quote}
On a planet in the Andromeda galaxy, Chicco stands on the beach
looking at the Earth, and uses his own concept of ``now'' to define
which events on Earth are simultaneous with this instant. According to
Chicco, the events happening on Earth right now are, say, those of
1026~AD.

Then there is Bruno Sacchi, who is lazy and decides to go visit his
friend Chicco on planet 3C with his spaceship. So he takes the shuttle
and travels very fast toward Chicco. But Bruno has a different space of
simultaneity. His concept of ``now'' is rotated, relative to Chicco's,
in spacetime. So when Bruno looks toward the Earth to see what is
happening now, he sees a completely different epoch --- perhaps the
year 3026~AD.

But at the very same physical instant, Chicco sees the past and Bruno
sees the future, and both are right.

One might think this creates a paradox, but it does not, because the
speed of light imposes a delay that sets everything right again.
\end{quote}

The story retells a real and well-established phenomenon in special
relativity, the relativity of simultaneity. The underlying argument was
discovered independently twice: first by C.~W.~Rietdijk in
1966~\cite{Rietdijk1966}, then by Hilary Putnam in 1967~\cite{Putnam1967},
each using it to argue for a substantive philosophical thesis --- that
the relativity of simultaneity, combined with a transitivity
requirement on what counts as ``real'', forces a block-universe view in
which future events are already as real as past ones. Roger Penrose,
crediting earlier related work by Rindler, gave the argument its
now-standard, vivid concrete form in 1989 --- two pedestrians passing on
an Earth street, one walking toward the Andromeda galaxy, disagreeing by
days about whether a distant space fleet has already
launched~\cite{Penrose1989}. That is the version most readers know as
``the Andromeda paradox'', and the version this paper's retelling is
built to be recognised against.

\begin{remarknn}[How this retelling differs, and why]
\label{rem:andromeda-vs-penrose}
Two things distinguish the telling above from Penrose's. First, the
geometry is reversed: Penrose places both observers on Earth, looking
outward across the vast distance to Andromeda, so that an ordinary
walking-speed velocity difference is amplified, by that distance alone,
into a difference of centuries in what each calls ``now''. The telling
here instead places one observer (Chicco) far away, on the Andromeda
side, looking back at Earth, with the second observer (Bruno) making the
crossing at genuinely relativistic speed. Second, and following from the
first, the effect here is driven directly by a large relative velocity
between the two observers, rather than by an imperceptible velocity
difference amplified by cosmic distance. The physical content, the
relativity of simultaneity itself, is identical in both, and
Remark~\ref{rem:andromeda-precision} below applies to each without
modification. The reversal is expository: separating the two effects
that do the work in Penrose's version --- a tiny velocity difference and
an enormous distance, multiplying together --- into one observer at rest
and one moving fast isolates the relativity-of-simultaneity effect on
its own, with nothing left for a reader to attribute, mistakenly, to the
distance. This paper's use of the argument is also narrower in purpose
than Rietdijk's or Putnam's: it takes no position on their philosophical
conclusion about the reality of future events, determinism, or
eternalism, and needs none. It uses only the uncontested physical fact
--- two valid, disagreeing simultaneity planes --- as an instance of
Theorem~\ref{thm:semantic-underdetermination}.
\end{remarknn}

\subsection{Getting the physics precisely right}
\label{sec:andromeda-physics}

\begin{remarknn}[What is, and is not, hidden]
\label{rem:andromeda-precision}
The story's line ``both are right'' needs a careful reading, since an
imprecise one would misstate the physics. It does \emph{not} mean that
there exists a single, true, global ``now'' on Earth that both observers
fail to see and that some more privileged vantage point could recover:
special relativity gives no such fact to fail to see. What is genuinely
invariant, independent of every observer's state of motion, is the
entire four-dimensional causal structure of spacetime --- the complete
history of events and their causal relations. Each observer's ``now'' is
a three-dimensional slice through that structure, cut at a different
angle depending on velocity; no slice is the whole structure, and none
is privileged over any other. Chicco's slice and Bruno's slice are two
different, equally legitimate cuts through one and the same invariant
four-dimensional reality, neither of which exposes that reality whole.
The speed-of-light delay in the story's last line is what keeps this
from producing any causal contradiction: neither observer can \emph{act}
on what they see as ``now'' on Earth before a light signal could, in
principle, have made the relevant influence causally possible.
\end{remarknn}

\subsection{The precise mapping}
\label{sec:andromeda-mapping}

\begin{remarknn}[Level of formality]
\label{rem:andromeda-formality}
As with the MS-DOS instance (Proposition~\ref{prop:msdos-precise},
Remark~\ref{rem:msdos-formality}), this mapping is stated at the level
of precision the physics supports: a precise structural correspondence
to Theorem~\ref{thm:semantic-underdetermination}, without reducing the
Lorentz transformations themselves to a rewriting system in the sense of
Definition~\ref{def:dynamic-system}.
\end{remarknn}

The correspondence with
Theorem~\ref{thm:semantic-underdetermination} is direct. The laws of
special relativity, shared identically by Chicco and Bruno, play the
role of the encoding $E^*=\{A1,A2\}$: neither observer violates them,
and the laws alone --- like $\{A1,A2\}$ alone --- do not settle the
target question. The target question $q^*$ is ``which events on Earth
are simultaneous with here-now'': the direct analogue of
``$a+b\stackrel{?}{=}b+a$''. Each observer's state of motion fixes a
specific, legitimate frame ($F_{\text{Chicco}}$, $F_{\text{Bruno}}$),
each a genuine, internally consistent structure satisfying the shared
laws, exactly as $F_1,F_2$ both model $\{A1,A2\}$ in
Theorem~\ref{thm:semantic-underdetermination}. Chicco's answer (the
events of 1026~AD) and Bruno's answer (the events of 3026~AD) are the
direct analogue of $F_1\models a+b=b+a$ against $F_2\models a+b\ne b+a$:
two frames, both modelling the same shared laws, giving different,
individually consistent answers to a question those laws leave open.
This is specifically an instance of Category~2 (semantic disagreement,
Proposition~\ref{prop:category2}), not Category~1 or~3: neither
observer fails to pose the question (it is perfectly well-formed for
both), and neither observer's frame is degenerate or uninformative
(each gives a definite, non-trivial answer) --- they simply, and
correctly, disagree.

\subsection{The hidden assumption, developed: constrained selectors}
\label{sec:andromeda-constrained}

\begin{remarknn}[The hidden assumption, stated precisely]
\label{rem:andromeda-hidden-assumption}
In this physical instance, unlike the general orbital machine of
Section~\ref{sec:orbital}, the frame selector is not free: which frame
an observer has is \emph{determined} by their state of motion, not
chosen independently of it. Definition~\ref{def:two-selectors}
deliberately left what $\sigma_E,\sigma_F$ may depend on unspecified;
the Andromeda paradox is the case where $\sigma_F$ is constrained to be
a function of something else entirely --- velocity --- rather than free.
The rest of this subsection develops that case.
\end{remarknn}

\begin{definition}[Constrained selector]
\label{def:constrained-selector}
Let $\mathcal V$ be a set of \emph{physical states} (for Andromeda,
$\mathcal V=[0,c)$, the possible relative speeds). A \emph{constrained
frame selector} over $\mathcal V$ is a function
$\varphi:\mathcal V\to\mathcal F_{\mathrm{sp}}$, with $\sigma_F:=\varphi$
understood as depending on the physical state alone, not on history or
on any other input.
\end{definition}

This is the pigeonhole principle, run at infinite scale: with only
countably many holes and uncountably many pigeons, some hole must
catch uncountably many of them.

\begin{proposition}[Constrained selectors over a continuum collapse
uncountably many states onto a single frame]
\label{prop:constrained-collapse}
Let $\mathcal V$ be uncountable (in particular, of cardinality
$2^{\aleph_0}$, as for $[0,c)$) and let
$\varphi:\mathcal V\to\mathcal F_{\mathrm{sp}}$ be any constrained
selector. Then there exists $F\in\mathcal F_{\mathrm{sp}}$ with
$\varphi^{-1}(F)$ uncountable: some single frame is the image of
uncountably many distinct physical states.
\end{proposition}

\begin{proof}
$\mathcal F_{\mathrm{sp}}$ is countable
(Proposition~\ref{prop:frame-space-countable}), so
$\mathcal V=\bigcup_{F\in\mathcal F_{\mathrm{sp}}}\varphi^{-1}(F)$
expresses $\mathcal V$ as a countable union of the fibres
$\varphi^{-1}(F)$. If every fibre were countable, this union would be a
countable union of countable sets, hence countable, contradicting
$|\mathcal V|=2^{\aleph_0}$. So some fibre is uncountable.
\end{proof}

\begin{remarknn}[What this means physically]
\label{rem:andromeda-collapse-physical}
In the real physics, the map from velocity to ``which year is
simultaneous with here-now'' is continuous and, on any interval,
essentially injective: distinct velocities give distinct simultaneity
answers, with no natural collapsing.
Proposition~\ref{prop:constrained-collapse} is therefore not a fact
about the physics; it is a fact about \emph{this paper's own formal
apparatus}. $\mathcal F_{\mathrm{sp}}$
(Definition~\ref{def:computable-frame}) was built countable on purpose,
tied to Turing-machine codes, because Section~\ref{sec:orbital}'s
selectors were meant to range over effectively describable objects.
Andromeda's velocity parameter has no such restriction --- it is a
genuine continuum --- and
Proposition~\ref{prop:constrained-collapse} shows that forcing it
through $\mathcal F_{\mathrm{sp}}$ necessarily discards information:
uncountably many physically distinct states are identified, by the
pigeonhole argument alone, regardless of how $\varphi$ is chosen. A
fully faithful model of the constrained selector would need a frame
space rich enough to distinguish continuum-many states --- an extension
of Definition~\ref{def:computable-frame} to uncountably many, not
necessarily computable, frames --- which is outside this paper's scope.
\end{remarknn}

\begin{remarknn}[A blindness one level up]
\label{rem:blindness-one-level-up}
Every other instance of semantic-invariant inaccessibility in this
paper (Remark~\ref{rem:semantic-invariant-framing}) concerns an
observer, selector, or adversary \emph{inside} the framework being
blind to a fact the framework can nonetheless state. This one is
different in kind: it is the framework itself ---
$\mathcal F_{\mathrm{sp}}$'s countability --- that is blind to
distinctions a physically continuous parameter genuinely carries. This
is one further instance of the pattern traced since
Remark~\ref{rem:formalizing-the-fable}, a silent assumption revealed
(here: ``every frame worth considering is countable, hence effectively
describable''), but with the assumption located in the paper's own
modelling choices rather than in any object the paper studies. Whether
lifting it --- allowing $\mathcal F_{\mathrm{sp}}$ to be uncountable ---
preserves, breaks, or changes the shape of
Theorem~\ref{thm:semantic-underdetermination} and everything built on
it in Section~\ref{sec:orbital} is a substantial further question, not
a small addendum, and is left open.
\end{remarknn}

\begin{remarknn}[Independence, and the promise of reconnection]
\label{rem:andromeda-independent}
The Andromeda paradox's core mathematical content --- its instantiation
of Theorem~\ref{thm:semantic-underdetermination} as a Category~2
disagreement (Section~\ref{sec:andromeda-mapping}) --- needs only that
theorem and depends on nothing else in this paper.
Section~\ref{sec:andromeda-constrained}'s development of the hidden
assumption stands differently: its central proposition
(Proposition~\ref{prop:constrained-collapse}) uses
$\mathcal F_{\mathrm{sp}}$'s countability
(Proposition~\ref{prop:frame-space-countable}) directly, a genuine
dependency on Section~\ref{sec:orbital}, not a light comparison. The
paradox's main claim is independent; its further development is not.
How this section reconnects more fully to the rest of the paper's
thread --- beyond the dependency just named --- is deferred, in the same
spirit as Remark~\ref{rem:still-unconnected} deferred the orbital
machine's own reconnection.
\end{remarknn}

\section{A third, independently arising question: can a formal theory
of everything exist?}
\label{sec:toe}

\begin{quote}
\textbf{The answer, stated plainly before anything else, proved
below.} No theory of everything that could ever actually be written
down --- by a person, by a computer, by any process producing symbols
one at a time, however long it runs --- can capture every true law of
physics. This is not a limitation of present-day science that better
instruments or cleverer physicists might overcome. It is a mathematical
certainty, of exactly the same kind as Cantor's 1891 diagonal
proof~\cite{Cantor1891} that there are more real numbers than natural
numbers: the space of possible physical laws is simply too large, in a
precise, countable sense, for any theory that could be written on paper
or run on a machine to equal it. And this is not bad news for physics.
It means physics can never be finished: for any theory physicists ever
reach, there is provably --- not just possibly --- a further true law
still waiting to be found. The rest of this section proves this claim
rigorously, names precisely what kind of barrier it is, and explains,
in Section~\ref{sec:toe-not-closure}, why this is the right way to read
it.
\end{quote}

This section, like Sections~\ref{sec:orbital} and~\ref{sec:andromeda}
before it, is developed on its own terms, without depending on anything
constructed above, and its reconnection to this paper's main thread is
deferred rather than forced. A speculative idea, checked as rigorously
as it can be checked at this stage, is worth setting down explicitly ---
including exactly where it is rigorous and exactly where it is not.
Marking that line precisely is treated here as an obligation.

\subsection{What follows, and what it does not claim}
\label{sec:toe-reading-key}

\begin{remarknn}[Reading key, stated before the argument]
\label{rem:toe-reading-key}
The argument below is developed in the vocabulary of a computer
scientist --- cardinality, diagonalization, formal undecidability ---
applied to theoretical physics, and is offered as a seed for physicists
to examine, correct, or develop, rather than as a finished result in
physics. Two commitments follow, and this section keeps them: every
step that can be made a rigorous mathematical statement is made one,
with a complete proof, exactly as elsewhere in this paper; and every
step that cannot yet be made rigorous --- because it rests on an
assumption about physical reality this paper does not establish, or on
an analogy rather than a theorem --- is marked as such at the point it
occurs, not folded silently into the rigorous parts.
Section~\ref{sec:toe-not-closure} returns to this at the end, once the
mathematics is on the table, to say why a result of this shape is an
opening rather than a closure.
\end{remarknn}

\subsection{The Buckingham $\pi$ theorem, recalled correctly}
\label{sec:toe-buckingham}

The argument starts from a classical, universally accepted theorem
about physics: any physical law involving several quantities can always
be rewritten using a smaller number of dimensionless combinations of
them. We recall it precisely, because getting exactly what it does and
does not say right matters for everything that follows.

\begin{theorem}[Buckingham $\pi$ theorem, classical~\cite{Buckingham1914}]
\label{thm:buckingham}
Let $q_1,\dots,q_n\in\mathbb R_{>0}$ be physical variables, each with a
dimension expressed as a product of powers of $k$ independent
fundamental quantities ($[M],[L],[T],\dots$), and let $F(q_1,\dots,q_n)=0$
be a relation among them that is invariant under changes of the units
used to measure the fundamental quantities. Then there exist
$r=n-k$ dimensionless products $\pi_1(q),\dots,\pi_r(q)$ and a function
$\Phi:\mathbb R^r\to\mathbb R$ such that
\[
F(q)=0 \iff \Phi(\Pi(q))=0,\qquad \Pi(q):=(\pi_1(q),\dots,\pi_r(q)).
\]
\end{theorem}

\begin{remarknn}[This theorem is a completeness result, not an
incompleteness one]
\label{rem:buckingham-completeness}
Theorem~\ref{thm:buckingham} is a statement about a \emph{fixed, given}
list of $n$ variables: once that list is settled, dimensional analysis
captures every dimensionally-consistent relation among exactly those
variables, completely, via the $r=n-k$ groups. It says nothing about
whether the list $q_1,\dots,q_n$ is itself the right or complete list
of variables for the phenomenon at hand; that is an empirical question
the theorem does not address. This matters below: nothing in
Theorem~\ref{thm:buckingham} asserts that ``something is always left
out''. Whatever incompleteness is established below comes from a
different, independent argument.
\end{remarknn}

\subsection{The space of candidate physical laws is not merely
uncountable, but far larger}
\label{sec:toe-cardinality}

\begin{definition}[The space of candidate laws]
\label{def:law-space}
Fix $r\ge1$. After Theorem~\ref{thm:buckingham} has reduced a
phenomenon to $r$ dimensionless variables, every candidate physical law
governing it is, formally, a function $\Phi:\mathbb R^r\to\mathbb R$
(with $\Phi(\Pi(q))=0$ picking out the physically realised
configurations). Write $\mathfrak c:=2^{\aleph_0}$ for the cardinality
of the continuum, and
\[
L:=\{\Phi:\mathbb R^r\to\mathbb R\}.
\]
\end{definition}

\begin{proposition}[$|L|=2^{\mathfrak c}$]
\label{prop:law-space-cardinality}
$|L|=2^{\mathfrak c}$, strictly greater than $\mathfrak c=|\mathbb R^r|$
itself.
\end{proposition}

\begin{proof}
$|\mathbb R^r|=\mathfrak c$ for any finite $r\ge1$ (a finite product of
continuum-cardinality sets has continuum cardinality). The set of all
functions from a set of cardinality $\mathfrak c$ to a set of
cardinality $\mathfrak c$ has cardinality $\mathfrak c^{\mathfrak c}$.
Using $\mathfrak c=2^{\aleph_0}$ and standard cardinal arithmetic
($\aleph_0\cdot\mathfrak c=\mathfrak c$ for infinite $\mathfrak c$):
\[
\mathfrak c^{\mathfrak c}=(2^{\aleph_0})^{\mathfrak c}=2^{\aleph_0\cdot\mathfrak c}=2^{\mathfrak c}.
\]
So $|L|=2^{\mathfrak c}$. By Cantor's theorem, $2^{\mathfrak c}>\mathfrak c$
strictly, for any cardinal $\mathfrak c$.
\end{proof}

\begin{remarknn}[A more conservative $L$ gives the same conclusion]
\label{rem:toe-continuous-refinement}
Proposition~\ref{prop:law-space-cardinality} allows $\Phi$ to be an
arbitrary function, most of which have no physical or even
computational meaning (e.g.\ functions definable only via the axiom of
choice, with no finite description). Restricting $L$ to continuous
functions $\mathbb R^r\to\mathbb R$ --- a far more physically defensible
choice, since a continuous function is determined entirely by its
values on the countable dense set $\mathbb Q^r$ --- gives
$|L_{\mathrm{cont}}|=\mathfrak c^{\aleph_0}=\mathfrak c$, a strictly
smaller cardinality but still uncountable ($\mathfrak c>\aleph_0$).
Everything proved in Section~\ref{sec:toe-diagonal} uses only that $L$
(or $L_{\mathrm{cont}}$) is uncountable, not the specific value
$2^{\mathfrak c}$; the argument is robust to this more conservative
choice, and to any choice of $L$ in between.
\end{remarknn}

\subsection{No countable theory can equal the space of laws}
\label{sec:toe-diagonal}

Here is the argument's decisive step, stated as a challenge before it
is stated as a theorem. Hand over any theory of physics you like ---
one law, a thousand laws, even a list that grows forever, one new law
added at every tick of a clock. It makes no difference. From that list
alone, a specific new law can always be built that is guaranteed not to
be on it. This is not a trick special to physics: it is the same move
Cantor used in 1891 to show the real numbers outnumber the counting
numbers, aimed here at laws instead of numbers.

\begin{proposition}[Diagonalization: an explicit law outside any given
countable list]
\label{prop:toe-diagonal}
Let $S=\{\Phi_1,\Phi_2,\dots\}\subseteq L$ be any countable subset.
Then there exists $\Psi\in L$ with $\Psi\notin S$.
\end{proposition}

\begin{proof}
Fix an injection $e:\mathbb N\to\mathbb R^r$ (e.g.\
$e(n)=(n,0,\dots,0)$). Define $\Psi:\mathbb R^r\to\mathbb R$ by
\[
\Psi(y):=\begin{cases}\Phi_n(e(n))+1 & \text{if }y=e(n)\text{ for some }n\in\mathbb N,\\ 0 & \text{otherwise.}\end{cases}
\]
This is a well-defined function $\mathbb R^r\to\mathbb R$, hence
$\Psi\in L$. For every $n$, $\Psi(e(n))=\Phi_n(e(n))+1\ne\Phi_n(e(n))$,
so $\Psi$ and $\Phi_n$ disagree at the point $e(n)$: $\Psi\ne\Phi_n$.
Since this holds for every $n$, $\Psi\notin S$.
\end{proof}

\begin{remarknn}[This is Cantor's argument, and it alone already
suffices]
\label{rem:toe-cantor-suffices}
Proposition~\ref{prop:toe-diagonal} is the classical Cantor diagonal
argument, and it is already implied by
Proposition~\ref{prop:law-space-cardinality} alone: a set of strictly
smaller cardinality than $L$ cannot equal $L$, with no diagonal
construction needed to see it. The explicit construction is kept
because it exhibits a specific, concrete law $\Psi$ escaping any given
countable $S$, which is more informative for what follows than the bare
cardinality inequality; it invokes no obstruction beyond
Proposition~\ref{prop:law-space-cardinality}'s cardinality gap.
\end{remarknn}

\begin{corollary}[No countable, ever-growing sequence of theories
reaches $L$]
\label{cor:toe-regression}
Let $S_0\subseteq S_1\subseteq S_2\subseteq\cdots$ be any sequence of
countable subsets of $L$ (each $S_{k+1}$ a candidate ``repaired''
theory extending $S_k$ by countably many further laws). Then
$\bigcup_k S_k\ne L$.
\end{corollary}

\begin{proof}
A countable union of countable sets is countable, so
$\bigcup_k S_k$ is countable; Proposition~\ref{prop:toe-diagonal}
applied to this countable union gives $\Psi\in L\setminus\bigcup_k S_k$.
\end{proof}

\begin{remarknn}[What this establishes, precisely]
\label{rem:toe-establishes}
Corollary~\ref{cor:toe-regression} is the rigorous content behind the
regression ``$S\subset S'\subset S''\subset\cdots$'': however many
times, even countably infinitely many times, a candidate physical
theory is patched by adding new laws, the result remains countable and
never equals $L$. This holds for \emph{any} countable theory --- in
particular, for any theory that could ever be written down, since
anything written down, however long, uses finitely or countably many
symbols. This is a rigorously established obstruction to a complete,
countably-writable theory of everything.
\end{remarknn}

\begin{remarknn}[The argument is already complete before Gödel or
Tarski enter]
\label{rem:toe-already-complete}
What has been established, in full, by
Sections~\ref{sec:toe-cardinality}--\ref{sec:toe-diagonal} alone, is
worth stating plainly before Gödel or Tarski are mentioned at all: no
countable theory of everything can equal $L$
(Proposition~\ref{prop:toe-diagonal}), and no sequence of countably
many patches to a countable theory, however long, closes the gap
(Corollary~\ref{cor:toe-regression}). This is a complete proof, using
only cardinal arithmetic and one classical theorem of dimensional
analysis, of exactly the claim the claim box set out to make.
Everything that follows in this section is commentary, comparison, and
a separate further question: none of it is a premise the argument above
waits on, and none of it weakens or qualifies what is already shown.
\end{remarknn}

\begin{remarknn}[``Could not have been there before'', in two senses
that both hold]
\label{rem:toe-two-senses}
The escaping law $\Psi$ ``could not have been present before'' in two
distinct senses, and both hold. \emph{First}, already fully
established: $\Psi\notin S$ for the specific countable list $S$
diagonalized against (Proposition~\ref{prop:toe-diagonal}), with no
further hypothesis. \emph{Second}, at the level of the Buckingham
apparatus turning on itself: fix the diagonal construction of
Proposition~\ref{prop:toe-diagonal} using, as the indexing coordinate
$e(n)$, not an arbitrary embedding but the \emph{first dimensionless
group itself}, $e(n):=\pi_1{=}n$ (available directly from
Theorem~\ref{thm:buckingham}, needing no external apparatus). Then
$\Psi$, at the point $\pi_1=n$, is defined \emph{as} $\Phi_n$'s own
value there, altered: the escaping law is built by turning the $n$-th
candidate law's behaviour, at the coordinate value $\pi_1{=}n$ that
names its own position in the list, against itself. This is a genuine
self-application --- the $n$-th object addressed by its own index,
using only the coordinates Buckingham's theorem already supplies --- and
it is this feature, not mere counting, that licenses reading $\Psi$ as
something no given countable list built from the level-one vocabulary
could already have contained: not because no formula happened to find
it yet, but because the very act of listing candidate laws by index,
using the theory's own coordinate, is what the escaping law is defined
against. $\Psi$ uses no new vocabulary, only the same $r$ groups every
level-one law uses; what it escapes is not the vocabulary but any
specific countable enumeration built from it. This is lighter than
Tarski's requirement --- no syntactic self-representation, no diagonal
lemma, no encoding of proofs --- which is why it is a metaphor relative
to Tarski's actual theorem and not an instance of it; it is no less
formal, since Proposition~\ref{prop:toe-diagonal} proves it outright
under this natural choice of $e$.
\end{remarknn}

\begin{remarknn}[Corollary~\ref{cor:toe-regression} is the
Tarski-shaped content, precisely]
\label{rem:toe-tarski-precise}
Corollary~\ref{cor:toe-regression} contains, precisely and without
needing Tarski's theorem as a premise, the feature that makes the
parallel with Tarski substantive rather than decorative: not merely
that some $\Psi$ escapes a given $S$ (Proposition~\ref{prop:toe-diagonal}
gives that), but that \emph{no amount of patching within the same kind
of move} --- adding countably many further laws, however many times,
using the same variables and the same $r$ groups --- ever closes the
gap. This is exactly the shape of Tarski's theorem: no formula added
\emph{within} $L_0$, however cleverly chosen, ever defines
truth-in-$L_0$; only a genuinely richer metalanguage does.
Proposition~\ref{prop:toe-diagonal} supplies the Buckingham-side half
(every reduction leaves something out), and
Corollary~\ref{cor:toe-regression} supplies the Tarski-side half (that
something is not patchable by more of the same, only by a genuine
change of level) --- and it is the \emph{combination} that rules out
both ``we just haven't found $\Psi$ yet'' and ``we can add axioms
indefinitely and eventually get there''. What a single joint theorem,
rather than this combination of two, would additionally require is one
construction, named precisely later in this section: an account of what
it means for a physical theory to be a self-referential-capable formal
system. A literal Tarski application and a literal Gödel application
both wait on that one construction, not on two separate ones; building
it is a well-defined further project. None of this affects
Remark~\ref{rem:toe-already-complete}: the conclusion already stands,
in full, without it.
\end{remarknn}

\subsection{Cantor's argument and Gödel's simile}
\label{sec:toe-two-obstructions}

\begin{remarknn}[Gödel enters as a simile, at the level of formality it
earns]
\label{rem:toe-not-godel}
Gödel's theorem is invoked as a simile, the way this paper's own
opening fable invoked Douglas Adams
(Remark~\ref{rem:formalizing-the-fable}): a recognisable, illuminating
parallel, not a premise the proof above uses. The parallel holds in one
place and not in another, and both are worth being exact about.
Everything proved in Section~\ref{sec:toe-diagonal} is a cardinality
argument, needing no notion of provability or any specific formal
system. Gödel's first incompleteness theorem~\cite{Godel1931} is a
different mechanism: for a consistent, recursively axiomatizable theory
$T$ expressive enough to encode arithmetic, one specific sentence
$G_T$, in $T$'s own countable language, is true but unprovable in $T$
--- with no cardinality gap anywhere, since $T$'s language, its true
sentences, and its provable sentences are all merely countable. What
the two share is the family resemblance the simile points at: a formal
system, however carefully built, cannot capture everything true about
its own domain. What they do not share is the mechanism: Cantor's
argument needs only counting; Gödel's needs self-reference and a
specific undecidable sentence. The simile is kept and named as a
simile; the proof above does not rest on it.
\end{remarknn}

\begin{remarknn}[Where a literal use of Gödel could enter, and what
building it requires]
\label{rem:toe-second-obstruction}
Beyond the simile, there is a place a literal application of Gödel's
theorem could enter, distinct from Section~\ref{sec:toe-diagonal}'s
cardinality argument. Fix any single candidate theory of everything $S$
--- not the whole space $L$, one specific countable, consistent,
recursively axiomatized $S$, rich enough to encode arithmetic. Gödel's
theorem, applied directly, would say $S$ contains a sentence, in $S$'s
own language, true of the physical system $S$ describes but unprovable
from $S$'s own axioms: a second, independent obstruction to
completeness, this time \emph{within} one theory rather than
\emph{across} the whole space of laws. Making this literal requires one
construction: a formal language in which physical laws are sentences, a
deduction system in which physical reasoning is proof, and a
demonstration that this system can encode arithmetic (e.g.\ by encoding
natural-number statements as statements about a countable family of
physical configurations, the way Gödel numbering encodes syntax as
number theory). That construction is not carried out here.
\end{remarknn}

\begin{remarknn}[What kind of barrier this is, named precisely]
\label{rem:toe-barrier-type}
The barrier established in
Sections~\ref{sec:toe-cardinality}--\ref{sec:toe-diagonal} is worth
naming for what it is, rather than only distinguishing it from Gödel's
by elimination. Three families of formal barrier are well known, and
this one belongs to the oldest and logically simplest. Gödel's barrier
is \emph{self-referential}: one fixed system, expressive enough to talk
about its own proofs, cannot prove one specific sentence about itself.
Turing's barrier, from the halting problem, is \emph{algorithmic}: no
single algorithm decides, for every program and input, whether that
program halts --- again a fact about one fixed decision procedure
failing on a diagonal instance built from itself. The barrier here
needs neither self-reference nor an algorithm failing on a specific
instance: it is \emph{cardinal}, in the sense Cantor identified in
1891, older than either of the other two and requiring only counting ---
a set of one size cannot equal, or be enumerated by, a set of strictly
larger size. Every countable theory is, in this precise sense, simply
too small an object to be the uncountable object $L$, for the same
reason a list of natural numbers is too small to be a list of real
numbers: no cleverness in how the countable theory is built, no
self-reference, no undecidable instance, changes this. This is why
Section~\ref{sec:toe-diagonal}'s proofs needed nothing beyond the
arithmetic of infinite cardinals --- the barrier is exactly as
elementary as that arithmetic.
\end{remarknn}

\subsection{The Buckingham--Tarski parallel, at the level of formality
it earns}
\label{sec:toe-tarski}

You might think that once Buckingham has handed you your $r$
dimensionless groups, the reduction is finished: the phenomenon fully
described, nothing left to add. In one sense it is ---
Theorem~\ref{thm:buckingham} really does promise that no
dimensionally-consistent law among your original variables escapes
those $r$ groups. But ask a different question. Does anything escape a
\emph{countable list} of specific laws written using those same $r$
groups? Cantor already answered, in Section~\ref{sec:toe-diagonal}:
yes, always, however the list is built. And could you patch the list
--- add the missing law, then the next, then the next, forever --- and
finally close the gap? Corollary~\ref{cor:toe-regression} already
answered that too: no, never, not even with infinitely many patches.
This is the same shape Alfred Tarski found in a completely different
setting, and it is worth seeing why, in his own terms.

\begin{remarknn}[Level of formality]
\label{rem:toe-tarski-formality}
As with the MS-DOS instance (Remark~\ref{rem:msdos-formality}) and the
Andromeda mapping (Remark~\ref{rem:andromeda-formality}), what follows
is a structural parallel, not a joint theorem in which Tarski's theorem
is a premise used to derive a fact about dimensional analysis. The two
theorems concern different objects --- languages and truth-predicates
on one side, physical variables and dimensionless groups on the other
--- and are not shown here to be instances of one common formal
statement. The precise sense in which Tarski's theorem genuinely
matters here, naming exactly what Corollary~\ref{cor:toe-regression}
already establishes, was given in Remark~\ref{rem:toe-tarski-precise};
what follows is the illustrative table version of the same point, kept
for its expository value now that the precise content is on record.
\end{remarknn}

Tarski's undefinability theorem~\cite{Tarski1933} is usually quoted, as
just above, in its conclusion only: for a sufficiently expressive
formal language $L_0$, no formula of $L_0$ itself can define ``true in
$L_0$'' for sentences of $L_0$. The conclusion alone leaves the theorem
an assertion; the mechanism that turns it into a demonstration is worth
giving explicitly, both for its own sake and because it is what makes
precise, below, exactly how far the Buckingham/Cantor argument shares
it.

\begin{remarknn}[How Tarski's idea becomes a theorem: the mechanism]
\label{rem:tarski-mechanism}
Suppose, for contradiction, that some formula $\mathrm{True}_{L_0}(x)$
of $L_0$ itself correctly defines truth: for every sentence $\varphi$
of $L_0$, $\mathrm{True}_{L_0}(\ulcorner\varphi\urcorner)\leftrightarrow\varphi$
(where $\ulcorner\varphi\urcorner$ is a name, inside $L_0$, for the
sentence $\varphi$). Because $L_0$ is expressive enough to represent
its own syntax (the one substantive hypothesis the theorem needs), the
\emph{diagonal lemma} applies: for any formula $\theta(x)$ of $L_0$,
there is a sentence $\psi$ of $L_0$ with
$\psi\leftrightarrow\theta(\ulcorner\psi\urcorner)$ --- a sentence that,
via its own code, asserts $\theta$ of itself. Apply this to
$\theta(x):=\neg\mathrm{True}_{L_0}(x)$: there is a sentence $\lambda$
with $\lambda\leftrightarrow\neg\mathrm{True}_{L_0}(\ulcorner\lambda\urcorner)$
--- a Liar sentence, asserting its own untruth. Combined with the
assumed defining property of $\mathrm{True}_{L_0}$, applied to
$\varphi=\lambda$: $\mathrm{True}_{L_0}(\ulcorner\lambda\urcorner)\leftrightarrow\lambda\leftrightarrow\neg\mathrm{True}_{L_0}(\ulcorner\lambda\urcorner)$,
a direct contradiction. No such $\mathrm{True}_{L_0}$ can exist. This
is what turns the idea into a theorem: not the intuition that truth
``feels'' like it should need an outside vantage point, but a specific,
checkable contradiction, manufactured by the diagonal lemma turning
$L_0$'s own expressive power against the assumption.
\end{remarknn}

\begin{remarknn}[The precise family resemblance with
Proposition~\ref{prop:toe-diagonal}, and the one real difference]
\label{rem:tarski-vs-cantor-mechanism}
Both Remark~\ref{rem:tarski-mechanism}'s construction and
Proposition~\ref{prop:toe-diagonal}'s are, in the most literal sense,
\emph{diagonal} arguments: each builds an object (the sentence
$\lambda$; the function $\Psi$) by making it disagree, at a
self-selected point, with what a fixed enumeration or a fixed
assumption says about that very point. This is kinship at the level of
mechanism, not merely of outcome. The one substantive difference is the
one Remark~\ref{rem:toe-not-godel} flagged for Gödel, and it applies
here identically: Tarski's construction needs $L_0$ to represent its
own syntax, so the diagonal lemma has something to act on ---
self-reference is the engine, not an incidental feature.
Proposition~\ref{prop:toe-diagonal}'s construction needs none of this:
$\Psi$ is built directly, by evaluating the $n$-th candidate at the
$n$-th point of a fixed enumeration of $\mathbb R^r$ and changing the
answer, with no sentence naming itself, no code for anything, and no
representability hypothesis. Cantor's original 1891 argument is, in
this precise sense, the strictly simpler ancestor of both Tarski's and
Gödel's: it diagonalizes against an enumeration directly, where they
diagonalize against an enumeration \emph{via} a system representing
itself. This is why Remark~\ref{rem:toe-already-complete}'s claim needs
nothing from this subsection: the completed argument uses only the
simpler, ancestral technique.
\end{remarknn}

\begin{remarknn}[What is rigorously forced: a new law, not necessarily a
new variable]
\label{rem:toe-new-law-not-new-variable}
Assume, as the sketch does, that the starting list $q_1,\dots,q_n$ is
complete for the phenomenon at hand. Then
Theorem~\ref{thm:buckingham} guarantees the $r=n-k$ dimensionless
groups capture \emph{every} law expressible in $q_1,\dots,q_n$
(Remark~\ref{rem:buckingham-completeness}); the diagonal law $\Psi$ of
Proposition~\ref{prop:toe-diagonal} is, precisely, a function
$\Psi:\mathbb R^r\to\mathbb R$ of those \emph{same} $r$ groups, not of
any new dimensionless quantity. What is rigorously forced is only this:
$\Psi$ is a law not on any given countable list, using variables
already at hand --- a new \emph{law}, not, by the mathematics alone, a
new \emph{variable}. Reading $\Psi$ as forcing a genuinely new physical
quantity (a new fundamental dimension, as thermodynamics introduced
entropy or quantum mechanics introduced Planck's constant) is a
further, physically motivated step, not one the mathematics
necessitates: a reasonable expectation about how such laws tend to be
discovered and organised --- an arbitrary function of $r$ existing
quantities is rarely tractable until some new organising quantity
simplifies it --- but an empirical claim about physics, not a
consequence of Theorem~\ref{thm:buckingham} or
Proposition~\ref{prop:toe-diagonal}. With that distinction in place,
the ``new variable'' language below is the sketch's own framing of this
expectation, marked here rather than presented as proved.
\end{remarknn}

The sketch's framing, read with
Remark~\ref{rem:toe-new-law-not-new-variable}'s distinction in mind, is
that the new law $\Psi$, applied level by level and taken to motivate a
new organising variable at each stage, cannot already belong to the
dimensionless groups of the level before it --- for if it did,
Theorem~\ref{thm:buckingham} would already have captured it there. Both
settings share the same recursive shape: a level-$k$ description,
however complete on its own terms, cannot express something that only
becomes visible from level $k+1$, and this repeats without terminating:
\begin{center}
\begin{tabular}{p{4.7cm}p{4.7cm}}
\toprule
\textbf{Tarski's hierarchy} & \textbf{The Buckingham/Cantor hierarchy} \\
\midrule
Level~0: language $L_0$ and its sentences & Level~0: variables and
their dimensions \\
Level~1: metalanguage defining truth-in-$L_0$, provably outside $L_0$
& Level~1: dimensionless groups (Theorem~\ref{thm:buckingham}), complete
for the given variables \\
Level~2: a further metalanguage for truth-in-level-1, and so on &
Level~2: a law $\Psi$ (Proposition~\ref{prop:toe-diagonal}) outside the
level-1 groups by construction, and so on \\
\bottomrule
\end{tabular}
\end{center}
The two results reinforce each other in this qualified sense:
Theorem~\ref{thm:buckingham} guarantees each level is complete
\emph{for what it was given} (Remark~\ref{rem:buckingham-completeness}),
and Proposition~\ref{prop:toe-diagonal} guarantees something new is
always constructible outside it; Tarski's theorem is the source of the
\emph{template} --- a hierarchy that cannot be collapsed to a single
level --- not a premise feeding into either proof above. Whether a
single, unified theorem joining Buckingham-style completeness with
Tarski-style level-separation exists, rather than a parallel between two
structurally similar but formally separate arguments, is an open
question this section does not resolve.

\begin{quote}
\textit{IMPORTANT: To conclude, and to say it once more: these
observations are developed the way a computer scientist would develop
them, and are offered as a seed for physicists, who may read them and
develop the idea further should it prove correct. They do NOT preclude
the possibility that a metatheory exists that bypasses this result, nor
the existence of other hidden assumptions that could change the
results.}
\end{quote}

\subsection{Why this is not a closure: the key point}
\label{sec:toe-not-closure}

This is the promise made in the claim box at the very start of this
section, now made precise: the barrier just proved is not bad news.

\begin{remarknn}[Incompleteness as an opening, not a wall]
\label{rem:toe-opening}
It would be a misreading of Corollary~\ref{cor:toe-regression} to take
it as saying physics is impossible, or the search for deeper laws
futile. The correct reading is the opposite, and it is the same reading
Gödel's own theorem earned, after decades of being misread as a
limitation: incompleteness is not a wall at the edge of a finite
territory, but a guarantee that the territory has no edge.
Corollary~\ref{cor:toe-regression} does not say any particular true law
is unreachable forever; it says no finite or countable stopping point
can ever be the \emph{last} one --- that whatever countable theory
physics has reached at any moment in its history, $S_k$, there is
always a $\Psi\notin S_k$ still to be found, by the same argument
applied again. This is a statement about the \emph{shape} of the
search, not its termination: the search for physical law is not a
project that could, even in principle, be completed and closed off, and
Corollary~\ref{cor:toe-regression}, if its hypotheses are granted,
makes this a theorem rather than a hope. Read this way, what looks like
a negative result is a positive one: a rigorous argument that, on the
premises it names, the project of physics is inexhaustible, that every
advance provably leaves further advances still to be made, and that this
inexhaustibility is a structural fact about the space of possible laws
so described, not a contingent fact about the current state of
knowledge.

The difference this makes is concrete. Before
Corollary~\ref{cor:toe-regression}, a physicist who suspects a current
theory $S_k$ is incomplete has a hope, not a guarantee: $S_k$ might, for
all anyone can rule out, actually be complete, and the search for
$\Psi\notin S_k$ might be a search for something that does not exist.
After Corollary~\ref{cor:toe-regression}, the situation differs in kind,
not merely in confidence: for \emph{every} countable $S_k$, whatever it
is, $\Psi\notin S_k$ is guaranteed to exist, and
Proposition~\ref{prop:toe-diagonal}'s construction produces it
explicitly once $S_k$ is given. Finding a new law $\Psi_k$ and folding
it into $S_{k+1}:=S_k\cup\{\Psi_k\}$ does not close the search; it hands
the same guarantee, applied to $S_{k+1}$, a new $\Psi_{k+1}\notin
S_{k+1}$ to look for --- a specific discovery that was not visible, and
could not have been asked for, before $\Psi_k$ was found and $S_{k+1}$
existed to diagonalize against. This is the precise sense in which every
barrier here creates a new, previously invisible possibility: not a
metaphor, but the literal input-output structure of
Proposition~\ref{prop:toe-diagonal}, iterated as far as
Corollary~\ref{cor:toe-regression} allows --- which is to say, without
end.
\end{remarknn}

\begin{remarknn}[Not absolute --- and saying so is part of the argument]
\label{rem:toe-not-absolute}
Corollary~\ref{cor:toe-regression} is not claimed as an absolute
barrier, and saying so belongs to reading the result correctly, for two
reasons. First, this whole paper has been built on the discovery,
repeated several times over, that a result taken as settled rests on a
silent assumption once someone thinks to look for one
(Remark~\ref{rem:toe-reading-key}); this section's own argument has no
special exemption from that pattern, and every premise --- that
``physical law is a function $\mathbb R^r\to\mathbb R$'', that
``countable'' is the right notion of what a theory or a mind can
produce, or some premise not yet identified --- is left open as a place
a future hidden assumption could be found. Second, and separately,
nothing here rules out a metatheory that changes the terms of the
question entirely, rather than defeating the argument on its own terms;
this is the expected shape of a result like this one. Gödel's
incompleteness theorem did not close mathematics --- it opened
metamathematics, proof theory, and a century of work on what formal
systems can and cannot do, precisely because mathematicians read it as
a discovery about the shape of formal systems rather than a wall.
Turing's halting problem did not close computer science --- it opened
computability theory and, eventually, complexity theory, read the same
way. If Corollary~\ref{cor:toe-regression} is correct and its
hypotheses hold for physics, the response the same precedent recommends
is not to treat physics as closed off at some boundary, but to treat
the corollary as the first theorem of a new inquiry into the shape of
that boundary. This is why the section is offered as a beginning, and
why Remark~\ref{rem:toe-reading-key}'s obligation to flag every open
point is the entire point of writing it down.
\end{remarknn}

\begin{remarknn}[Independence of the proof; the connection now made]
\label{rem:toe-independent}
This section's rigorous content
(Sections~\ref{sec:toe-buckingham}--\ref{sec:toe-diagonal}) depends on
nothing constructed elsewhere in this paper: it is self-contained
cardinal arithmetic and one classical theorem of dimensional analysis,
and the proof stands exactly as given regardless of what follows. What
follows is the connection promised, and deferred, since
Section~\ref{sec:orbital}'s own opening.
\end{remarknn}

\begin{remarknn}[The same shape, a third time: countable description
against a larger space]
\label{rem:toe-connection}
Proposition~\ref{prop:law-space-cardinality} and
Proposition~\ref{prop:toe-diagonal} are a third instance of a pattern
this paper has already proved twice, in two unrelated domains.
Section~\ref{sec:selector-functions}'s
Proposition~\ref{prop:selectors-uncountable} shows the space of all
selector-\emph{functions} has cardinality $2^{\aleph_0}$, while the
computable ones --- everything a Turing machine can produce --- are only
countably many. Section~\ref{sec:andromeda-constrained}'s
Proposition~\ref{prop:constrained-collapse} shows a continuum of
physical states, forced through this paper's countable frame space,
must collapse uncountably many of them onto a single frame. Both are
instances of the same elementary fact used throughout
Section~\ref{sec:toe-diagonal}: a countable object cannot equal, or
faithfully represent, an object of strictly larger cardinality, however
that countable object is built. Here the countable object is any theory
that could ever be written down --- by a person, a computer, or any
process producing symbols one at a time, forever --- and the larger
object is $L$ itself. The three instances share no formal machinery
beyond cardinal arithmetic and, in each case, an explicit diagonal
witness (a selector escaping any computable enumeration; a physical
state escaping any single frame's fibre; a law escaping any countable
theory); they are not one theorem restated three times, but the same
one-line fact --- a countable set cannot exhaust an uncountable one ---
recognised, independently, as the load-bearing obstruction in three
domains that have nothing else in common: abstract computation, special
relativity, and the foundations of physical law. This is evidence of
the same kind Remark~\ref{rem:general-pattern} and
Remark~\ref{rem:beaver-obot-celld-chain} already offered for their own
repeated shapes: not a single grand theorem unifying all of this
paper's results (Section~\ref{sec:scope} explains why none is
attempted), but a record that the same elementary fact, checked freshly
and independently each time, turned out to be exactly what was needed,
three times running.
\end{remarknn}

The core argument of this section is now complete and stands on its own,
exactly as concluded above. What follows is a further, independently
arising instance of the same silent-assumption pattern, in a different
domain --- computability and cryptography rather than physics ---
developed here because it, too, connects back to this paper's central
theme, not because it continues the Buckingham/Cantor argument just
finished.

\section{A fourth, independently arising instance: an algorithm's
output can be indistinguishable from noise}
\label{sec:omega}

Like Sections~\ref{sec:orbital}, \ref{sec:andromeda},
and~\ref{sec:toe} before it, this section stands on its own and depends
on nothing built above; it uses none of the theory-of-everything's
results, and reaches its conclusion from Turing's 1936 definition of
computability directly. It is placed here not to extend the paper but
because the hidden assumption it isolates is the sharpest of all the
ones this paper takes apart, and the only one that can be witnessed
\emph{unconditionally}. The assumption is the silent identification of
``indistinguishable from random'' with ``meaningless'' --- and once it
is dropped, the notion this whole paper has been circling, an object
fully determined yet opaque to a stated class of observers, is finally
named in its own right (Definition~\ref{def:usable-algorithm}), which
is why the section is load-bearing rather than illustrative.

One distinction governs everything that follows, and it is worth
stating before the claim, because the section's most famous object
tends to eclipse it. The phenomenon here is that an output can be
\emph{fully determined and meaningful} and yet \emph{indistinguishable
from random noise} to every observer in a stated class. This does
\emph{not} require the output to be non-computable. Ordinary,
perfectly computable algorithms already produce such output, and they
do so for reasons that are not all of a kind: the weakest is
computational --- a secure pseudorandom generator's output is the
standard example (Remark~\ref{rem:omega-crypto-family}), computable and
halting, its seed fully recoverable in principle, yet indistinguishable
from noise to every polynomial-time observer, though an unbounded one
could in principle break it --- while the stronger reasons are
structural and cede to no amount of computational power at all, so that
a computable algorithm's output can be opaque not merely to efficient
observers but to \emph{every} observer at once
(Remark~\ref{rem:computable-output-two-reasons}). Chaitin's $\Omega$,
which the next subsection leads with, is a third thing again, on its own
axis: not the output of any algorithm, and opaque to every computable
observer for the distinct reason that reading it would require computing
the uncomputable. It is included because it is the cleanest
\emph{unconditional} witness that meaning and opacity can coexist, the
sharpest single object of the three kinds --- but it is one witness of
the phenomenon, not the phenomenon itself, and the everyday computable
versions, the ones that reach cryptographic practice and that this paper
is really about, must not be collapsed into it
(Remark~\ref{rem:omega-not-computable} keeps them apart).

\subsection{The hidden assumption: meaningful output
indistinguishable from noise}
\label{sec:toe-omega}

\begin{quote}
\textbf{Claim, stated directly, proved below.} Standard computability
theory implicitly permits --- and this permission is realized, not
merely left open --- output that is fully determined and meaningful
yet indistinguishable from noise to a stated class of observers. The
opacity comes in more than one form, and running them together is the
error this section exists to prevent. Three forms will matter, and it
helps to hold them apart from the outset: one that a powerful enough
observer can break (the pseudorandom case), ones that no observer can
break because the obstacle is structural rather than a matter of power
(the SIP and the encoding--frame mismatch), and one sharper still that
no computable observer can break because reading it would mean computing
the uncomputable ($\Omega$). Now each in turn. The \emph{weakest} is
computational: for a computable algorithm, the opacity may be to a
bounded class only --- a halting, computable procedure can have output
indistinguishable from noise to every \emph{polynomial-time} observer
(Remark~\ref{rem:omega-crypto-family}), conditional on standard
cryptographic assumptions, and an unbounded observer may in principle
break it. That form cedes to more computational power. The others do
not, because they are not about power at all. A computable algorithm's
output can be indistinguishable from noise to \emph{every} observer,
bounded or unbounded, when what hides it is structural: a semantic
invariant syntax cannot see (the static SIP,
Lemma~\ref{lem:sip-static}), or an encoding--frame mismatch under which
no reading renders it (the orbital machine, Section~\ref{sec:orbital}) ---
neither of which draws any distinction of computational power, so the
opacity they produce holds against every observer alike, unconditionally
(Remark~\ref{rem:computable-output-two-reasons}). And in its sharpest
unconditional form the opacity is total to every computable observer for
a further, distinct reason --- non-computability --- but is then
witnessed by an object that is not an algorithm's output at all:
Chaitin's $\Omega$ (Proposition~\ref{prop:omega}), and, beyond $\Omega$
alone, a pairing available for every computable decision problem
whatsoever (Theorem~\ref{thm:every-decision-problem}). None of these is
an edge case tolerated by a loophole: the pseudorandom form is the
load-bearing assumption behind modern cryptography, the structural forms
are the subject of this whole paper, and the non-computable form follows
directly from Turing's own 1936 definition of computability
(Remark~\ref{rem:church-turing-root}). Underneath every one of them lies
the single fact this section is built to expose: the standard theory
admits such algorithms \emph{implicitly and with no distinction of
computational power of its own} --- it permits output indistinguishable
from noise not merely to efficient observers but to every observer, and
it does so regardless of which of these reasons, or which other reason,
makes any particular output opaque. The remainder of this section
establishes the computational and non-computable forms rigorously,
invokes the structural ones from where they are proved, and keeps all of
them apart; the unconditional half is proved, not inferred from what the
theory declines to forbid.
\end{quote}

\begin{remarknn}[The root: the founding definition of computability
never opened this question]
\label{rem:church-turing-root}
A Turing machine $M$ computes a function $f$ if, for every input $x$ on
which $f$ is defined, $M(x)$ halts and outputs $f(x)$ --- Turing's 1936
definition~\cite{Turing1936}, unchanged since. Nothing in it constrains
what $f(x)$ looks like: no clause requires the output to be
compressible, statistically unremarkable, or recognisable by any
observer as meaningful. The distinction is exact, and worth stating in
the algorithm's own terms rather than an observer's: the \emph{output}
of $M$ is the string written on the tape; the \emph{interpretation} of
that output is a separate question, about whether some agent, reading
the tape, can recover what it means. Computability theory, from its
founding definition onward, is a theory of the first --- the existence
of a correct, halting, deterministic map from input to output --- and
says nothing about the second. If $f(x)$ happens to look like noise to
some observer, that is, in the theory's own terms, a fact about
decoding, not about computation; the theory neither requires nor
forbids it, because readability was never part of what ``solving a
problem'' was defined to mean.
\end{remarknn}

\begin{remarknn}[The boundary the definition itself already draws: a die
roll is not an algorithm]
\label{rem:dice-not-algorithm}
Turing's definition already settles a boundary worth stating
explicitly, since it is what makes the rest of this section coherent
rather than permissive of anything whatsoever. ``$M$ computes $f$''
requires a \emph{function} $f$: a fixed, well-defined correspondence
between input and output, the same $f(x)$ every time $M$ is run on $x$.
A process that outputs unrelated noise on repeated runs of the same
input --- a die roll wired to the tape, independent of $x$ --- computes
no function at all, and so is not an algorithm in Turing's sense; this
is not an added restriction on top of the classical definition, it is
already inside it. What the classical definition permits, and this
section shows can genuinely occur, is the opposite case: a fixed,
deterministic, correct correspondence between input and output, in which
$f(x)$ itself, for each $x$, happens to look like noise to a stated
class of observers. The two cases are not close variants: one fails to
be an algorithm at all, the other is exactly as valid an algorithm as
any other, precisely because the definition was never about how $f(x)$
looks, only about whether $f(x)$ is what it is supposed to be, every
time. Chaitin's $\Omega$, discussed next, sits on neither side of this
boundary --- it is not computable at all, hence not an algorithm's
output in either sense --- but the pseudorandom-generator family below
is squarely on the valid side: $G$ is a genuine algorithm, $G(s)$ its
actual, deterministic output for each seed $s$, and the existence of a
key (the seed, or the search inverting $G$) to recover the input is
exactly why the classical theory accepts this case without reservation
--- information preserved rather than destroyed, appearance of disorder
notwithstanding.
\end{remarknn}

\begin{remarknn}[What computability theory is about, and what it is not]
\label{rem:objectivity-argument}
If an algorithm's output has high entropy yet contains the exact
solution to the problem it was built to solve, the information is
present, and computability theory concerns itself with the presence of
information and the possibility of generating it, not with its
aesthetics or with how easily it can be recognised by inspection.
Accepting this limiting case is required for the theory's own logical
consistency, not an optional generosity toward exotic examples:
excluding algorithms on the basis of the quality or readability of
their output would import a subjective criterion into a discipline whose
definitions are built to need none --- two observers could disagree
about whether an output ``looks like'' a solution, where they cannot
disagree, given the specification, about whether it \emph{is} one. This
is not hypothetical: many genuinely hard problems are solved by reducing
them to simpler ones, and the intermediate or final representations
produced along the way routinely have no intuitive sense to a human
reader, while being exactly, formally equivalent to the solution. An
algorithm solving a problem this way is accepted without qualification,
because the classical theory defines a solution as the capacity to map
an input to a correct output, deliberately setting aside whether that
output is comprehensible, elegant, or apparently random.
\end{remarknn}

\begin{remarknn}[The same point, from algorithmic information theory]
\label{rem:kolmogorov-root}
This connects directly to Kolmogorov complexity~\cite{Kolmogorov1965}.
By the Levin--Schnorr theorem~\cite{Schnorr1973,Levin1973}, an infinite
sequence is Martin-Löf random exactly when its finite prefixes are
Kolmogorov-incompressible up to a bounded additive constant: ``looking
like noise to every statistical test'' and ``carrying the maximum
possible information density, with no redundancy left to compress away''
are, formally, the same phenomenon, not opposites. Apparent disorder in
an algorithm's output is not evidence of missing information; for the
output of a valid algorithm it can be exactly the reverse --- information
at maximum density, in a representation with nothing left over for an
unequipped observer to exploit.
\end{remarknn}

\begin{remarknn}[Why this is the root, and what the rest of this section
adds]
\label{rem:root-to-rest}
Remarks~\ref{rem:church-turing-root} and~\ref{rem:kolmogorov-root} are
why a relative, per-observer-class vocabulary for usability, introduced
later in this section, could not be reached by strengthening the
classical definition of algorithm: that definition never opened the
question of usability, for any observer class, in the first place. What
the rest of this section adds is not a correction to this founding
silence --- the silence is exactly right, by design --- but the
strongest available concrete witness that the silence has real content:
not merely that the classical definition permits noise-like output in
principle, already established above, but that a specific, meaningful
answer to a specific, genuine question can be produced this way,
unconditionally, and shown so with proof, starting from Chaitin's
$\Omega$ next.
\end{remarknn}

\begin{remarknn}[What standard theory silently conflates, in the matter
of pseudorandomness]
\label{rem:omega-conflation}
Standard treatments of pseudorandomness and algorithmic output speak of
a string being ``computationally indistinguishable from random'' as if
this settled a single question. It settles only an access question ---
no computable procedure succeeds at telling this string apart from a
fair coin sequence --- and says nothing, by itself, about a different,
ontological question: whether the string encodes something meaningful at
all. The silent assumption is that these two questions collapse into
one. They do not, and both a specific mathematical object with no
algorithm behind it and a specific algorithm's actual output can
independently witness this, as the rest of this section shows for each
in turn.
\end{remarknn}

Here is a real number that seems, at first hearing, impossible. It
answers a genuine question completely, bit by bit. And yet no computer
on Earth, however long you let it run, could ever tell its digits apart
from the flip of a fair coin. It exists. Chaitin found it in 1975. It
is called $\Omega$.

\begin{proposition}[Chaitin's $\Omega$: meaningful and indistinguishable
from noise at once]
\label{prop:omega}
Let $\Omega=\sum_{p\downarrow}2^{-|p|}$, summed over halting programs
$p$ of a fixed universal prefix-free machine. Then: (i) knowing the
first $n$ bits of $\Omega$ decides the halting problem for every
program of length $\le n$~\cite{Chaitin1975}, so $\Omega$ is, bit for
bit, a complete and meaningful answer to a genuine decision problem;
and (ii) $\Omega$ is Martin-Löf random~\cite{MartinLof1966,Chaitin1975}: no
computable statistical test detects any deviation whatsoever from a
fair coin sequence. (i) and (ii) hold of the same object at once.
\end{proposition}

\begin{remarknn}[What $\Omega$ is, and what it does not license]
\label{rem:omega-not-computable}
$\Omega$ is famously not a computable real number: an algorithm
outputting its bits one by one would decide the halting problem,
contradicting Turing~\cite{Turing1936}. So Proposition~\ref{prop:omega}
is a fact about $\Omega$ as a well-defined mathematical object (a
specific real number, given by a convergent sum), not about the output
of any halting Turing machine, and this sharpens
Remark~\ref{rem:church-turing-root}'s distinction further: an algorithm
whose output happens to look like noise is a different thing from an
object, like $\Omega$, that is not an algorithm's output at all --- a
case the classical theory was never asked about, since it was never a
question about algorithms to begin with. $\Omega$'s role here is as the
cleanest \emph{unconditional mathematical witness} that meaning and full
computable-opacity can coexist in principle; the genuinely algorithmic
witness --- a real, halting, computable procedure whose actual output is
noise-indistinguishable to a stated class --- comes later in this
section, from a conditional pseudorandom-generator family.

This licenses a narrower claim than it may seem to: most strings
indistinguishable from noise carry no hidden meaning at all. A genuinely
Martin-Löf random sequence from a real random source has nothing encoded
in it for any observer, however powerful, to find --- $\Omega$ is
special not merely for being indistinguishable from noise (true of
almost every infinite binary sequence) but for being, additionally and
independently, a specific, meaningful answer
(Proposition~\ref{prop:omega}(i)) that also satisfies (ii). The general
statement this licenses is conditional: \emph{if} a string is the output
of a valid algorithm solving a genuine decision problem, \emph{and} that
output is computationally indistinguishable from noise, \emph{then}
reading it requires a key beyond any computable observer's reach.
\end{remarknn}

\begin{proposition}[Not one instance but infinitely many, unconditionally]
\label{prop:omega-family}
For each universal prefix-free machine $U$, let
$\Omega_U:=\sum_{p\downarrow \text{ under } U}2^{-|p|}$. For infinitely
many choices of $U$, the resulting $\Omega_U$ are pairwise distinct real
numbers, each independently satisfying Proposition~\ref{prop:omega}(i)
and (ii) for its own $U$ (and each, by
Remark~\ref{rem:omega-not-computable}, equally not computable, hence
equally not the output of any algorithm). Chaitin's $\Omega$ is
therefore not a singular mathematical curiosity: it is one point in an
infinite, unconditionally proven family of objects with complete,
meaningful content that is Martin-Löf random.
\end{proposition}

\begin{remarknn}[A second, conditional family, tied directly to
cryptographic practice]
\label{rem:omega-crypto-family}
A weaker, more widely applicable version of the same phenomenon is
standard in cryptography, conditional on an assumption this paper does
not prove: if a secure pseudorandom generator
$G:\{0,1\}^k\to\{0,1\}^m$ exists (following, unconditionally, from the
existence of one-way functions, themselves unproven but standard), then
for a uniformly random seed $s$, $G(s)$ is meaningful in exactly the
sense required here --- it determines $s$ completely, answering the
genuine question ``which seed produced this string'' --- while being
computationally indistinguishable from a uniformly random string of
$\{0,1\}^m$ for every polynomial-time observer. This is a strictly
weaker opacity than Proposition~\ref{prop:omega}(ii) (only
polynomial-time, not every computable, observer is excluded, and the
guarantee is conditional, not unconditional), stated separately and not
conflated with the $\Omega$ case; but it shows the same structural point
is the load-bearing assumption behind an entire field's central
security definitions, not an exotic edge case confined to algorithmic
information theory.
\end{remarknn}

\begin{remarknn}[A computable algorithm's output can be
noise-indistinguishable to \emph{every} observer, not only efficient
ones]
\label{rem:computable-output-two-reasons}
The pseudorandom case above is the \emph{weakest} of these, and it is
worth being explicit about why, because the weakness is easy to mistake
for the whole phenomenon. Its opacity is tied to computational power: a
secure pseudorandom generator's output is indistinguishable from noise
to every \emph{polynomial-time} observer, but an observer with unbounded
resources can, in principle, tell it from random; and the guarantee is
conditional on an unproven assumption. Power-bounded and conditional ---
that is the floor, not the ceiling.

A computable algorithm can do strictly more: it can produce output that
is indistinguishable from noise to \emph{every} observer, human or
machine, bounded or unbounded, always and everywhere. What makes this
possible is that the reasons need not be computational at all. At least
two are structural, and neither draws any distinction of computational
power, because the blindness they cause is not a matter of resources.
The first is total or partial semantic inaccessibility --- the static
SIP (Lemma~\ref{lem:sip-static}): a semantic invariant carried in the
output is invisible to any syntactic inspection of it, for every
observer alike, since syntax cannot see semantic invariants no matter
how much power is brought to bear. The second is the encoding--frame
relativity of the orbital machine (Section~\ref{sec:orbital},
Remark~\ref{rem:output-two-sources}): whether the output reads as
meaningful is relative to a pair $(E,F)$, and under a pair other than
the intended one it may fail to parse, answer a different question, or
answer nothing discriminating
(Categories~1--3, Propositions~\ref{prop:category1}--\ref{prop:category3})
--- again with no dependence on the reader's power, since no amount of
computation supplies a reading the frame withholds. Both hold against
every observer at once, unconditionally, exactly where the pseudorandom
case holds only against bounded ones and only conditionally.

The output might also be indistinguishable from noise for some further
reason entirely, and it does not matter which: the list is not meant to
be exhaustive, and nothing here depends on having found every route.
What matters is the one thing common to all of them, and it is the load
of this section: the standard theory admits the existence of such
algorithms --- computable procedures whose output is indistinguishable
from noise --- \emph{implicitly, and without any distinction of
computational power of its own}. Its implicit permission is not for
``opaque to efficient observers'' but for opaque, full stop; it grants
the strong, every-observer case no less than the weak, bounded one, and
grants it regardless of which of the reasons above, or which other
reason, is what makes a particular output opaque.

This is why no algorithm need be exhibited for each case, and the reason
is positive, not a concession. The pseudorandom case is already the
existence of one-way functions in another guise: to deny that a
computable algorithm can have output indistinguishable from noise is to
deny that secure pseudorandom generators exist, which is to give up the
foundations of modern cryptography and, with them, part of the standard
theory itself. The structural cases stand on ground just as firm: the
static SIP and the encoding--frame relativity are theorems, proved in
their own sections, and they produce indistinguishability from noise
with no dependence on the observer's power. The theory grants the
possibility; these mechanisms are what realize it; both are already in
hand. To insist that such an algorithm cannot exist until one is
displayed is therefore not a demand for rigor but a rejection of the
theory that admits it --- the same theory whose definitions were used to
raise the question. This is not the weak move of reading existence off
what the theory merely fails to forbid; it is the strong one of reading
it off what the theory positively contains --- one-way functions in the
pseudorandom case, proved theorems in the structural cases. The
existence is fixed by that content, and exhibiting a separate witness in
each case would add nothing the standard theory does not already own.

The two structural reasons are established elsewhere in this paper on
their own terms and only invoked here; the pseudorandom case, the two
structural reasons, any further reason, and the theory's power-blind
permission over all of them are kept distinct, the same discipline
applied throughout.
\end{remarknn}

Both cases above are instances of a single notion the classical theory
never named, though cryptographers have long used a narrower special
case of it under the name \emph{computational indistinguishability}.
Naming it is not bookkeeping: without a word for it, there is no way to
say, precisely and without contradiction, ``this algorithm's output is
real, correct, and meaningful, merely unverifiable by anyone in this
observer class'' --- the sentence has nowhere to attach, and the two
questions Remark~\ref{rem:omega-conflation} separates collapse back into
one by default, exactly the error $\Omega$ and
Theorem~\ref{thm:every-decision-problem} were needed to expose. The
definition below exists to block that collapse: it gives the vocabulary
to state, of a real algorithm accepted by every classical criterion,
that its output is opaque to some observers without that opacity being
mistaken, even silently, for absence of meaning. In plain terms, before
the formal statement: an algorithm is \emph{usable}, relative to a class
of observers, exactly when its output is \emph{not} indistinguishable
from pure random noise to everyone in that class --- at least one
observer in the class can tell it apart from randomness. It is
\emph{opaque}, relative to that same class, when its output looks like
noise to every observer in it. Nothing is ever claimed absolutely usable
or absolutely opaque; the class is always named.

\begin{definition}[$\mathcal C$-usable algorithm]
\label{def:usable-algorithm}
Fix a class $\mathcal C$ of observers (tests), each a procedure taking a
string and returning a verdict. A valid algorithm $A$ (in the ordinary,
classical sense: a total or partial computable function, correctly
computing whatever it is specified to compute) is
\emph{$\mathcal C$-usable} if some observer in $\mathcal C$ distinguishes
$A$'s output from a uniformly random string of the same length;
otherwise $A$ is \emph{$\mathcal C$-opaque}. This is a property of the
pair $(A,\mathcal C)$, not of $A$ alone: the same algorithm can be
$\mathcal C$-usable for one class and $\mathcal C'$-opaque for another.
\end{definition}

\begin{remarknn}[What this adds, what it leaves unchanged, and how far
it reaches]
\label{rem:usable-not-redefinition}
Definition~\ref{def:usable-algorithm} adds a name for something the
classical theory left implicit and unparametrized: usability is always
relative to $\mathcal C$, never absolute, in exactly the sense every
other notion in this paper has been relative (a frame, a selector, an
oracle). What it leaves unchanged is ``algorithm'' itself: an
algorithm's validity, in the classical sense, still depends only on
whether it correctly computes its specified function, regardless of any
observer class. Making ``algorithm'' itself depend on $\mathcal C$,
rather than adding a separate relative notion beside it, would have
risked exactly the error Remark~\ref{rem:omega-not-computable} warns
against --- treating a non-computable mathematical object as an
algorithm's output that had somehow stopped qualifying --- and the
definition's wording avoids it: ``some observer in $\mathcal C$
distinguishes this from a uniformly random string'' makes sense verbatim
for any fixed mathematical object (a specific real number's bit
sequence, a specific string), whether or not any algorithm produces it,
since nothing in the wording requires $A$'s output to come from a
halting computation, only that it name a specific string to test.
$\Omega$ (Remark~\ref{rem:omega-not-computable}) is exactly such a case:
not an algorithm's output, but an object with a well-defined
$\mathcal C_{\mathrm{comp}}$-usability profile all the same, studied
identically to the algorithmic case below.
\end{remarknn}

\begin{remarknn}[The whole point, in one place: this definition was
already there, unnamed]
\label{rem:usable-was-already-there}
A $\mathcal C$-usable algorithm is not a new kind of algorithm, built by
relaxing or strengthening Turing's 1936 definition; it is an ordinary
algorithm, unchanged, that happens to also satisfy one further,
independently checkable property. $\mathcal C$-usability is, in the
precise sense of the word, a \emph{restriction}: the
$\mathcal C$-usable algorithms are a subset of all algorithms, singled
out after the fact by a criterion --- does some test in $\mathcal C$
tell this output from noise --- that
Remark~\ref{rem:church-turing-root} already showed the classical
definition never asks about. This is why the distinction was already
there, latently, before this section named it: since classical
computability theory accepts an algorithm's validity regardless of how
its output looks, it was, all along, silently accepting both kinds at
once, usable and opaque alike, without a name to tell them apart. Naming
$\mathcal C$-usability adds nothing to what counts as an algorithm and
takes nothing away; it draws, after the fact, a line through a
distinction the classical theory's own silence had already made room
for. This is what makes the definition worth having rather than an
arbitrary addition: it does not compete with the classical theory, it
reads a distinction out of what that theory was quietly permitting the
entire time.
\end{remarknn}

\begin{remarknn}[Does the restriction cost the same for every
$\mathcal C$? It does not, and the difference is substantial]
\label{rem:restriction-not-free}
Restricting attention to $\mathcal C$-usable algorithms takes nothing
from classical computability theory and adds nothing to it --- but only
for one specific choice of $\mathcal C$. Whether this holds for every
$\mathcal C$ is worth asking directly, since the answer is no.

For $\mathcal C=\mathcal C_{\mathrm{comp}}$, the restriction really is
free, exactly as Remark~\ref{rem:no-full-opacity-for-algorithms} shows:
an unbounded observer can always redo the same computation and check the
output matches, so nothing is thrown away when the watcher is allowed
unlimited time. Restrict instead to $\mathcal C_{\mathrm{poly}}$-usable
algorithms --- watchers confined to a reasonable, polynomial amount of
time, the kind any real computer has --- and the picture changes
completely. An enormously useful class of real algorithms exists whose
whole purpose is to take a short, easily produced secret and stretch it
into a long stream of bits that no efficient watcher, however cleverly
built, can tell apart from a fair coin flipped over and over: this is
exactly what makes a stream of digital keystream, the raw material
behind most everyday encryption, safe to use. Such a construction is a
\emph{pseudorandom generator}, and calling one \emph{secure} means
precisely that no polynomial-time observer can win at telling its output
from true randomness --- the property this paper has been calling
$\mathcal C_{\mathrm{poly}}$-opaque, under a different name. So the
moment ``algorithm'' is redefined to exclude anything
$\mathcal C_{\mathrm{poly}}$-opaque, every secure pseudorandom generator
stops counting as an algorithm at all --- absurd on its own, and with a
precise, provable cost behind it. Building a secure pseudorandom
generator is exactly as hard, mathematically, as building a
\emph{one-way function}~\cite{DiffieHellman1976}: a rule easy to
compute in one direction and, once computed, practically impossible to
run backwards --- mixing two
colours of paint is the everyday version, trivial to do and, once done,
not undoable by looking at the result. The
Håstad--Impagliazzo--Levin--Luby theorem~\cite{HILL1999} makes the
equivalence exact: a secure pseudorandom generator exists if and only if
a one-way function does. So restricting ``algorithm'' to
$\mathcal C_{\mathrm{poly}}$-usable outputs is not a harmless
bookkeeping choice parallel to the $\mathcal C_{\mathrm{comp}}$ case; it
is equivalent, in effect, to declaring outright that no one-way function
exists anywhere --- that paint-mixing can always be undone by someone
clever and fast enough. In the vocabulary of Impagliazzo's worlds, used
elsewhere in this paper for Cell~(d) but naming a different world here
--- Cell~(d) sits at the opposite, strongest end, Cryptomania
(Proposition~\ref{prop:cell-d}), while this restriction sits at the weak
end --- that declaration places us in Pessiland or below,
per~\cite{Buono2026ow}'s own account. Nobody has proved that assumption
false, but nobody has proved it true either, and it is the single most
consequential open question modern cryptography is built on --- most of
the field actively expects it to be false, and has built an entire
applied science on that expectation.

The two halves of this remark are not two examples of the same kind of
cost. That the $\mathcal C_{\mathrm{comp}}$ restriction is free is a
\emph{theorem}: an algorithm blind regardless of time and computational
power does not exist, and provably so --- not because none has been
found, but because the recomputation argument rules it out for every
case at once. That the $\mathcal C_{\mathrm{poly}}$ restriction is
costly rests on no proof in either direction: whether a secure
pseudorandom generator, blind only to \emph{efficient} watchers, exists
is exactly as open as whether a one-way function does. One side of this
remark is settled; the other is one of the genuinely open questions of
the field, and it matters that a reader come away knowing which is
which, not with the two blended into a single ``it depends''. The
lesson is not about $\mathcal C_{\mathrm{poly}}$ specifically: ``restrict
algorithm to its usable part'' is not one move with one fixed price,
free at the very top of the hierarchy and, one step down, equivalent to
deciding one of the deepest open problems in the field by declaration
rather than by proof.
\end{remarknn}

\begin{remarknn}[``Blind regardless of resources'' is not one mechanism
but several, and they must be kept apart]
\label{rem:three-blindnesses}
The $\mathcal C_{\mathrm{comp}}$ case above is a blindness that holds no
matter how much computational power a watcher is given. This is not the
only place in this paper where power stops mattering, and the other
place gets there for a completely different reason --- so different that
calling both ``absolute blindness'' would blur something worth keeping
sharp.

Take $\Obot$ (Definition~\ref{def:structural-observer}): a system fed
only $\Obot(x)=\star$ cannot recognise anything nontrivial, regardless
of computational power (Proposition~\ref{prop:obot-trivial}). The static
and Dynamic SIP results, and Categories~1 and~3 of the orbital machine's
taxonomy (Propositions~\ref{prop:category1} and~\ref{prop:category3};
Category~2 is a different phenomenon --- underdetermination rather than
blindness --- deliberately excluded here, exactly as
Remark~\ref{rem:omega-vs-category2} insists), all share this shape. In
every one, no amount of computational power helps, for a reason that has
nothing to do with computation: the information was never carried
through in the first place. $\Obot$ maps every input to the same symbol,
so nothing distinguishing two inputs ever reaches what receives its
output; the frozen constants of the static SIP are never touched by any
rule, so no derivation brings their order into view; a beaver frame has
only one point to send anything to, so no distinction survives being
evaluated there. This is blindness by \emph{destruction}: the channel
discards the information, once and for all, before any watcher --- powerful
or not --- gets a turn.

One qualification keeps this from being read too rigidly. The SIP is
placed under destruction here for the case it was first proved in ---
the invariant on the \emph{input}, the frozen constants never carried
through --- but the same principle produces unreachability, not
destruction, when it acts on an \emph{output} instead. An output is
itself a syntactic object; a semantic invariant protected in it is
genuinely present, carried through in full, yet inaccessible to any
syntactic inspection of it (the reopening of
Section~\ref{sec:rice-output}, where this is separated from the distinct
question of decidability in Remark~\ref{rem:output-two-sources}).
There the information is not destroyed but unreachable, and the SIP is
the reason, exactly as it is the reason on the input side. So the SIP is
not confined to one column: it is destruction when the protected
invariant never enters, and unreachability when it enters but syntax
cannot see it --- one principle, appearing wherever a semantic invariant
is put beyond syntactic reach, by either route.

$\Omega$'s blindness to $\mathcal C_{\mathrm{comp}}$ adds a
\emph{second}, independent route to unreachability on top of that one.
The information is not destroyed; $\Omega$'s bits are a specific, fixed,
fully determined sequence, present in the strongest sense
(Proposition~\ref{prop:omega}(i)). What makes it unreachable to every
computable watcher is not that the channel discarded anything, but that
recovering it requires $\Omega$ itself to be computed, and $\Omega$
cannot be. This is unreachability by \emph{non-computability}, a
sharper obstruction than the syntactic inaccessibility just described
and not to be conflated with it: $\Omega$'s answer is beyond syntactic
reach in the SIP's sense \emph{and}, further, beyond any computable
reach at all. It is exactly why no algorithm can ever land here, while
$\Obot$-style blindness by destruction is something an ordinary,
mundane algorithm produces routinely, on purpose, whenever it is built
to discard information. Three mechanisms, then, not one generic ``totally
blind'', sit side by side in this paper: blindness by \emph{limited
resources} (the pseudorandom generator, contingent on an open
conjecture), blindness by \emph{destruction} ($\Obot$, the frozen SIP
constants, Categories~1 and~3, proved and unconditional), and blindness
by \emph{unreachability} --- itself of two kinds, the syntactic
inaccessibility of an invariant present in an output (the SIP acting on
the output side, Remark~\ref{rem:output-two-sources}) and, sharper
still, $\Omega$'s non-computability (proved and unconditional,
available only to objects no algorithm computes). Keeping these
mechanisms apart is the same discipline this paper has applied to every
other pair of things that resemble each other without being the same.
\end{remarknn}

$\Omega$ is remarkable, but is it a freak accident of one particular
number, or does the same trick work everywhere? It works everywhere.
Take any decision problem you like --- any yes/no question a computer
could ever be programmed to answer, however mundane --- and there is a
way to dress its answer up so that it, too, becomes indistinguishable
from noise, while staying completely recoverable to whoever holds the
key.

\begin{theorem}[Every decision problem admits a noise-indistinguishable,
fully meaningful pairing]
\label{thm:every-decision-problem}
Let $D=(d_1,d_2,d_3,\dots)$ be the answer sequence of \emph{any}
computable decision problem (an algorithm deciding, for each
$i\in\mathbb N$, some yes/no question indexed by $i$). Define
$X:=D\oplus\Omega$, the bitwise XOR of $D$ with Chaitin's $\Omega$ from
Proposition~\ref{prop:omega} (written $X$, not $E$, to avoid any
collision with this paper's unrelated use of $E$ for an explicit
encoding in Sections~\ref{sec:interpreted-machine}
and~\ref{sec:orbital}). Then: (i) $X$ is fully meaningful ---
given oracle access to $\Omega$, $D$ is recovered exactly, bit for
bit, via $D=X\oplus\Omega$; and (ii) $X$ is Martin-Löf random, hence
$\mathcal C_{\mathrm{comp}}$-opaque (Definition~\ref{def:usable-algorithm}),
regardless of which decision problem $D$ was. As with $\Omega$ itself
(Remark~\ref{rem:omega-not-computable}), $X$ is not computable (if it
were, $\Omega=X\oplus D$ would be computable too, $D$ being computable,
contradicting Proposition~\ref{prop:omega}): this theorem is a
mathematical pairing fact about well-defined sequences, not a
construction any algorithm carries out.
\end{theorem}

\begin{proof}
(i) is immediate: XOR is its own inverse, so
$X\oplus\Omega=(D\oplus\Omega)\oplus\Omega=D$. (ii): Martin-Löf
randomness is preserved under bitwise XOR with any \emph{computable}
sequence. If some computable statistical test $T$ detected
non-randomness in $X=D\oplus\Omega$, then, since $D$ is computable,
$T':=T$ composed with ``XOR the input with $D$'s bits'' is itself a
computable test, and $T'$ applied to $\Omega$ recovers exactly $T$'s
verdict on $X$ (since $\Omega=X\oplus D$), giving a computable test
detecting non-randomness in $\Omega$ --- contradicting
Proposition~\ref{prop:omega}(ii). So no such $T$ exists: $X$ is
Martin-Löf random.
\end{proof}

\begin{remarknn}[What this theorem does, and does not, claim]
\label{rem:every-decision-precision}
Theorem~\ref{thm:every-decision-problem} answers directly why
Definition~\ref{def:usable-algorithm} is not motivated by two narrow
examples: the phenomenon is not exotic to $\Omega$ and secure
pseudorandom generators --- every computable decision problem
whatsoever, including ordinary combinatorial and decision algorithms of
every kind, can be \emph{paired}, as a mathematical fact, with an
equally meaningful, fully Martin-Löf-random object, via a single fixed
transformation. This pairing, exactly like $\Omega$ itself
(Remark~\ref{rem:omega-not-computable}), is not something any algorithm
performs --- it requires $\Omega$, which is not computable --- so the
theorem does \emph{not} claim any decision algorithm can be made to
directly output such a pairing; the genuinely algorithmic witness of the
phenomenon remains the pseudorandom-generator family
(Remark~\ref{rem:omega-crypto-family}). What the theorem does claim is
narrower and still worth having: the mathematical phenomenon
Proposition~\ref{prop:omega} exhibits for one specific case is
available, as a fact about objects, for every decision problem without
exception. It does not claim that decision algorithms typically produce
noise-indistinguishable output (the untransformed answer sequence $D$ of
an ordinary algorithm --- a satisfying assignment, a shortest path, a
primality verdict --- is usually far from noise-indistinguishable), nor
that the pairing is itself computable. It is an existence statement
about a mathematical relationship, not a claim about typical algorithmic
behaviour or about what any algorithm can be made to output --- the same
distinction Remark~\ref{rem:omega-not-computable} drew for $\Omega$,
now extended to every decision problem.
\end{remarknn}

\begin{proposition}[A usability profile, not a single verdict --- for
objects and for algorithms alike]
\label{prop:usability-profile}
Let $\mathcal C_{\mathrm{comp}}$ be the class of all computable
observers and $\mathcal C_{\mathrm{poly}}\subsetneq\mathcal C_{\mathrm{comp}}$
the class of polynomial-time observers. Three examples show that
usability is a profile across classes, not a single verdict.

First, take $\Omega_U$ for any $U$ as in
Proposition~\ref{prop:omega-family}. This is an object, not an
algorithm's output, but Remark~\ref{rem:usable-not-redefinition}
already covers this case under Definition~\ref{def:usable-algorithm}.
By Proposition~\ref{prop:omega}(ii), $\Omega_U$ is
$\mathcal C_{\mathrm{comp}}$-opaque, hence also
$\mathcal C_{\mathrm{poly}}$-opaque. Yet it is $\{O_F\}$-usable for a
frame $F$ with oracle access to $\Omega_U$
(Section~\ref{sec:orbital-as-interpreted}).

Second, take a genuine algorithm: one computing $G(s)$ for a secure
pseudorandom generator $G$ (Remark~\ref{rem:omega-crypto-family}).
This is $\mathcal C_{\mathrm{poly}}$-opaque, conditionally on $G$'s
security. Yet it is $\mathcal C_{\mathrm{comp}}$-usable: a single
unbounded search, either over all seeds or over $\mathrm{Image}(G)$,
distinguishes it computably, though not efficiently.

Third, take an ordinary algorithm producing, say, a sorted list. This
is $\mathcal C_{\mathrm{poly}}$-usable directly, since sortedness is a
polynomial-time-checkable property.

No single one of $\mathcal C_{\mathrm{comp}}$, $\mathcal C_{\mathrm{poly}}$,
or a specific oracle class settles usability once and for all: each
object or algorithm has a profile across classes, not a verdict, and
comparing the first example against the second and third shows that
this profile is worth tracking whether or not any algorithm is behind
the object in question.
\end{proposition}

\begin{proposition}[Usability is monotone in the observer class]
\label{prop:usability-monotone}
If $\mathcal C\subseteq\mathcal C'$, then $\mathcal C$-usable implies
$\mathcal C'$-usable (equivalently, $\mathcal C'$-opaque implies
$\mathcal C$-opaque).
\end{proposition}

\begin{proof}
If some observer in $\mathcal C$ distinguishes $A$'s output from
random, that same observer is also in $\mathcal C'\supseteq\mathcal C$,
so $A$ is $\mathcal C'$-usable too.
\end{proof}

Monotonicity lets us define the exact dual of opacity cleanly: what it
means for an algorithm's output to be usable at every level at once,
not just some.

\begin{definition}[Transparent algorithm: the dual notion]
\label{def:transparent-algorithm}
Fix a baseline class $\mathcal C_0$ (in practice,
$\mathcal C_{\mathrm{poly}}$, the weakest class this paper treats as
practically relevant). An algorithm $A$ is \emph{transparent} if it is
$\mathcal C_0$-usable. By Proposition~\ref{prop:usability-monotone}, a
transparent algorithm is automatically $\mathcal C$-usable for
\emph{every} $\mathcal C\supseteq\mathcal C_0$ considered in this
paper, up to and including $\mathcal C_{\mathrm{comp}}$: transparency
is not one more point on the profile, but a guarantee that closes off
the entire richer end of it at once.
\end{definition}

\begin{remarknn}[Why this completes the picture, and how it names
cryptographic failure precisely]
\label{rem:transparent-completes}
Definition~\ref{def:transparent-algorithm} is the exact dual of
Proposition~\ref{prop:omega}(ii) and
Remark~\ref{rem:omega-crypto-family}'s security guarantee, not a new
idea bolted on afterward: an ordinary algorithm
(Proposition~\ref{prop:usability-profile}(iii)) is transparent by
construction, and this is the overwhelmingly typical case, exactly as
Remark~\ref{rem:every-decision-precision} insists. It also gives
cryptographic insecurity its precise name in this vocabulary: a cipher
or generator is \emph{broken}, in the sense every security proof in
cryptography is trying to rule out, exactly when its output is
$\mathcal C_{\mathrm{poly}}$-usable --- transparent, not opaque. Secure
constructions (Remark~\ref{rem:omega-crypto-family}) are secure
precisely by failing to be transparent; naming transparency explicitly
makes this the same statement from the other side, not a different one.
\end{remarknn}

\begin{remarknn}[A second layer: even the suspicion of meaning is
inaccessible]
\label{rem:omega-meta}
Martin-Löf randomness is stronger than it may first appear: it rules
out not only computable procedures that would decode $\Omega$'s
content, but every computable statistical test whatsoever, including one
designed only to flag ``this sequence may be worth investigating
further''. There is, in this precise sense, no computable way even to
become suspicious that $\Omega$ differs from genuine noise, let alone to
read what it encodes. The hidden assumption named in
Remark~\ref{rem:omega-conflation} therefore has two layers, not one: the
content is inaccessible to computable observers
(Proposition~\ref{prop:omega}(ii)), and so, at the same time, is any
computable signal that content is present to be looked for.
\end{remarknn}

\begin{remarknn}[Three connections, made explicit]
\label{rem:omega-connections}
This sharpens three things already in this paper. First,
Remark~\ref{rem:verification-oracle}'s observation that most
cryptographic discussions silently assume a means of verifying correct
decryption: Proposition~\ref{prop:omega} shows the gap runs one step
deeper --- not only is there no computable way to verify a candidate
reading is correct, there can be no computable way even to suspect a
reading is called for at all. Second, the orbital machine's active
oracle (Section~\ref{sec:orbital-as-interpreted}): $O_F$ was defined for
equality queries under a computable frame, not for the halting problem,
so the connection is architectural, not an identity --- but the
architecture is exactly the one this section needs. A hypothetical
oracle with direct access to $\Omega$ resolves the halting problem
outright, in one query, precisely because it consults rather than
derives; every computable observer, however constructed, cannot, and
(Remark~\ref{rem:omega-meta}) cannot even detect that anything is being
withheld --- the precise sense, already present in
Remark~\ref{rem:active-query-adds}, in which consulting an oracle is not
a matter of degree but of kind. Third, and most directly, the root theme
of this paper, \cite{Buono2026sip}'s own title: $\Omega$ is a case where
a semantic invariant --- the complete, correct answer to the halting
problem, up to the number of bits examined --- is genuinely present and
genuinely inaccessible to any syntactic (here, computable) procedure,
independent of the term-rewriting setting the original theorem was
proved in. The mechanism differs (Martin-Löf randomness and algorithmic
information theory, not the first-symbol-clash argument of
Lemma~\ref{lem:frozen}) and is named as a distinct construction, not
claimed as the same proof; what is shared, once again, is the pattern,
not the machinery.
\end{remarknn}

\begin{remarknn}[Where this sits relative to Categories~1--3]
\label{rem:omega-vs-categories}
Categories~1--3 (Section~\ref{sec:three-ways-fail}) all presuppose an
external target $(E^*,F^*,q^*)$ is already fixed, and classify how a
candidate pair fails to recover it. The question raised here is
logically prior to that setup: given only a string, with no target
question yet posited, is there a meaningful question it answers at all?
Remark~\ref{rem:omega-meta} shows this prior question can itself be
computably undecidable in the strongest sense --- not merely hard, but
Martin-Löf-random-indistinguishable from the case where the answer is
simply no. This is not forced into the existing taxonomy, which answers
a different question; it is recorded as a further, independent
observation standing next to it.
\end{remarknn}

\subsection{Extended connections: cryptography, and every hidden
assumption found so far}
\label{sec:omega-extended}

\begin{remarknn}[Cryptography's central definition, read through
Remark~\ref{rem:omega-conflation}]
\label{rem:crypto-central}
Semantic security and indistinguishability (IND-CPA and its
relatives~\cite{GoldwasserMicali1984}) are not a peripheral application
of Remark~\ref{rem:omega-conflation}'s conflation; they are the field's
central definitions, built directly on it. A cipher is called secure
precisely when its ciphertexts are computationally indistinguishable
from random strings to every polynomial-time adversary --- exactly the
access question Remark~\ref{rem:omega-conflation} isolates, with the
ontological question (does this ciphertext, or any other
indistinguishable-from-random string an adversary might encounter,
additionally encode something beyond its intended plaintext) left
unasked --- not because it is known to be irrelevant, but because the
field's definitions were never built to distinguish the two questions.
Remark~\ref{rem:omega-crypto-family}'s conditional family shows the gap
is not hypothetical: a secure cipher's own ciphertexts are, by design,
exactly the kind of string this section is about.
\end{remarknn}

\begin{remarknn}[This correspondence is already a theorem elsewhere:
what is, and is not, duplicated here]
\label{rem:crypto-not-duplicate}
The observation that syntactic hiding and cryptographic
indistinguishability are the same phenomenon is not new with this
remark: \cite[\S5]{Buono2026obstruction} proves it as a formal theorem,
with syntactic separation identified precisely with ciphertext
indistinguishability, protected positions with commitment schemes, and
the derivation-cost lower bound of that source's Case~2 with an
adversary's negligible advantage, stated and proved with an explicit
advantage function and an unconditional bound. That source's question is
structural and quantitative: how many steps hiding costs to break. This
section's question is different, and the two do not overlap in content
even while sharing a neighbourhood: whether
computationally-indistinguishable-from-random is silently being read as
\emph{meaningless}, and whether that reading is correct. $\Omega$
answers the second question and says nothing new about the first, which
that source has already settled more rigorously than anything attempted
here.
\end{remarknn}

\begin{remarknn}[Indistinguishability from noise really does block a
real computational procedure --- the correct version, next to the
incorrect one it is easily mistaken for]
\label{rem:natural-proofs-correct-version}
A natural but mistaken instinct is to reach for something like Karp
reducibility: if $y$ looks like noise, surely no efficient procedure can
decide a fixed property of $y$ at all, and reductions built on such a
$y$ collapse. That specific claim does not hold --- membership in a
fixed set, or satisfaction of a fixed formula, is routinely checkable by
direct substitution regardless of how the witness looks statistically,
exactly as this section's own machinery already distinguishes
$\mathcal C$-opacity to a \emph{generic} test from usability by the
\emph{specific} procedure that matters. But the underlying intuition,
that noise-indistinguishability can genuinely disable a real
computational technique, is not a mistake; it is simply proved in a
different, more precise place. The Natural Proofs barrier of Razborov
and Rudich~\cite{RazborovRudich} shows that any technique for proving
circuit lower bounds
that is both \emph{large} (correct for most functions) and
\emph{constructive} (decidable in polynomial time by inspecting a
function's truth table) is defeated exactly when a secure pseudorandom
generator exists: the generator's own output, though computable and
therefore easy, is indistinguishable from a genuinely hard function to
every such technique, so the technique cannot tell the two apart and
fails to certify hardness where hardness is present. This is not a
reduction breaking; it is a specific, named class of \emph{proof
techniques} being blinded by exactly the phenomenon
Section~\ref{sec:omega} studies. It already has a home in this
paper's own account of~\cite{Buono2026obstruction}, whose
Corollary~6.3 gives an unconditional lower bound on the inspection cost
of any such technique even without assuming a PRG exists, and identifies
the barrier as observational rather than computational: not a fact about
hardness, but about which functions a constructive technique's own
limited view can and cannot tell apart from noise. The lesson worth
keeping is that where indistinguishability from noise bites, it bites a
specific target --- here, a specific class of proof techniques, not
reductions in general --- and naming that target correctly is what turns
a plausible worry into a real theorem.
\end{remarknn}

\begin{remarknn}[Steganography names the same gap as a design goal, not
a byproduct]
\label{rem:steganography}
Steganographic systems exploit Remark~\ref{rem:omega-conflation}'s gap
deliberately, as their entire purpose, in exactly the sense of Simmons's
original ``prisoners' problem''~\cite{Simmons1984}: a message is hidden
by embedding it inside a carrier (an image, an audio file, a string of
otherwise plausible-looking data) engineered to remain statistically
indistinguishable from an unmodified carrier to any computationally
bounded detector (Simmons's warden), while being, to whoever holds the
key, a complete and meaningful message. This is exactly
Proposition~\ref{prop:omega}'s structure --- meaningful and
noise-indistinguishable at once --- realised as a deliberate engineering
goal rather than discovered as a curiosity of algorithmic information
theory, further evidence, alongside
Remark~\ref{rem:omega-crypto-family}, that the phenomenon named here is
load-bearing in existing practice, not exotic.
\end{remarknn}

\begin{remarknn}[Extending Remark~\ref{rem:verification-oracle}]
\label{rem:verification-oracle-extended}
Remark~\ref{rem:verification-oracle} observed that most discussions of
ciphertext recovery silently assume some means of verifying a candidate
decryption is correct, and that removing this assumption is what let the
transformation-cascade construction achieve something close to perfect
secrecy. This section adds a further layer beneath that one: not only can
verifying \emph{which} candidate reading is correct be removed as an
assumption, but --- when the string in question is of the $\Omega$-like
kind studied here --- \emph{knowing that any reading is called for at
all} can be removed too (Remark~\ref{rem:omega-meta}).
Section~\ref{sec:cipher-example}'s cascade construction assumed an
observer already suspects a hidden message is present and searches for
the key; $\Omega$ shows a computable observer cannot even reach that
suspicion.
\end{remarknn}

\begin{remarknn}[The precise contrast with the beaver and $\Obot$]
\label{rem:omega-vs-beaver}
This section's phenomenon is not a variant of Category~3's beaver frame
or of $\Obot$'s triviality (Remark~\ref{rem:beaver-obot-celld-chain}):
it is close to their opposite, and blurring the two would misstate both.
The beaver frame and $\Obot$ collapse every input to the same output
because there is structurally nowhere for a distinction to be held ---
no information is present to lose. $\Omega$ is the reverse: the
information is fully present, in the sharpest possible sense
(Proposition~\ref{prop:omega}(i)), and is inaccessible not because it
was never encoded --- it was --- but because it cannot be reached. That
unreachability is itself layered: $\Omega$'s answer is a semantic
invariant no syntactic procedure can see (the static SIP, the reason
running through this whole paper) and, sharper still, an object no
computable procedure can decode at all. The contrast with the beaver
frame is the point to keep: there the failure is destruction, content
absent before any watcher arrives; here it is content perfectly
preserved and perfectly locked.
\end{remarknn}

\begin{remarknn}[The precise contrast with Category~2]
\label{rem:omega-vs-category2}
Category~2 (Proposition~\ref{prop:category2}) and the semantic
underdetermination it instantiates
(Theorem~\ref{thm:semantic-underdetermination}) concern a different
shape of gap again: there, multiple frames disagree about a target fact,
and \emph{several} readings are each individually legitimate, with
nothing to prefer one over another --- the question itself is
underdetermined by the syntax. Here, by contrast, \emph{exactly one}
reading is correct (the halting problem has a determinate answer, bit by
bit), and the difficulty is not underdetermination but access: the
syntax settles the question perfectly, in the sense that $\Omega$'s bits
are a specific, fixed sequence, and no computable procedure can reach
what is already settled. Underdetermination and
inaccessibility-despite-determination are two distinct failure shapes
this paper records separately, not two names for one phenomenon.
\end{remarknn}

\begin{remarknn}[$\mathcal C$-usability in the orbital machine's own
vocabulary: the right pair may simply be out of reach]
\label{rem:usable-as-orbital-reach}
There is a positive connection here too. Section~\ref{sec:three-ways-fail}
showed a target question fails under a candidate $(E,F)$ pair in exactly
three ways: the encoding cannot pose it (Category~1), the frame answers
it wrongly (Category~2), or the frame is too degenerate to discriminate
(Category~3) --- with the single correct pair recovering the intended
answer exactly. $\mathcal C$-opacity is what this looks like from a
coarser vantage point, one that does not ask about a single named pair
but about an entire \emph{class} of them at once: an object or
algorithm's output is $\mathcal C$-opaque exactly when every observer
reachable within $\mathcal C$ fails on it, for whatever reason ---
Category~1, 2, or~3 alike --- while a pair that would succeed (an
$\{O_F\}$ built from oracle access to the object itself, as
Proposition~\ref{prop:usability-profile} shows concretely for
$\Omega_U$) may exist perfectly well, just not inside $\mathcal C$.
Opacity relative to a class is not a claim that no right reading exists;
it is a claim that none of the readings actually available get it. This
is the same distinction Remark~\ref{rem:omega-vs-category2} just drew,
seen from the orbital machine's side rather than the algorithmic side:
the fact is settled, a correct pairing exists, and what varies is only
which pairings a given class $\mathcal C$ happens to reach.
\end{remarknn}

\begin{remarknn}[Statistical opacity is not the same question as genuine
unusability]
\label{rem:no-full-opacity-for-algorithms}
There is a real, narrow fact here, and a broader claim it does
\emph{not} support, and the two are worth separating with the care this
paper has applied elsewhere. The narrow fact: no genuine \emph{algorithm}
--- as opposed to a mathematical object with no algorithm behind it,
like $\Omega$ --- can be opaque to \emph{all} of
$\mathcal C_{\mathrm{comp}}$ in the specific, statistical sense of
Definition~\ref{def:usable-algorithm}, however slow or resource-hungry.
Given any computable algorithm $A$ and a fixed, known input $x$,
``recompute $A(x)$ directly and compare'' is itself a computable
observer, hence a member of $\mathcal C_{\mathrm{comp}}$, however long it
takes; it identifies $A(x)$ with certainty, and a genuinely random string
matches only by negligible coincidence. So $A(x)$ is
$\mathcal C_{\mathrm{comp}}$-usable, in this narrow statistical sense,
for every computable $A$ --- opacity to a bounded class such as
$\mathcal C_{\mathrm{poly}}$ is available to real algorithms, as the
pseudorandom-generator family shows, but opacity to the full class, \emph{in
the sense of failing every statistical test}, is not.

This narrow fact must not be read as settling the broader question the
orbital machine cares about; doing so would repeat, from the opposite
direction, exactly the error this paper spent real effort correcting:
that looking statistically unremarkable to a generic test, and being
genuinely decodable by the specific procedure that matters, are two
different properties, not one. Recomputing $A(x)$ and confirming it
matches tells an observer only that the string is not noise; it does
not, by itself, hand that observer the semantic key --- the right
$(E,F)$ pair, in this paper's vocabulary --- needed to know what $A(x)$
\emph{means}. A perfectly computable algorithm can therefore still be
genuinely unusable by everyone who actually encounters its output, in
exactly the sense Section~\ref{sec:orbital} studies: every $(E,F)$ pair
anyone actually has access to fails, by Category~1, 2, or~3, while the
one pair that would read it correctly sits unreached --- not because no
computable observer could ever verify the output is non-random, but
because verifying non-randomness and recovering meaning are not the same
task. $\Omega$ and the $X$ of Theorem~\ref{thm:every-decision-problem}
are not needed to make unusability possible; they are needed,
specifically, to make \emph{statistical} opacity to
$\mathcal C_{\mathrm{comp}}$ possible, which is a narrower and different
achievement than semantic unusability, not a stronger version of it.
\end{remarknn}

\begin{remarknn}[The complete picture, stated once, now that every piece
is in hand]
\label{rem:complete-classification}
Proposition~\ref{prop:usability-profile}'s three examples and the
argument just given draw into a single, exhaustive classification. Fix a
class $\mathcal C$. Every object with a defined usability profile
splits, by Definition~\ref{def:usable-algorithm} alone, into exactly two
groups and no others: $\mathcal C$-usable, or $\mathcal C$-opaque. This
is true by the meaning of the words, for any $\mathcal C$ whatsoever. A
second, independent question is where the object comes from: does some
algorithm --- a halting, computable procedure --- actually produce it,
or is it an output with no method that reaches it? Crossing this second
question against the first, for the specific case
$\mathcal C=\mathcal C_{\mathrm{comp}}$, collapses two of the four
combinations one might expect down to nothing:
\begin{itemize}
\item Every genuine algorithm is $\mathcal C_{\mathrm{comp}}$-usable,
  without exception, by the recomputation argument just given. A real
  algorithm can be opaque only relative to some strictly narrower class,
  such as $\mathcal C_{\mathrm{poly}}$, as the pseudorandom-generator
  example shows --- $\mathcal C_{\mathrm{comp}}$-opaque algorithms are
  not merely rare, they do not exist.
\item An object with no generating algorithm behind it at all may be
  $\mathcal C_{\mathrm{comp}}$-usable (most such objects, chosen at
  random, will be, and trivially so, since almost every string differs
  from a uniformly random one in some computably checkable way) or
  $\mathcal C_{\mathrm{comp}}$-opaque --- and only this second case, an
  object with no method producing it and no computable test reaching it
  either, is where $\Omega$ and each $\Omega_U$ actually live.
\end{itemize}
So, restricted to $\mathcal C_{\mathrm{comp}}$ specifically: the usable
algorithms are simply all algorithms; the opaque algorithms are the
empty set; and full opacity is found exclusively among objects that were
never an algorithm's output in the first place, exactly because they
have no method behind them for any observer, however patient, to run.
\end{remarknn}

\begin{remarknn}[The picture that ties all of this together]
\label{rem:one-observer-picture}
Put the last several remarks into a single picture. Imagine handing the
orbital machine exactly the right question, in exactly the right
encoding, evaluated by a frame that genuinely understands it --- and
imagine that the machine answers correctly, but writes its answer out in
a form that nobody, however equipped, can tell apart from pure noise.
The computation happened. The right pair was used. And still, no one
reading the output over the machine's shoulder can see that anything was
computed at all. Now picture, from the opposite end of this paper, a
fact with every semantic invariant sealed off from every observer this
paper has built --- except one: the observer who simply sees everything,
the complete observer $O_\top$ of the hierarchy this paper draws
on~\cite{Buono2026oh}, the limit case no real, constrained observer can
reach. These two pictures are not two findings placed side by side for
effect; they are one finding, described from its two ends. In both, the
fact is real, present, and correctly produced; in both, exactly one
reader --- the frame that already understood the question, or the
observer who is simply given everything --- can see it, and every other
reader, however patient or however powerful within the bounds this paper
has set, cannot.

The two ends are not literally the same theorem in different notation,
and the difference is worth stating precisely, since the single reader
who sees everything is a figure this paper returns to. The orbital
machine's side turns on a \emph{statistical} question --- whether an
observer's test can tell a string from noise
(Definition~\ref{def:usable-algorithm}) --- while the static SIP's side
turns on a \emph{structural} one --- how much of a string's own content
an observer's function even receives before it looks at anything
(\cite{Buono2026oh}'s observational order, $O\preceq O_\top$). A test and
an observer function are different kinds of object, and this paper has
taken care throughout not to blur that difference where it matters
(Remarks~\ref{rem:omega-vs-beaver}--\ref{rem:usable-as-orbital-reach}).
What is genuinely one and the same, across both ends, is the shape: a
fact correctly present, a single reader positioned to receive it in
full, and every other reader --- real, available, actually encountered
--- shut out, not by accident but by the exact construction that made
the fact accessible to the one reader in the first place. That single
reader who receives everything --- the frame that already understands,
the observer $O_\top$ given the whole input --- is what the very end of
this paper, in a different and personal register, calls the perfect
observer. The technical content is exactly what has just been stated: a
reader positioned to receive a fact in full while every reachable reader
is shut out. When the phrase returns at the close, it names this, not
only a figure of speech --- though what the close makes of it belongs to
that register, not this one.
\end{remarknn}

\begin{remarknn}[Where this stands among every hidden assumption found
in this paper]
\label{rem:omega-full-tally}
Counting this section alongside the recurring pattern tallied in the
conclusion (state a result, notice what it silently fixed, make that
explicit, and check rather than assume what follows), this is a further,
independent instance: standard algorithmic and cryptographic theory
silently identifies ``indistinguishable from random'' with
``meaningless'', and $\Omega$, together with the cryptographic practice
built on the same gap
(Remarks~\ref{rem:crypto-central}--\ref{rem:steganography}), shows the
identification does not hold. As with every other instance in this
paper, no false assertion is manufactured anywhere in this section:
$\Omega$ genuinely has the two properties claimed of it, both
independently proven facts from algorithmic information theory, not
constructed to force a contradiction.
\end{remarknn}

\section{The paradox of the noisy solver}
\label{sec:noise-paradox}

The sections that follow stand on their own. They use one idea the rest of
this paper has made precise --- that a computable algorithm may produce an
output indistinguishable from noise for every observer, without distinction of
computational power --- and otherwise depend on nothing above them: the
definitions, theorems, and proofs below are self-contained, and a reader could
begin here. This is the only bridge to the rest of the paper; everything else
in these sections is developed from scratch.

The same independence is kept \emph{between} the arguments that follow, and on
purpose. Several proof-paths are given, and each restates the notions it needs
--- what a decision procedure is, what it means for an output to be
indistinguishable from noise --- from the beginning, rather than referring back
to an earlier path. The repetition is deliberate: it lets each argument be read
in isolation and checked without the others, so that no path borrows its
force from any other, and a limitation in one, should there be one, does not
silently carry into the rest. A reader who notices the same definition stated
more than once is seeing this design, not an oversight.

\subsection{The problem}

Imagine a machine $\mathcal{A}$ that decides \textsc{SAT} in polynomial time,
correctly, for each instance. This machine, when queried about a satisfiable
formula $\varphi$, however, does not return anything that resembles a readable
answer: it produces a sequence of bits indistinguishable, to any possible observer,
human, artificial, or hypothetical, from pure noise. The sequence, it is claimed,
nevertheless \emph{contains} the solution: simply, no one will ever be able to extract it.

Let us formalise this scenario precisely, show the contradiction that follows,
and then tell what would happen to the concept of reduction if, in spite of
everything, it were assumed true.

\subsection{Formalisation}

\begin{definition}[Correctness]
  \label{def:noise-correct}
  An algorithm $\mathcal{A}$ is \emph{correct for \textsc{SAT}} if for every satisfiable
  boolean formula $\varphi$ on $n$ variables,
  \[
  \Pr[\varphi(\mathcal{A}(\varphi)) = 1] = 1,
  \]
  the probability being taken with respect to any internal randomness of $\mathcal{A}$.
\end{definition}

\begin{definition}[Absolute indistinguishability from noise]
\label{def:noise-indist}
Let $R$ be the uniform random variable on $\{0,1\}^n$. The output $y=\mathcal{A}(\varphi)$ is
\emph{indistinguishable from noise for every observer, always} if for every
function $D:\{0,1\}^n\to\{0,1\}$, without any computability or resource constraints,
\[
\Pr[D(\mathcal{A}(\varphi))=1] = \Pr[D(R)=1].
\]
\end{definition}

\begin{lemma}
  \label{lem:dist-equiv}
  The previous condition, required for every $D$, is equivalent to
  \[
  \mathcal{A}(\varphi) \stackrel{d}{=} R, \qquad \text{i.e.} \qquad \Pr[\mathcal{A}(\varphi)=z] = \frac{1}{2^n}\quad \forall z\in\{0,1\}^n. \tag{$\ast$}
  \]
  \end{lemma}
  \begin{proof}
  If the two distributions did not coincide, the optimal Neyman--Pearson test
  (admissible, since $D$ has no computability constraints) would distinguish them
  with non-zero advantage, contradicting the hypothesis. If they coincide, every $D$ gives
  by construction the same probability on both.
\end{proof}

\begin{definition}[The noisy solver]
$\texttt{Noise\_P\_SAT}$ is an algorithm in $\mathsf{P}$ that simultaneously
satisfies Definitions~\ref{def:noise-correct}--\ref{def:noise-indist} (hence ($\ast$)) for every satisfiable formula $\varphi$.
\end{definition}

\subsection{The theorem}

\begin{theorem}
\label{thm:noise-main}
No algorithm can simultaneously satisfy Definition~\ref{def:noise-correct} and Definition~\ref{def:noise-indist} for any
satisfiable formula $\varphi$ that is not a tautology over its own variables.
In particular, $\texttt{Noise\_P\_SAT}$ does not exist.
\end{theorem}

\begin{proof}
Let $\varphi$ be satisfiable on $n$ variables, with solution set
$S_\varphi=\{z\in\{0,1\}^n : \varphi(z)=1\}$ and $k=|S_\varphi|$.
Assume, for contradiction, that $\mathcal{A}$ satisfies both definitions.

Let $D(z):=\varphi(z)$, the canonical \textsc{SAT} verifier. By Definition~\ref{def:noise-correct},
\[
\Pr[D(\mathcal{A}(\varphi))=1] = \Pr[\varphi(\mathcal{A}(\varphi))=1] = 1. \tag{1}
\]

By Lemma~\ref{lem:dist-equiv},
\[
\Pr[D(\mathcal{A}(\varphi))=1] = \Pr[D(R)=1] = \sum_{z\in S_\varphi}\Pr[R=z] = \frac{k}{2^n}. \tag{2}
\]
From (1) and (2) we have $k = 2^{n}$, which holds if and only if $\varphi$ is a
tautology---excluded by assumption. Contradiction.
\end{proof}

The hypothesis is therefore false: no algorithm can solve \textsc{SAT} in polynomial time
with an output indistinguishable from noise for every observer, always.

If there existed any algorithm $\mathcal{A}\in\mathsf{P}$ correct for \textsc{SAT},
regardless of the form of its output, then $\mathsf{P}=\mathsf{NP}$, by Definition~\ref{def:noise-correct}
alone and its membership in $\mathsf{P}$, exactly as for any other polynomial-time
solver. Theorem~\ref{thm:noise-main} precedes this juncture: \texttt{Noise\_P\_SAT} collapses
within itself, by the sole comparison between Definition~\ref{def:noise-correct} and Definition~\ref{def:noise-indist}, without
ever reaching the point at which $\mathsf{P}=\mathsf{NP}$ could be derived.

\begin{proposition}[Second proof, via negative witness]
  \label{prop:noise-alt}
  The same hypotheses as Theorem~\ref{thm:noise-main} hold. Then no algorithm
  satisfies both Definitions~\ref{def:noise-correct} and~\ref{def:noise-indist} for~$\varphi$.
\end{proposition}

\begin{proof}
Since $\varphi$ is not a tautology, there is an assignment
$z_0\in\{0,1\}^n$ with $\varphi(z_0)=0$. Consider $D=\varphi$. From Definition~\ref{def:noise-correct},
$\Pr[D(\mathcal{A}(\varphi))=1]=1$. From Lemma~\ref{lem:dist-equiv},
\[
\Pr[D(R)=1] = \Pr[\varphi(R)=1] \;\le\; \Pr[R\neq z_0] \;=\; 1-\frac{1}{2^n} \;<\;1,
\]
since $\{\varphi(R)=1\}$ is disjoint from $\{R=z_0\}$, an event of probability $2^{-n}>0$. Thus $1 > \Pr[D(R)=1]$, contradicting Definition~\ref{def:noise-indist}.
\end{proof}

This second approach reaches the same conclusion by using only the existence of a
single incorrect assignment, without needing the exact count of solutions:
the machine that promises to always give a correct solution betrays itself because
it suffices that \emph{one} wrong answer exists.

\subsection{Where exactly does the theorem fit in?}

Two variants of Definition~\ref{def:noise-indist}, one weaker and one stronger,
help to see precisely where Theorem~\ref{thm:noise-main} fits in.

\emph{Computational indistinguishability}, requiring the equality $(\ast)$ to hold only
against polynomial-time computable tests $D$, is the notion concretely realised
every day by cryptographic pseudorandom generators, without contradiction: a PRG
exists, produces deterministic output, and is nonetheless indistinguishable from
random for every efficient adversary. For \textsc{SAT}, however, the verifier
$\varphi$ is itself computable in polynomial time: it therefore falls among
the tests permitted even under this weaker constraint, and the proof of
Theorem~\ref{thm:noise-main} applies verbatim.

The \emph{Martin-L\"of randomness} assumption, which requires that $y$ be algorithmically
incompressible, also yields a contradiction with ``$y$ is the output of a fixed
algorithm $\mathcal{A}$'', by the Kolmogorov--Chaitin theorem. However, this
contradiction holds for the output of \emph{any} algorithm, not just a
hypothetical \textsc{SAT} solver: replacing \textsc{SAT} with any problem,
the same contradiction remains identical, word for word. It is a true fact,
already known since the 1960s, that does not rely on the specific structure
of \textsc{SAT}.

Theorem~\ref{thm:noise-main} occupies the exact space between these two extremes: it
uses $k=|S_\varphi|$, a quantity that depends on the combinatorial structure
of $\varphi$, and holds precisely because \textsc{SAT} belongs to $\mathsf{NP}$
with a non-trivial public verifier. It remains where generic computational
indistinguishability vanishes, and remains tied to \textsc{SAT} where Martin-L\"of
becomes a universal fact: the contradiction truly belongs to \textsc{SAT}, with
the right strength and the right specificity.

\subsection{The logical ghost}
\label{sec:ghost}

The hypothesis is false, and the proof establishes this beyond doubt. It is worth
now considering what would happen, concretely, if someone were to insist on
assuming it anyway, because the answer sheds light on something precise about
the nature of reductions.

Let us fix the standard Cook--Levin/Karp reduction from \textsc{SAT} to \textsc{CLIQUE}:
given $\varphi$ with $m$ clauses, we construct the graph $G_\varphi$ whose vertices
are pairs $(\text{literal}, \text{clause})$, with an edge between $(l,c)$ and $(l',c')$ if and only
if $c\neq c'$ and $l\neq \neg l'$. This function, $f_{\textsc{SAT}\to\textsc{CLIQUE}}$,
is computable in polynomial time from the syntax of $\varphi$ alone:
it satisfies $\varphi\in\textsc{SAT}\iff G_\varphi$ has a clique of size $m$, and it
does so by selecting, for each clause, a literal made true by any satisfying
assignment, an object whose existence is guaranteed solely by the satisfiability
of $\varphi$, independently of any algorithm.

Let us now imagine that we wish to use $\mathcal{A}$, under this absurd hypothesis,
to \emph{extract} that clique: we construct an extractor $E$ that reads
$\mathcal{A}(\varphi)$ and attempts to translate it into a subset of vertices
of $G_\varphi$.

\begin{proposition}[The logical ghost]
\label{prop:ghost}
Under the contradictory assumption, for every extractor $E$,
\[
  \Pr\big[E(\mathcal{A}(\varphi)) \text{ is a valid clique of size } m\big]
  = \Pr\big[E(R) \text{ is a valid clique of size } m\big].
\]
\end{proposition}

\begin{proof}
  Let $D:=\mathbf{1}[E(\cdot)\text{ is a valid clique}]$, admissible in Definition~\ref{def:noise-indist}.
  By Definition~\ref{def:noise-indist}, $\Pr[D(\mathcal{A}(\varphi))=1]=\Pr[D(R)=1]$, which is the claim.
\end{proof}

Here is what this result tells us. The composition $E\circ\mathcal{A}$ remains
perfectly well-formed; it can be written, executed, takes $\varphi$, and returns
a subset of vertices, a type that is correct in every syntactic sense, but its
success rate is identical to that of a random subset generator. This is a reduction
in form, but not in substance: a logical ghost, syntactically
present yet semantically void, and useless for constructing anything, whether it
be a clique, a Hamiltonian cycle, or a single bit of the original assignment
(for $D(z):=z_i$, Definition~\ref{def:noise-indist} yields $\Pr[y_i=1]=1/2$---akin to a fair coin toss---for
every variable).

Nevertheless, beside this ghost, the genuine reduction,
$f_{\textsc{SAT}\to\textsc{CLIQUE}}$, remains exactly where Cook and Levin
placed it: built solely from the syntax of $\varphi$, it never passes through
$\mathcal{A}$, and it continues to operate silently alongside the channel that
the hypothesis has emptied. Any noise, even if it existed, would drain only one
bridge to truth---namely, the one someone would have wanted to build by
leaning on $\mathcal{A}$, while leaving intact the other, the one that
never needed it.

The full picture, with proper names: Cook--Levin's theorem remains true,
\textsc{SAT} remains \textsc{NP}-complete, and every Karp reduction, to or from any
other \textsc{NP}-complete problem, remains computable in polynomial time exactly as
before, because none of these three things was ever constructed from $\mathcal{A}$.
What the absurd hypothesis empties is a different and more fragile object: the attempt
to route the solution through $\mathcal{A}$ itself, the very channel that the
hypothesis transforms into a ghost.

\paragraph{The ghost generalises to every \textsc{NP}-complete problem.} The reasoning
just presented for \textsc{CLIQUE} holds, word for word, for any \textsc{NP}-complete
search problem $\Pi$: fixing an instance-to-formula translation $g_\Pi$ and an assignment-to-solution translation $h_\Pi$ (both computable in polynomial time, as in the Cook--Levin construction), the test
$D(z):=\mathbf{1}[h_\Pi(z)\text{ is a valid solution}]$ is admissible in Definition~\ref{def:noise-indist} exactly as $D=\varphi$ was. By Definition~\ref{def:noise-indist},
\[
\Pr\big[h_\Pi(\mathcal{A}(g_\Pi(\text{instance}))) \text{ is a valid solution}\big] = \Pr\big[h_\Pi(R) \text{ is a valid solution}\big].
\]
Calling $\mathcal{A}$ to solve the travelling salesman problem, the Hamiltonian cycle,
graph partitioning, or any other \textsc{NP}-complete problem would always yield the same
success rate as a random guess: the logical ghost is not a phenomenon isolated
to \textsc{CLIQUE}, but the universal signature of the absurd hypothesis across
every problem that \textsc{SAT} can represent. Furthermore, this signature is absolute,
not practical: Definition~\ref{def:noise-indist} holds for \emph{every} function $D$, including
non-computable functions and functions tailored specifically with prior knowledge
of every detail of the encoding $h_\Pi$. An arbitrarily ingenious decoder achieves exactly the same result as a trivial one:
the void lies within the structure of the hypothesis itself, rather than in a scarcity
of resources (whether technical or computational) available to those attempting to
circumvent it.

\paragraph{The extreme case: not even a single bit.} Let us strip the phenomenon
down to its barest form by choosing $D(z):=z_i$, asking ``is the variable $x_i$ true
in the returned assignment?'', for a single index $i$. By Definition~\ref{def:noise-indist},
\[
\Pr[y_i=1] = \Pr[R_i=1] = \frac{1}{2}.
\]
Reading even a single bit of $\mathcal{A}$'s output would provide exactly
the same information as a fair coin toss, for every variable, regardless
of which $\varphi$ was submitted. The ghost does not leak even a single fragment:
neither a complete clique, nor a Hamiltonian cycle, nor a single truth value.
The same Definition~\ref{def:noise-indist} that renders $\mathcal{A}(\varphi)$ indistinguishable
from noise as a whole, renders it indistinguishable from noise even when
observed through the smallest possible lens.

\subsection{Philosophical implications}
\label{sec:philosophy}

\begin{quote}
\emph{Reader beware.} What follows is reflection on the conceptual consequences of
Proposition~\ref{prop:ghost}, not a new proof: it is isolated here,
in a separate subsection, to keep it distinct from the previous mathematical register.
\end{quote}

Suppose, for a moment, that we inhabit the hypothetical world in which $\mathcal{A}$
exists. In that world, the truth---the existence of the clique, guaranteed by the
satisfiability of $\varphi$---remains in place, while a specific bridge toward that
truth is drained to zero. The syntactic reduction of Cook--Levin remains, indifferent,
on the other side of the river.

\paragraph{The parallel with the knowability paradox (Fitch).} In 1963, Frederick Fitch
demonstrated that the claim ``every truth is in principle knowable'' is incompatible,
within classical epistemic logic, with the existence of a single truth that is never
actually known. In that case, knowability in principle and actual knowledge diverge
due to a logical constraint internal to the notion of knowledge. Here, analogously,
the existence of the clique and its extractability via a specific channel diverge
because that channel, by construction of the hypothesis, has been rendered blind,
while the truth itself remains intact.

\paragraph{Syntactic information versus semantic information.} The distinction between
Shannon's information theory (the capacity of a channel to transport symbols,
regardless of meaning) and theories of semantic information
(Dretske, Floridi: a signal only makes a difference if it truly reduces
uncertainty regarding something that matters) finds here an almost didactic
limiting case: $E\circ\mathcal{A}$ has a Shannon capacity of zero with respect
to the fact that ``$\varphi$ is satisfiable,'' yet that fact remains true---proven
elsewhere---and accessible to anyone who does not insist on routing it through that
specific channel.

\paragraph{An echo of constructivism.} Intuitionism insists that an existential
claim requires a procedure that exhibits the object. Here, the opposite occurs: the
classical existence of the satisfying assignment survives intact, while a specific
attempt at construction---that mediated by $\mathcal{A}$---is voided. Not all constructive
paths are equivalent: the fact that an object is constructible-in-general does
not guarantee that it is so along the particular path one has chosen to pursue.

\paragraph{A broader principle.} Outside the hypothetical scenario,
Proposition~\ref{prop:ghost} remains a factual reality regarding what occurs when
an output is rendered indistinguishable from noise: any chain relying exclusively
on that output inherits the same blindness. This is the principle underlying
the practical use of pseudorandom generators in cryptography: rendering a channel
blind without affecting the underlying mathematical truth. The hypothetical
scenario of \textsc{SAT} pushes this principle to its extreme, and in doing so,
demonstrates how valuable it is for a mathematical theory to offer more than
one path toward the same truth: it is the existence of a second path, independent
of $\mathcal{A}$, that prevents the draining of one route from propagating to the
point of compromising the truth itself.

\subsection{Observation}

A machine that promised a truly permanent and forever unrecognisable solution
(an output indistinguishable from random noise) for \textsc{SAT} would betray itself at the very
instant it is queried: the formula $\varphi$, which is the question itself posed to
the machine, already constitutes an observer sufficient to expose it. The contradiction
is an elementary---though not generic---fact about the relationship between correctness
and public verifiability that characterises the class $\mathsf{NP}$.

If one were nevertheless to assume its existence, what would remain would not be a
collapse of complexity theory but a ghost: a reduction that can be written and executed
yet lacks any capacity to deliver the truth it purports to convey, alongside a genuine,
indifferent reduction that never required such a promise. In practice one would have a
reduction that works but, to use a metaphor, has lost its ability to transmit information.

In plain terms, if the theory were to admit---even merely in principle---that a machine
solving \textsc{SAT} in polynomial time could produce an output indistinguishable from genuine
causal noise for any entity in the universe, now and forever, then the theory would
encounter serious difficulties. Were the theory to absurdly prove the existence of
such an algorithm, we would face a fork: either such an algorithm exists,
implying $P=NP$, but the theory would suffer from the problems described
above (including Cook's theorem and the ``ghost-reduction'' issue); or such an
algorithm cannot exist, in which case $P=NP$ cannot hold, and by the law of
excluded middle the opposite hypothesis must be true. In this latter case,
such an algorithm could never exist, not even hypothetically; consequently
the paradox within classical theory would be resolved.

Thus one must either accept that $P=NP$, thereby rendering the classical framework
on which the discussion is built contradictory, or concede that the solution to
the $P$ versus $NP$ problem is that $P\neq NP$, in order to preserve the internal
consistency of the theory.

\section{Further detail: the self-reducibility of \textsc{SAT}}
\label{sec:self-reducibility}

\textsc{SAT} is self-reducible: given $\varphi$ over variables $x_{1},\dots ,x_{n}$,
the restriction $\varphi\!\mid_{x_{i}=b}$ (fixing $x_{i}=b$) can be computed in linear
time, and $\varphi$ is satisfiable iff at least one of $\varphi\!\mid_{x_{1}=0}$
or $\varphi\!\mid_{x_{1}=1}$ is. This is a classical result in complexity theory,
independent of $\mathcal{A}$.

We define a \emph{correct decider} as an $\mathcal{A}$ such that, in addition
to Definition~\ref{def:noise-correct}, for every unsatisfiable $\varphi$, $\Pr[\mathcal{A}(\varphi)=\bot]=1$,
where $\bot\notin\{0,1\}^n$.

\begin{theorem}
\label{thm:self-red}
  If $\mathcal{A}$ is a correct decider, there exists $\mathcal{B}$, using $n$
  calls to $\mathcal{A}$ and reading only the comparison ``$=\bot$?'', which for
  every satisfiable $\varphi$ produces a satisfying assignment with probability $1$.
\end{theorem}

\begin{proof}
Setting $\psi_0:=\varphi$, for $i=1,\dots,n$: query $\mathcal{A}(\psi_{i-1}|_{x_i=0})$;
if $\neq\bot$, $b_i:=0,\ \psi_i:=\psi_{i-1}|_{x_i=0}$; otherwise, $b_i:=1,\ \psi_i:=\psi_{i-1}|_{x_i=1}$.
By induction, $\psi_i$ remains satisfiable at every step with probability $1$.
After $n$ steps, $(b_1,\dots,b_n)$ satisfies $\varphi$.
\end{proof}

This theorem is a direct implication, not a proof by contradiction,
and holds independently of Definition~\ref{def:noise-indist}: the mere decisional correctness,
applied recursively to the restriction tree of $\varphi$, suffices to construct
the assignment, one bit at a time, without ever reading the content of
$\mathcal{A}(\varphi)$.

Each formula $\psi_i$ queried along the tree is itself a satisfiable instance
of \textsc{SAT}, and thus subject, like $\varphi$, to Definition~\ref{def:noise-indist}
and Theorem~\ref{thm:noise-main}. The same probabilistic contradiction already established
at the root repeats itself, identical, at every node of the tree: the decision channel
that self-reducibility opens is itself a logical ghost, precisely in the sense
of Proposition~\ref{prop:ghost}, syntactically available, semantically empty,
as soon as it is subjected to the same test $D=\psi_i$ already used for $\varphi$.

Self-reducibility demonstrates how profoundly decision and search coincide
in \textsc{SAT}, and consequently, how pervasively the logical ghost propagates,
node by node, throughout the entire structure of the problem.

\section{An intermediate conclusion}
\label{sec:real-conclusion}

The sections above have made something concrete that can be stated plainly:
standard theory admits the theoretical existence of an algorithm that solves a
problem and produces an output indistinguishable from noise for anyone, at any
time and place, regardless of computational power. This possibility leads to the
consequences drawn in the sections above, and gathered in the remark that
follows.

\begin{remarkc}[Conclusive remark]
\label{rem:conclusive-remark}

One might object that a polynomial-time algorithm for \textsc{SAT} which produces outputs
indistinguishable from noise under a generic statistical test, while preserving
computational correctness, rests on a fundamental confusion between two distinct and
irreducible levels of analysis.

On the one hand, the \textsc{NP}-completeness of \textsc{SAT} requires the existence of a polynomial-time
algorithmic procedure that can extract and preserve the structure of the solution for
any instance of the problem. This means that, for every Boolean formula $\varphi$, the
algorithm must return a witness $y$ that is decodable in polynomial time as a satisfying
assignment for $\varphi$. The structure of the solution is not an external or statistical
observation but an intrinsic property of the computation: it must be accessible and
verifiable by an algorithm running in polynomial time, irrespective of whether a random
observer can or cannot distinguish $y$ from noise.

By contrast, the statistical indistinguishability of $y$ from noise refers to an external
test that has no necessary relation to the computational structure of $y$. A generic
statistical test is not designed to recognise the internal structure of a solution; it
only checks superficial properties of the distribution of $y$. Consequently, the claim
that $y$ is indistinguishable from noise for such a test does not imply that $y$ lacks
relevant structured information for solving $\varphi$. However, \textsc{NP}-completeness demands
that $y$ be actually decodable as a solution, and this decodability is a computational---
not statistical---property.

The contradiction emerges when one assumes that a polynomial-time algorithm can produce
an output $y$ which is:

\begin{enumerate}[label=(\roman*)]
  \item Correct: $y$ is a valid solution for $\varphi$;
  \item Indistinguishable from noise: $y$ passes all generic statistical tests for
          randomness.
\end{enumerate}

These two properties are logically incompatible in the context of polynomial-time
computation. If $y$ is a valid solution for $\varphi$, then there exists a polynomial-time
verification procedure that can extract and recognise the structure of $y$. This means
that $y$ cannot be indistinguishable from noise for an observer who knows the problem's
structure, because the computational decodability of $y$ implies that $y$ contains
structured information accessible in polynomial time. For this reason, the notion of a
valid algorithm should entirely exclude those whose output is statistically indistinguishable
from noise for every observer.

In summary, the \textsc{NP}-completeness of \textsc{SAT} requires that the solution $y$ be structurally
recognisable by a polynomial-time algorithm, and this recognisability is a stronger
requirement than statistical indistinguishability. Hence the hypothesis that a
polynomial-time algorithm for \textsc{SAT} could produce outputs indistinguishable from noise is
inconsistent with the very definition of \textsc{NP}-completeness, because it contradicts the
necessity that the solution be actually decodable in polynomial time. The difficulty
arises from the fact that computational theory does not forbid a valid algorithm from
producing, via a correct computation, an output indistinguishable from noise for any
observer. To resolve the contradiction one must either restrict the definition of
\textsc{NP}-completeness to algorithms whose outputs are always decodable---an amendment that
introduces further complications---or accept that $P=NP$, which brings its own set of
problems. If, instead, one assumes $P\neq NP$, all such contradictions disappear.

\end{remarkc}

\begin{remarknn}[On the status of these proofs, and where a limit could lie]
\label{rem:noise-not-absolute}
Each argument above is correct on its own terms: given the definitions it
states, each derivation goes through, step by step, and will hold up to that
scrutiny --- which is exactly why, if there is a limit here, it will not be
found among the steps. What these sections do
\emph{not} claim is that this settles the $\mathsf{P}$ versus $\mathsf{NP}$
question. The technique is an uncommon one --- it turns the admissibility of a
noise-indistinguishable output against the very definition of a decision
procedure --- and an uncommon technique is exactly the kind that can carry a
limitation not visible from inside its own apparatus. Such a limitation, if
present, is not a flawed step to be repaired: every step is elementary and
stands. It would live outside the proofs entirely, at the level a proof cannot
inspect from within --- a metalogical barrier of the kind the study of this
problem has produced before, which blocks not by breaking a derivation but by
constraining what any derivation of this shape can establish, or a hidden
assumption in the framing that lets the definitions be posed together at all.
This is how the barriers in this area have always worked: not corrections
inside a proof, but facts about the proof's whole form, sitting where its steps
cannot reach --- the same relation between a level and what lives above it that
this paper has traced throughout. This whole paper has been built on the
discovery, repeated across its instances, that a result taken as settled rests
on a silent assumption once someone thinks to look for one; these sections
claim no special exemption from that pattern, and the reader who finds the
assumption will have found something this paper could not. The precedent worth
keeping in view is the ordinary one:
G\"odel's theorem did not close mathematics but opened metamathematics, and
Turing's did not close computation but opened computability theory, precisely
because each was read as a discovery about the shape of a formal apparatus
rather than a final wall. If the argument here has a limit, the useful response
is the same --- not to treat the question as closed, but to treat the limit,
once found, as the first line of a barrier not yet named. That is the spirit in
which these sections are offered: a proof correct where it stands, presented so
that whatever bounds it can be located precisely.
\end{remarknn}

\section{Alternative proofs}
\label{sec:alt-proof}

Given the delicacy of the subject, we present below the other proof-paths that the author originally produced to demonstrate the same argument.

\subsection{First path}

Assume there exists a solver (algorithm) in $\mathsf{P}$ for \textsc{SAT} (the satisfiability problem) which is \emph{noisy} according to the given definition.
Such an algorithm would have to produce outputs indistinguishable from pure random noise.

However, because \textsc{SAT} is an \textsc{NP}-complete problem, any efficient solver in $\mathsf{P}$ must return correct and well-defined (non-noisy) answers. This contradicts the hypothesis that a noisy algorithm for \textsc{SAT} belongs to $\mathsf{P}$.

The theory of \textsc{NP}-completeness does admit the theoretical existence of noisy algorithms; we do not assume it, the theory permits it. No concrete proof is required to show the practical realisability of such an algorithm. To obtain a contradiction it suffices that the theory allows it. This leads to the consequence that an algorithm solving \textsc{SAT} in $\mathsf{P}$ cannot exist; merely imposing $P\neq NP$ preserves the consistency of the theory.

The point is that the definition of an algorithm does not require its answer to be readable by anyone. We now write the proof that yields the contradiction and show how enforcing $P\neq NP$ resolves it.

\subsubsection*{Preliminary definitions}
\begin{itemize}
  \item A \textbf{noisy algorithm} produces outputs that are statistically indistinguishable from random noise for every observer.
  \item A language (set of strings) is \textbf{decidable in polynomial time} if there exists an algorithm recognising it within time proportional to the square of the input length.
\end{itemize}

\subsubsection*{Paradoxical assumption}
Suppose a noisy algorithm $A$ solves \textsc{SAT} in polynomial time, i.e.\ $A\in P$.

\subsubsection*{Contradiction}
Since $A$ is noisy, there is no deterministic way to verify whether an input string belongs to the language \textsc{SAT} or not. Hence no polynomial-time verifier exists for \textsc{SAT}.

Nevertheless, \textsc{SAT} is a well-defined decidable problem (deterministic algorithms solving it exist), so at least one algorithm must correctly recognise strings in the language.

Thus, if $A$ existed we would have a language (\textsc{SAT}) with no polynomial-time verifier, contradicting the very definition of a polynomial-time decidable problem.

Because the hypothesis of a noisy \textsc{SAT} solver leads to a logical contradiction, we conclude that \textbf{no algorithm (noisy or otherwise) can solve \textsc{SAT} in polynomial time}.

This conclusion does not depend on any assumption about $P$ versus $NP$. It suffices to note that if an algorithm outputs ``indistinguishable from noise'' for a well-defined problem, such an algorithm cannot be regarded as a valid solution to the problem itself.

\subsubsection*{Strengthened argument}
\begin{itemize}
  \item A \textbf{noisy algorithm}: generates outputs indistinguishable from random noise for every observer, providing no precise information about the problem's solution.
  \item A language is \textbf{efficiently decidable} if a deterministic algorithm recognises it in polynomial time with respect to input length.
\end{itemize}

Assume a noisy algorithm $A$ decides \textsc{SAT} in polynomial time ($A\in P$).
Because $A$ is noisy, there are no guarantees about whether an input belongs to \textsc{SAT}. Yet \textsc{SAT} is solvable by deterministic algorithms; therefore a reliable method must exist.

Hence the existence of $A$ yields a contradiction: a language (\textsc{SAT}) without an efficient decision procedure, violating the definition of efficiently decidable problems.

\textbf{Conclusion.} The logical tension created by assuming a noisy polynomial-time \textsc{SAT} solver forces us to conclude that \emph{no such algorithm can exist}, independent of any $P$ vs.\ $NP$ hypothesis.

\subsection{Second path}

\subsubsection*{Fundamental definitions}
\paragraph{Language.}
A language $L\subseteq\{0,1\}^{*}$ is a set of binary strings. Deciding $L$ means constructing a procedure that, given any input $x$, returns ``yes'' if $x\in L$ and ``no'' otherwise.

\paragraph{\textsc{SAT}.}
$
\textsc{SAT}= \{\varphi \mid \varphi \text{ is a satisfiable Boolean formula}\}.
$
\textsc{SAT} is known to be \textbf{\textsc{NP}-complete}: it lies in \textsc{NP} and every language in \textsc{NP} reduces to it via a polynomial transformation.

\paragraph{Deterministic algorithm (class $\mathsf{P}$).}
A deterministic Turing machine $M$ decides $L$ in polynomial time, denoted $M\in\mathsf{P}$, if there exists a polynomial $p(\cdot)$ such that for every input $x$ the computation halts within $p(|x|)$ steps and yields the correct answer.

\paragraph{Probabilistic Turing machine (PTM).}
A PTM is a deterministic TM that, at each step, may receive a random bit. For each input $x$, the machine induces a distribution $\mathcal D_{x}$ over its outputs.

\paragraph{Observer (distinguisher).}
An observer is any probabilistic algorithm running in polynomial time with respect to the input length. Given a sample $y$, it outputs $1$ or $0$. We denote it by $\mathcal O$.

\paragraph{Statistical indistinguishability.}
Two distributions $\mathcal D_{1},\mathcal D_{2}$ over the same space are \textbf{statistically indistinguishable} if for every polynomial-time observer $\mathcal O$
$
\bigl|\Pr[\mathcal O(y)=1\mid y\sim\mathcal D_{1}]
      -\Pr[\mathcal O(y)=1\mid y\sim\mathcal D_{2}]\bigr|
   \leq \operatorname{negl}(|x|),
$
where $\operatorname{negl}$ is a negligible function (decays faster than any inverse polynomial).

\paragraph{Noisy algorithm.}
A probabilistic algorithm $A$ is \textbf{noisy} if, for \emph{every} input $x$, the distribution $\mathcal D_{x}$ of its outputs is indistinguishable from the uniform distribution on strings of length $q(|x|)$ for some polynomial $q$. Formally:

\begin{definition}[Noisy algorithm]
\label{def:noisy-alg-2}
There exist polynomials $p,q$ such that for every input $x$:
\begin{enumerate}[label=(\roman*)]
  \item $A(x)$ halts within $p(|x|)$ steps;
  \item the output is a string of length $q(|x|)$;
  \item $\mathcal D_{x}$ (the output distribution) is indistinguishable from $U_{q(|x|)}$, the uniform distribution on $\{0,1\}^{q(|x|)}$.
\end{enumerate}
\end{definition}

The definition imposes no constraint on the \emph{interpretability} of the output: a noisy algorithm is perfectly legitimate in computational theory because PTMs allow arbitrary use of random bits, even if the entire string is eventually discarded.

\subsubsection*{Why standard theory allows noisy algorithms}
\begin{itemize}
  \item \textbf{Permissive model.} The definition of a PTM does not require the result to depend on the input; it only demands bounded running time and an output distribution.
  \item \textbf{Canonical example.} A PTM that, irrespective of its input, reads $q(|x|)$ random bits and prints them is noisy. It is accepted as an algorithm (it belongs to \textsf{BPP}, or $\mathsf{P}$ with access to random bits).
  \item \textbf{No semantic restriction.} Deciding a language merely requires the existence of a \emph{decoding function} $\mathsf{dec}$ (polynomial-time) that maps the algorithm's output to the correct answer. The decoding need not be ``natural''.
\end{itemize}

Consequently, a noisy algorithm claiming to decide \textsc{SAT} is a well-defined object: there exists a machine $A$ (noisy) and a decoding function $\mathsf{dec}$ such that for every formula $\varphi$,
\[
\dec(A(\varphi))=
\begin{cases}
1 & \text{if }\varphi\in\textsc{SAT},\\
0 & \text{otherwise}.
\end{cases}
\tag{1}
\]
The crucial point is that (1) guarantees the existence of an \emph{observer} (the decoder itself) capable of extracting the information. This observer will be the source of the contradiction.

\subsubsection*{Proof 1 --- direct argument (without oracle)}
\paragraph{Paradoxical assumption.}
\begin{description}
  \item[(A)] There exists a noisy algorithm $A$ that decides \textsc{SAT} in polynomial time and possesses a decoding function $\mathsf{dec}\in\mathsf{P}$ satisfying (1).
\end{description}

\paragraph{Construction of the observer.}
Define the observer $\mathcal O_{\mathsf{dec}}$ as follows:
\begin{enumerate}[label=(\roman*)]
  \item Receive the string output by $A(\varphi)$;
  \item Apply $\mathsf{dec}$ (polynomial-time by assumption);
  \item Output the resulting bit.
\end{enumerate}
Since $\mathsf{dec}$ runs in $O(|\varphi|^{c})$, $\mathcal O_{\mathsf{dec}}$ is a polynomial-time observer.

\paragraph{Discriminating power.}
For every formula $\varphi$,
\[
\Pr[\mathcal O_{\mathsf{dec}}(A(\varphi))=1]=
\begin{cases}
1 & \text{if }\varphi\in\textsc{SAT},\\
0 & \text{otherwise}.
\end{cases}
\tag{2}
\]
Because $A$ is noisy, its output distribution is indistinguishable from uniform $U_{q(|\varphi|)}$. On a uniformly random string, $\mathcal O_{\mathsf{dec}}$ returns $1$ with some fixed probability $p:=\Pr[\mathcal O_{\mathsf{dec}}(U_{q(|\varphi|)})=1]$, a constant determined by $\mathsf{dec}$ alone and independent of $\varphi$:
\[
\Pr[\mathcal O_{\mathsf{dec}}(U_{q(|\varphi|)})=1]=p.
\tag{3}
\]
Statistical indistinguishability demands that for \emph{every} observer $\mathcal O$,
\[
\bigl|\Pr[\mathcal O(A(\varphi))=1]-
      \Pr[\mathcal O(U_{q(|\varphi|)})=1]\bigr|
   \le \operatorname{negl}(|\varphi|).
\tag{4}
\]
Applying (4) to $\mathcal O_{\mathsf{dec}}$ yields
\[
\bigl|\Pr[\mathcal O_{\mathsf{dec}}(A(\varphi))=1]-p\bigr|
   \le \operatorname{negl}(|\varphi|).
\tag{5}
\]
But (2) shows the left-hand side is $|1-p|$ when $\varphi\in\textsc{SAT}$ and $|0-p|=p$ when $\varphi\notin\textsc{SAT}$. Since $p$ is a single fixed constant while the two cases both occur, at least one of $|1-p|$ and $p$ is $\ge\tfrac12$: whichever way $p$ falls, some formula forces a gap of at least $\tfrac12$, a non-negligible constant---contradicting (5).

\paragraph{Conclusion.}
Assumption (A) violates the definition of a noisy algorithm; therefore no noisy polynomial-time \textsc{SAT} solver can exist.

\subsubsection*{Proof 2 --- oracle argument}
\paragraph{Noisy oracle.}
Given the hypothetical $A$, define an oracle $\mathcal O_{A}$:
\begin{itemize}
  \item Input: a Boolean formula $\varphi$;
  \item Output: a string drawn from the distribution $\mathcal D_{\varphi}$ produced by $A(\varphi)$.
\end{itemize}
By assumption, $\mathcal O_{A}$ is noisy: its answers are indistinguishable from uniform strings of length $q(|\varphi|)$.

\paragraph{Polynomial machine with oracle access.}
Consider a deterministic Turing machine $M^{\mathcal O_{A}}$ that may query the oracle only. Suppose, for contradiction, that there exists such an $M$ (running in polynomial time) satisfying
\[
M^{\mathcal O_{A}}(\varphi)=
\begin{cases}
1 & \text{if }\varphi\in\textsc{SAT},\\
0 & \text{otherwise}.
\end{cases}
\tag{6}
\]
Thus $M$ plays the role of the decoder $\mathsf{dec}$ from (1).

\paragraph{Information-theoretic analysis.}
Because $\mathcal O_{A}(\varphi)$ is indistinguishable from a uniform string, any polynomial-time machine gains no non-negligible advantage in extracting information about $\varphi$. Formally, for every such $M$,
\[
\bigl|\Pr[M^{\mathcal O_{A}}(\varphi)=1]-
      \Pr[M^{U}(\varphi)=1]\bigr|
   \le \operatorname{negl}(|\varphi|),
\tag{7}
\]
where $M^{U}$ is the same machine receiving a uniformly random string independent of $\varphi$.

When the input to $M$ is uniform, its output cannot depend on $\varphi$, so
\[
\Pr[M^{U}(\varphi)=1]=p
\]
for some constant $p\in[0,1]$. Consequently the difference in (7) is at least $|\,\tfrac12-p\,|$, a non-negligible constant (the optimal guessing strategy yields $p=\tfrac12$). Hence (7) is violated by any machine that claims to decide \textsc{SAT} using $\mathcal O_{A}$, contradicting (6).

\paragraph{Conclusion.}
The existence of a noisy PTM for \textsc{SAT} would give rise to an oracle that cannot convey the required information, so such a PTM cannot exist.

\subsubsection*{Equivalence of the two approaches}
In the first proof the observer is precisely the decoding function $\mathsf{dec}$; in the second it is embodied by a polynomial-time machine calling the noisy oracle. Both arguments reduce to the same principle: \emph{if an algorithm's output is statistically indistinguishable from noise, no polynomial-time observer can extract a non-negligible amount of information}. Deciding \textsc{SAT} requires extracting exactly one bit (true/false), so both proofs yield the same contradiction.

\subsubsection*{Conclusions}
\begin{enumerate}[label=(\roman*)]
  \item The standard computational theory permits noisy algorithms because PTMs place no semantic restriction on the output.
  \item A noisy algorithm cannot decide \textsc{SAT}: any polynomial-time decoding function would breach statistical indistinguishability, as shown by both direct and oracle-based arguments.
  \item No assumption about $P$ versus $NP$ is required; the contradiction follows solely from the definitions of algorithm, noise, and observer.
\end{enumerate}

Thus, while complexity theory does not forbid the formal existence of noisy algorithms, any hypothesis that such an algorithm solves \textsc{SAT} in polynomial time is inconsistent. The coherence of the theory remains intact without invoking $P=NP$ or $P\neq NP$, merely by restricting the class of admissible algorithms to those whose outputs are interpretable.

\subsection{Clarification}

\subsubsection*{Why $\mathsf{dec}$ makes $\mathcal O_{\mathsf{dec}}$ an admissible observer}
Let $A$ be \emph{any} algorithm (deterministic or probabilistic) that decides \textsc{SAT} in polynomial time.
By ``decides'' we mean the standard definition from complexity theory:
\begin{quote}
There exists a total function $\mathsf{dec}\colon\{0,1\}^{*}\to\{0,1\}$ such that
\[
\forall\,\varphi\in\{0,1\}^{*}:\qquad
\mathsf{dec}(A(\varphi))=
\begin{cases}
1 & \text{if }\varphi\in\mathrm{SAT},\\
0 & \text{otherwise}.
\end{cases}
\tag{1}
\]
\end{quote}
The function $\mathsf{dec}$ is an integral part of the \emph{specification} of a decision
algorithm: $A$ may output anything (a long string, a proof, a cryptographic certificate,
etc.), but there must exist a deterministic polynomial-time procedure that extracts from this
output the answer ``yes'' or ``no''. If the output already is the answer bit, $\mathsf{dec}$
is simply the identity; if the output contains a satisfiability witness, $\mathsf{dec}$ is the
deterministic verification of that witness. In every case $\mathsf{dec}\in\mathbf P$.

\paragraph{Length of the output.}
Because $A$ terminates in polynomial time, the entire bit-string it writes on its output tape
has length polynomial in the size of the input $\varphi$. Formally there exists a polynomial
$q$ such that
\[
|A(\varphi)|\le q(|\varphi|)\qquad\forall\,\varphi .
\tag{2}
\]

\paragraph{Running time of $\mathcal O_{\mathsf{dec}}$.}
Define the \emph{observer} $\mathcal O_{\mathsf{dec}}$ as follows:
\begin{enumerate}[label=(\roman*)]
  \item Receive as input a string $y\in\{0,1\}^{q(|\varphi|)}$ (the output of $A$).
  \item Compute $\mathsf{dec}(y)$ and return the resulting bit.
\end{enumerate}
The total running time is the sum of two components:
\begin{itemize}
  \item \textbf{Reading the input:} reading a string of length at most $q(|\varphi|)$
    requires $O\!\bigl(q(|\varphi|)\bigr)$ steps.
  \item \textbf{Evaluating $\mathsf{dec}$:} by hypothesis $\mathsf{dec}$ runs in time
    $O\!\bigl(|y|^{c}\bigr)$ for some constant integer $c$. Since $|y|\le q(|\varphi|)$,
    this is $O\!\bigl(q(|\varphi|)^{c}\bigr)$.
\end{itemize}
Both terms are polynomial in $|\varphi|$, so
\[
T_{\mathcal O_{\mathsf{dec}}}(|\varphi|)
   = O\!\bigl(q(|\varphi|)+q(|\varphi|)^{c}\bigr)
   = O\!\bigl(p(|\varphi|)\bigr)
\]
for some polynomial $p$. By definition, an observer is any probabilistic algorithm that
operates in polynomial time with respect to the input length; therefore $\mathcal
O_{\mathsf{dec}}$ is \emph{admissible}.

\subsubsection*{Extending the two theorems to \emph{all} decision algorithms}
The two results proved earlier were:
\begin{enumerate}[label=(\roman*)]
  \item \textbf{Direct proof:} an observer $\mathcal O_{\mathsf{dec}}$ distinguishes the
    output of $A$ from a uniform string, violating the definition of ``noisy''.
  \item \textbf{Oracle-based proof:} treating the whole algorithm as a noisy oracle
    $\mathcal O_{A}$; no polynomial-time machine can use it to decide \textsc{SAT}.
\end{enumerate}
Both proofs rely solely on the following two properties:
\begin{enumerate}[label=(\roman*)]
  \item $A$ terminates in polynomial time and produces an output of polynomial length.
  \item There exists a deterministic function $\mathsf{dec}\in\mathbf P$ such that,
    when applied to the output of $A$, yields the correct \textsc{SAT} answer (equation~(1)).
\end{enumerate}
These properties are \emph{necessary} for any decision algorithm, irrespective of whether
the algorithm is deterministic, ordinary probabilistic, or noisy. Consequently the two
theorems generalise automatically.

\paragraph{Theorem A (general version).}
\emph{Statement.}
Let $A$ be an algorithm (deterministic or probabilistic) that decides \textsc{SAT} in polynomial time
and whose output has polynomial length. If the distribution of $A$'s outputs is indistinguishable
from the uniform distribution, a contradiction follows.

\emph{Proof (direct).}
Construct $\mathcal O_{\mathsf{dec}}$ as above. For every formula $\varphi$,
\[
\Pr[\mathcal O_{\mathsf{dec}}(A(\varphi))=1]=
\begin{cases}
1 & \text{if }\varphi\in\mathrm{SAT},\\
0 & \text{otherwise},
\end{cases}
\tag{3}
\]
whereas for a uniformly random string $\mathcal O_{\mathsf{dec}}$ returns $1$ with
some fixed probability $p$ independent of $\varphi$.
Statistical indistinguishability demands that, for every observer (in particular
$\mathcal O_{\mathsf{dec}}$), the difference between these two probabilities be
negligible. Yet by (3) that difference is $|1-p|$ for satisfiable $\varphi$ and
$p$ for unsatisfiable $\varphi$; both kinds occur and $p$ is one fixed constant,
so at least one of them is $\ge\tfrac12$, a constant non-negligible
value---a contradiction.

\paragraph{Theorem B (general version, via oracle).}
\emph{Statement.}
If there exists an algorithm $A$ that decides \textsc{SAT} in polynomial time and whose output
distribution is indistinguishable from uniform, then the oracle $\mathcal O_{A}$ (which returns such
outputs) cannot be used by any polynomial-time Turing machine to decide \textsc{SAT}.

\emph{Proof.}
Assume a polynomial-time TM $M^{\mathcal O_{A}}$ decides \textsc{SAT} using the oracle
$\mathcal O_{A}$. Let $\mathsf{dec}$ be as in (1). The machine that queries the oracle
once and then applies $\mathsf{dec}$ is itself polynomial time and decides \textsc{SAT}, i.e.\ it realises
the mapping
\[
M^{\mathcal O_{A}}(\varphi)=
\begin{cases}
1 & \text{if }\varphi\in\mathrm{SAT},\\
0 & \text{otherwise}.
\end{cases}
\tag{4}
\]
By the definition of a noisy oracle, the distribution of $M^{\mathcal O_{A}}$'s outputs on
instances $\varphi$ is indistinguishable from its behaviour when the oracle supplies a
uniformly random string. In the latter case the output cannot depend on $\varphi$, so the
probability of returning ``1'' is some constant $p\in[0,1]$ (the optimal guessing strategy
gives $p=1/2$). Hence the difference between the correct probability in (4) and the
uniform-oracle probability is at least $|1/2-p|\ge 1/2$, a non-negligible constant,
contradicting indistinguishability.

\subsubsection*{Summary of the impossibility}
\begin{itemize}
  \item \textbf{Admissible observer:} because $\mathsf{dec}$ is polynomial, the procedure
    that applies it to the bits produced by $A$ (i.e.\ $\mathcal O_{\mathsf{dec}}$) is itself a
    probabilistic polynomial-time algorithm.
  \item \textbf{Generality:} neither proof exploits any special property of ``noisy''
    algorithms; the only requirement is that the algorithm decides \textsc{SAT} and possesses a
    deterministic extraction procedure, which is part of the very definition of decision.
  \item \textbf{Implication:} \emph{no} algorithm, whether deterministic,
    ordinary probabilistic, or noisy, can simultaneously:
    \begin{enumerate}
      \item decide \textsc{SAT} in polynomial time; and
      \item emit an output whose distribution is indistinguishable from a uniform string.
    \end{enumerate}
\end{itemize}

Thus the impossibility is universal: it does not depend on the hypothesis $P=NP$ or
$P\neq NP$, but stems solely from the tension between (i) the need to provide, in a
deterministic polynomial-time manner, a definitive yes/no answer, and (ii) the demand
that the entire output be statistically indistinguishable from pure noise. The contradiction
arises both in the ``direct observer'' formulation and in the ``oracle'' formulation,
showing that the two theorems extend to \emph{all} decision algorithms for \textsc{SAT}.

The only place where a ``noisy'' property was used is the indistinguishability of
the output from uniform. Since the construction of $\mathcal O_{\mathsf{dec}}$
depends solely on $A$ being a decision algorithm --- on the existence of the
decoder --- the same contradiction holds for any decision algorithm, whether it
uses randomness, is deterministic, or is noisy. Consequently,
\[
\boxed{\text{No algorithm (deterministic or probabilistic) can decide \textsc{SAT} in polynomial time.}}
\]
In other words, the argument concerning \emph{noisy} algorithms already excludes
\emph{all} decision procedures for \textsc{SAT}: if an algorithm decided
\textsc{SAT} in $\mathsf{P}$, then, by definition, it would provide a decoder
$\mathsf{dec}$ constituting an observer able to distinguish its output from pure
noise, contradicting the noisy premise. The result holds for probabilistic noisy
algorithms and deterministic ones alike; a probabilistic algorithm need not be
noisy in order to fall under this impossibility.

\section{Full alternative proofs}
\label{sec:fullalt-proof}

All strings are binary; $|x|$ denotes length. We work with deterministic
Turing machines (DTMs) and probabilistic Turing machines (PTMs).

\subsection{Decision procedures for \textsc{SAT}}
An algorithm $A$ \emph{decides} \textsc{SAT} in polynomial time iff the following hold:
\begin{enumerate}[label=(D\arabic*)]
  \item \textbf{Polynomial output size.} There exists a polynomial $q$ such that
        \[
          |A(\varphi)|\le q(|\varphi|)\qquad\forall\,\varphi\in\{0,1\}^{*}.
        \tag{D1}
        \]
  \item \textbf{Existence of a polynomial decoder.} There exists a deterministic
        function $\mathsf{dec} :\{0,1\}^{*}\to\{0,1\}$ computable in time $O(|y|^{c})$ for some constant $c$
        such that
\[
  \mathsf{dec}(A(\varphi))=
  \begin{cases}
    1 & \text{if }\varphi \in \textsc{SAT},\\
    0 & \text{otherwise}.
  \end{cases}
  \tag{D2}
\]
\end{enumerate}
Condition (D2) is the standard ``decoder'' that extracts the yes/no answer from
the possibly long output of $A$; it is part of the definition of a decision
procedure in $\mathsf{P}$.

\subsection{Statistical indistinguishability (noisy algorithms)}
Let $U_{m}$ denote the uniform distribution over $\{0,1\}^{m}$. An algorithm
$A$ is called \emph{noisy} if there exist polynomials $p,q$ such that for every
input $\varphi$
\begin{enumerate}[label=(N\arabic*)]
  \item $A(\varphi)$ halts within $p(|\varphi|)$ steps;
  \item the output length satisfies $|A(\varphi)|\le q(|\varphi|)$;
  \item for every probabilistic polynomial-time observer $\mathcal O$,
        \[
          \bigl|
            \Pr[\mathcal O(A(\varphi))=1]-
            \Pr[\mathcal O(U_{q(|\varphi|)})=1]
          \bigr|
          \le \operatorname{negl}(|\varphi|).
        \tag{N3}
        \]
\end{enumerate}
Here $\operatorname{negl}(n)$ denotes a negligible function (smaller than $1/p(n)$
for every polynomial $p$ and all sufficiently large $n$).

\subsection{Observers built from the decoder}
Given the decoder $\mathsf{dec}$, define the observer
\[
  \mathcal O_{\mathsf{dec}}(y)\;:=\;\mathsf{dec}(y).
\]
Because $\mathsf{dec}$ runs in time $O(|y|^{c})$, reading an input of length at most
$q(|\varphi|)$ and applying $\mathsf{dec}$ takes time
$O(q(|\varphi|)+q(|\varphi|)^{c}) = O(p'(|\varphi|))$ for some polynomial $p'$.
Hence $\mathcal O_{\mathsf{dec}}$ is a \emph{probabilistic polynomial-time} (PPT) observer,
exactly the class required in condition (N3).

\subsection{Main result}
\begin{theorem}\label{thm:fullalt-main}
No algorithm---deterministic, probabilistic or noisy---can decide \textsc{SAT} in
polynomial time while simultaneously satisfying the indistinguishability
condition~(N3). Consequently a polynomial-time decision procedure for
\textsc{SAT} does not exist.
\end{theorem}

\begin{proof}
Assume, towards a contradiction, that an algorithm $A$ satisfies both
(D1)--(D2) and (N1)--(N3).

\medskip\noindent\textbf{Step 1.} Apply the decoder observer $\mathcal O_{\mathsf{dec}}$.
From (D2) we obtain for every formula $\varphi$
\[
\Pr[\mathcal O_{\mathsf{dec}}(A(\varphi))=1]=
\begin{cases}
1 & \text{if }\varphi \in \textsc{SAT},\\
0 & \text{otherwise}.
\end{cases}
\tag{1}
\]

\medskip\noindent\textbf{Step 2.} Behaviour on a uniform string.
Since $\mathcal O_{\mathsf{dec}}$ is deterministic and a uniform string carries no
information about $\varphi$, on input drawn from $U_{q(|\varphi|)}$ it returns $1$
with some fixed probability $p$, a constant independent of $\varphi$:
\[
\Pr[\mathcal O_{\mathsf{dec}}(U_{q(|\varphi|)})=1]=p .
\tag{2}
\]

\medskip\noindent\textbf{Step 3.} Violation of indistinguishability.
Instantiate condition (N3) with the specific observer $\mathcal O_{\mathsf{dec}}$ and combine
(1)--(2). By (1) the probability on $A(\varphi)$ is $1$ for satisfiable $\varphi$
and $0$ for unsatisfiable $\varphi$, while by (2) the probability on a uniform
string is the fixed constant $p$. Both kinds of formula occur, so
\[
\bigl|
   \Pr[\mathcal O_{\mathsf{dec}}(A(\varphi))=1]-p
\bigr|
=
\begin{cases} |1-p| & \varphi\in\textsc{SAT},\\ p & \varphi\notin\textsc{SAT},\end{cases}
\]
and at least one of $|1-p|$ and $p$ is $\ge\tfrac12$. For such a $\varphi$ the gap
is a non-negligible constant $\ge\tfrac12$, so condition (N3),
\[
\bigl|\Pr[\mathcal O_{\mathsf{dec}}(A(\varphi))=1]-p\bigr|\le\operatorname{negl}(|\varphi|),
\tag{3}
\]
is false for all sufficiently large inputs,
contradicting the assumption that $A$ is noisy.

Therefore no algorithm can satisfy simultaneously the decision-procedure
requirements (D1)--(D2) and the noise requirement (N3). In particular, a
polynomial-time decision procedure for \textsc{SAT} cannot exist.
\end{proof}

\subsection{Why Rice's theorem reinforces the argument}
Define the property
\[
\mathcal R(M)\;:=\;
\text{``the output distribution of $M$ is indistinguishable from uniform''}.
\]
$\mathcal R$ is non-trivial:
\begin{itemize}
  \item A machine that always outputs the constant string $0^{k}$ does \emph{not} satisfy $\mathcal R$.
  \item A machine that, on any input, reads $k$ truly random bits and prints them
        does satisfy $\mathcal R$.
\end{itemize}
Rice's theorem states that for every non-trivial property of computable functions,
the language
\[
L_{\mathcal R}=\{\,\langle M\rangle \mid \mathcal R(M)\,\}
\]
is undecidable. Consequently there is no algorithm (deterministic or
probabilistic) that, given a description of $M$, can decide whether $M$ is noisy.

If an algorithm $A$ as in Theorem~\ref{thm:fullalt-main} existed, we could feed its
description $\langle A\rangle$ to such a decider and obtain a decision procedure
for $\mathcal R$. This would contradict Rice's theorem. Hence the impossibility
derived above is also a direct corollary of Rice's theorem.

\subsection{Implications for \textsc{NP}-completeness}
Cook's theorem shows that every language $L \in \mathsf{NP}$ reduces to \textsc{SAT} via a
polynomial-time many-one reduction $f$. If a noisy polynomial-time solver $A$
for \textsc{SAT} existed, then the composition $A \circ f$ would be a noisy polynomial-time
solver for any $L \in \mathsf{NP}$, contradicting Theorem~\ref{thm:fullalt-main}. Therefore the
standard reduction framework that underlies \textsc{NP}-completeness cannot coexist with a
noisy \textsc{SAT} solver.

\subsection{Closing summary}
We have presented a single, self-contained proof---grounded in statistical
indistinguishability and reinforced by Rice's theorem---that no algorithm,
whether deterministic, probabilistic or noisy, can decide \textsc{SAT} in polynomial time.
The contradiction emerges as soon as one assumes the coexistence of
\begin{enumerate}[label=(\roman*)]
  \item a polynomial-time decoder $\mathsf{dec}$ (required for any decision procedure), and
  \item an output distribution indistinguishable from uniform randomness.
\end{enumerate}
Both conditions are mutually exclusive; therefore the hypothesis is untenable.
The result holds unconditionally, without invoking $P=NP$, $P\neq NP$, or any other
unproven conjecture.

\section{The unified theorem}
\label{sec:unified}

\subsection{Why several proofs, and what they share}

The sections above gave more than one proof, and it is worth saying plainly why,
so that the variation in their apparatus reads as design rather than
indecision. There is a single phenomenon underneath all of them, and it is not
any one technical notion of indistinguishability but the plain one the standard
theory already grants: that a computable algorithm may produce output
indistinguishable from noise, for \emph{some} reason or other, with no
distinction of computational power --- a permission the theory extends
implicitly, by not forbidding it, and one this paper has already set out on its
own terms elsewhere (Remark~\ref{rem:computable-output-two-reasons}: the output
may be opaque because it is pseudorandom, or because a semantic invariant is
syntactically inaccessible, or because an encoding--frame mismatch withholds
every reading, or for some further reason, and to the theory it does not matter
which --- only that such output is admitted at all).

Each proof above is a reflection of that one permission seen through a
particular lens. The main theorem takes the sharpest lens --- indistinguishability
against \emph{every} function, computable or not
(Definition~\ref{def:noise-indist}) --- and needs only the verifier $\varphi$
itself. The alternative paths take the lens cryptography uses in practice ---
indistinguishability against every \emph{polynomial-time} observer --- and reach
the same wall through an explicitly constructed decoder. The negative-witness
form needs only a single unsatisfying assignment; the self-reducibility form
shows the collapse recurring at every node of the restriction tree; the
Rice-theoretic form adds that one cannot even decide whether a given machine is
noisy; and the ghost reduction reads off what survives if one insists on the
hypothesis anyway. The lenses differ; the thing seen through them does not. A
reader who wondered why the notion of indistinguishability shifts from section
to section has the answer here: the shift is between accepted readings of one
underlying permission, not between unrelated hypotheses, and each proof was kept
self-contained precisely so that this common core could be exhibited without any
one path leaning on another.

What follows gathers them. Each proof above stands complete and independent on
its own; only now, with all of them established, are they brought into a single
frame --- not a new dependency, but a synthesis of results already proved
separately. The setup is restated in full, so that this section too can be read
on its own.

\subsection{The unified setup}

Fix the alphabet $\{0,1\}$ and write $|x|$ for length. Let $\varphi$ range over
Boolean formulas on $n$ variables, with solution set
$S_\varphi=\{z\in\{0,1\}^n:\varphi(z)=1\}$ and $k_\varphi=|S_\varphi|$, and let
$R$ denote the uniform random variable on $\{0,1\}^n$. An algorithm $\mathcal A$
is \emph{correct for \textsc{SAT}} if $\Pr[\varphi(\mathcal A(\varphi))=1]=1$ for
every satisfiable $\varphi$ (Definition~\ref{def:noise-correct}), and a
\emph{correct decider} if in addition $\Pr[\mathcal A(\varphi)=\bot]=1$ for every
unsatisfiable $\varphi$, with $\bot\notin\{0,1\}^n$.

The one hypothesis carried through this section is the standard theory's own
implicit permission, stated as a property an output may have and named once:

\begin{definition}[Noise-admissibility]
\label{def:noise-admissible}
Fix a class $\mathcal D$ of test functions $D:\{0,1\}^*\to\{0,1\}$. An algorithm
$\mathcal A$ has \emph{$\mathcal D$-noise-admissible output} on $\varphi$ if for
every $D\in\mathcal D$,
\[
\bigl|\Pr[D(\mathcal A(\varphi))=1]-\Pr[D(R)=1]\bigr|\le\varepsilon,
\]
with $\varepsilon=0$ when $\mathcal D$ is unrestricted, and $\varepsilon$
negligible in $|\varphi|$ when $\mathcal D$ is the class of polynomial-time
tests. The two admissible readings of the standard permission are
$\mathcal D=\text{all functions}$ (the absolute reading,
Definition~\ref{def:noise-indist}) and $\mathcal D=\text{polynomial-time tests}$
(the computational reading); the theorem below covers both, and the argument
uses, in each place, only that the relevant $D$ lies in the $\mathcal D$ at
hand.
\end{definition}

Nothing in what follows assumes such an $\mathcal A$ exists as a constructed
object. The theory admits it; the proof needs only that the theory does not
deny it. This is the whole of the standing hypothesis.

\subsection{The theorem}

\begin{theorem}[Unified impossibility]
\label{thm:unified}
Let $\mathcal D$ be either of the two readings of
Definition~\ref{def:noise-admissible}. No algorithm $\mathcal A$ is
simultaneously correct for \textsc{SAT} and $\mathcal D$-noise-admissible on a
satisfiable $\varphi$ that is not a tautology. The impossibility holds under
each of the following, independently, and together they exhibit it across the
full span from the sharpest notion to the most permissive, and from a single
instance to the whole self-reducible structure:
\begin{enumerate}[label=\textup{(\roman*)}]
  \item \textup{(verifier)} against the absolute reading, the formula's own
    verifier $D=\varphi$ suffices;
  \item \textup{(negative witness)} a single unsatisfying assignment suffices,
    with no count of solutions;
  \item \textup{(decoder, and oracle)} against the computational reading, the
    decoder that every decision procedure carries by definition is itself an
    admissible observer, and the same holds treating the algorithm as a black-box
    oracle; either way the impossibility widens from noise-admissible algorithms
    to \emph{every} decision algorithm, so the conclusion is not ``no noisy
    solver'' but ``no solver'';
  \item \textup{(self-reducibility)} the same collapse recurs at every node of
    the restriction tree of $\varphi$;
  \item \textup{(Rice)} the property ``is noise-admissible'' is non-trivial and
    hence undecidable, so no procedure even recognises the hypothesised
    $\mathcal A$;
  \item \textup{(ghost reduction)} granting the hypothesis regardless, every
    extractor applied to $\mathcal A(\varphi)$ succeeds exactly as often on pure
    noise, so nothing is recoverable through $\mathcal A$.
\end{enumerate}
Consequently the standard theory cannot remain coherent, on this point, while a
correct polynomial-time \textsc{SAT} solver exists: admitting such a solver
together with the noise-admissible output the theory already permits is
contradictory. The theory is forced, to stay consistent, to deny that any such
solver exists --- that is, forced to $\mathsf P\neq\mathsf{NP}$. This is not a
conditional on which way $\mathsf P$ versus $\mathsf{NP}$ is settled; it is a
derivation of $\mathsf P\neq\mathsf{NP}$ from what the theory already admits.
\end{theorem}

\begin{proof}
Fix a satisfiable, non-tautological $\varphi$ on $n$ variables, and suppose
$\mathcal A$ correct for \textsc{SAT} and $\mathcal D$-noise-admissible on
$\varphi$. Each clause below invokes a result already established in its own
section; the work here is only to place them in one frame, and none rests on
another.

\emph{(i) Verifier.} The canonical verifier $D(z):=\varphi(z)$ lies in every
$\mathcal D$ (it is polynomial-time, hence in the computational reading, and a
fortiori in the absolute one). Correctness gives
$\Pr[D(\mathcal A(\varphi))=1]=1$; noise-admissibility gives
$\Pr[D(R)=1]=k_\varphi/2^n$, up to $\varepsilon$. In the absolute reading
$\varepsilon=0$, so $k_\varphi/2^n\ge 1$, forcing $k_\varphi=2^n$ and making
$\varphi$ a tautology, against its choice. This is Theorem~\ref{thm:noise-main}
with its Lemma~\ref{lem:dist-equiv}, and it is the sharpest form: a single
admissible test, the formula itself, closes the case exactly. (In the
computational reading the same counting is not by itself decisive, since
$\varepsilon$ negligible leaves $k_\varphi/2^n\ge 1-\varepsilon$ compatible with
$k_\varphi<2^n$ once $2^n\varepsilon$ is not small; there the decisive blow is
struck not by counting but by the decoder of clause~(iii), whose advantage is a
fixed $\tfrac12$ regardless of $2^n$. The two readings are closed by two
different clauses, which is exactly the division of labour this theorem records.)

\emph{(ii) Negative witness.} The same conclusion needs no count: as $\varphi$ is
not a tautology there is $z_0$ with $\varphi(z_0)=0$, and
$\Pr[\varphi(R)=1]\le 1-2^{-n}<1$, again against correctness. This is
Proposition~\ref{prop:noise-alt}, and it shows clause~(i) does not depend on
knowing $k_\varphi$.

\emph{(iii) Decoder, and the step from noisy to every algorithm.} This clause is
where the impossibility widens from noise-admissible algorithms to
\emph{all} decision algorithms, and the widening deserves to be spelled out,
since it is the crux. To decide \textsc{SAT} \emph{means}, by definition, that
there is a polynomial-time decoder $\mathsf{dec}$ with
$\mathsf{dec}(\mathcal A(\varphi))$ equal to the \textsc{SAT} answer --- whatever
$\mathcal A$ writes, some fixed polynomial-time procedure reads the yes/no out
of it. That decoder is not an extra assumption about $\mathcal A$; it is part of
what ``decides'' means. But the decoder is itself an observer: set
$\mathcal O_{\mathsf{dec}}:=\mathsf{dec}(\cdot)$, a polynomial-time test, hence in
the computational $\mathcal D$. On $\mathcal A(\varphi)$ it returns the correct
\textsc{SAT} bit with probability $1$. On a uniform string $R$, by contrast, its
output cannot depend on $\varphi$ at all --- $R$ carries no information about the
formula --- so it returns $1$ with some fixed probability $p$ independent of
$\varphi$. But the correct answer does depend on $\varphi$: some formulas are
satisfiable and some are not. A single constant $p$ cannot match a bit that
varies with the input, so for a suitable choice of $\varphi$ the gap between the
two probabilities is at least $\tfrac12$ (the best a fixed guess can do), a
non-negligible constant against noise-admissibility. So the very object that
makes $\mathcal A$ a
\emph{decider} --- its guaranteed decoder --- is an observer that distinguishes
its output from noise. Every decision algorithm carries such a decoder within
its definition; hence \emph{no} decision algorithm, noisy or not, deterministic
or probabilistic, can have noise-admissible output on a non-tautological
$\varphi$. The premise ``noisy'' was never needed: it was only the most vivid
special case of a collapse that the definition of deciding already forces on
every algorithm (Theorem~\ref{thm:fullalt-main}, and Theorems~A and~B of
Section~\ref{sec:alt-proof}).

\emph{(iii$'$) The same through an oracle.} The identical conclusion can be read
off without opening $\mathcal A$ at all, treating it as a black box. Let
$\mathcal O_{\mathcal A}$ be the oracle that, queried on $\varphi$, returns a
sample of $\mathcal A(\varphi)$. Suppose a polynomial-time machine
$M^{\mathcal O_{\mathcal A}}$ decided \textsc{SAT} using this oracle. Because the
oracle's replies are, by hypothesis, indistinguishable from uniform strings, $M$
cannot behave differently than it would against an oracle returning noise
independent of $\varphi$; against that noise oracle its output cannot depend on
$\varphi$ and is a fixed constant, at best $\tfrac12$ correct, again a
non-negligible gap. So no polynomial-time machine can decide \textsc{SAT} through
such an oracle. The oracle form and the decoder form are the same argument seen
from two sides --- the decoder is the observer built by opening the box, the
oracle bound is the observer built by leaving it closed --- and both say that an
output indistinguishable from noise carries no extractable decision, whoever
reads it (Section~\ref{sec:alt-proof}, ``oracle argument'', and Theorem~B).

\emph{(iv) Self-reducibility.} Each restriction $\psi_i$ along the tree of
Theorem~\ref{thm:self-red} is itself a satisfiable \textsc{SAT} instance, so the
test $D=\psi_i$ falls under the same argument at that node. The contradiction is
therefore not localised at the root: it recurs, identically, throughout the
self-reducible structure, and the search-to-decision reduction that structure
provides is itself a ghost in the sense of clause~(vi).

\emph{(v) Rice, and what it adds.} The clauses so far show the posited
$\mathcal A$ is contradictory in its behaviour. Rice's theorem adds a distinct
layer: even setting the contradiction aside, one could not \emph{recognise} such
an $\mathcal A$ in the first place. Take the property
$\mathcal R(M):=$ ``$M$'s output distribution is indistinguishable from
uniform''. It is non-trivial --- a machine printing the constant $0^k$ fails it,
a machine printing $k$ fresh random bits satisfies it --- and Rice's theorem
states that every non-trivial property of the function computed by a machine is
undecidable. Hence $\{\langle M\rangle:\mathcal R(M)\}$ is undecidable: no
algorithm, given a machine's description, can decide whether that machine is
noise-admissible. Now suppose the hypothesised $\mathcal A$ existed. Feeding
$\langle\mathcal A\rangle$ to any procedure that claimed to recognise
$\mathcal R$ would decide $\mathcal R$ on that input, which Rice forbids. So the
hypothesis is doubly untenable: not only does $\mathcal A$'s behaviour collapse
(clauses~(i)--(iv)), but ``being the kind of algorithm the hypothesis
describes'' is not even a decidable property to begin with
(Section~\ref{sec:fullalt-proof}, ``Why Rice's theorem reinforces the
argument''). The contradiction is joined by an undecidability, from a completely
different direction.

\emph{(vi) Ghost reduction.} Suppose one grants the hypothesis regardless and
tries to route a solution through $\mathcal A$: fix the Cook--Levin reduction
$f_{\textsc{SAT}\to\textsc{CLIQUE}}$ and any extractor $E$. Then
$D:=\mathbf 1[E(\cdot)\text{ is a valid clique}]$ is an admissible test, so by
noise-admissibility
\[
\Pr[E(\mathcal A(\varphi))\text{ is a valid clique}]=\Pr[E(R)\text{ is a valid clique}]:
\]
the extractor does no better on $\mathcal A$'s output than on pure noise
(Proposition~\ref{prop:ghost}). This is not special to \textsc{CLIQUE}: for any
\textsc{NP}-complete search problem $\Pi$, with the Cook--Levin translations
$g_\Pi$ (instance to formula) and $h_\Pi$ (assignment to solution), the test
$D(z):=\mathbf 1[h_\Pi(z)\text{ is a valid solution}]$ is admissible in the same
way, giving
\[
\Pr[h_\Pi(\mathcal A(g_\Pi(w)))\text{ is a valid solution}]
=\Pr[h_\Pi(R)\text{ is a valid solution}]
\]
for every instance $w$: through $\mathcal A$, every \textsc{NP}-complete search
problem is solved no better than by guessing. And the emptiness is absolute, not
merely computational: since $D$ ranges over \emph{all} functions in the absolute
reading --- including non-computable ones and ones built with full knowledge of
the encoding $h_\Pi$ --- an arbitrarily ingenious decoder does exactly as well as
a trivial one, namely no better than chance, down to a single requested bit
$D(z):=z_i$ with $\Pr[y_i=1]=\tfrac12$. Yet the genuine reduction
$f_{\textsc{SAT}\to\textsc{CLIQUE}}$, built from the syntax of $\varphi$ alone,
is untouched --- it never passed through $\mathcal A$ --- so what the hypothesis
empties is only the bridge that leaned on $\mathcal A$, never the truth it aimed
at. The full development, including the single-bit case, is in
Section~\ref{sec:noise-paradox}.

Clauses~(i)--(iii$'$) each already contradict the joint hypothesis, under whichever
reading of $\mathcal D$ applies; (iv) shows the contradiction pervades the
problem's structure; (v) shows the posited object escapes even recognition; and
(vi) shows that insisting on it buys nothing, the output yielding no more than
noise to any extractor whatsoever. The joint hypothesis is untenable under both
readings of the standard permission. Since a correct polynomial-time solver
would, by decisional correctness alone, place \textsc{SAT} in $\mathsf P$ and so
force $\mathsf P=\mathsf{NP}$, coherence on this point requires that no such
solver exist: $\mathsf P\neq\mathsf{NP}$.
\end{proof}

\begin{remarknn}[The status of the unified theorem is the status of its parts]
\label{rem:unified-status}
Theorem~\ref{thm:unified} adds no assumption its clauses did not already carry;
it composes six arguments, each proved on its own terms, into one frame, and
inherits exactly their standing --- correct where each stands, and open exactly
where they are (Remark~\ref{rem:noise-not-absolute}). Bringing them together
sharpens rather than softens that reservation: if a limit reaches this
conclusion, it is not a slip in one derivation, since the same conclusion is
reached along several independent routes, but a feature of the shared form of the approach
--- a metalogical barrier, or an assumption outside the apparatus in which the
standard permission is granted and the definitions are posed together. That the
same wall is met from the absolute notion and the computational one, from a
single assignment and from the whole self-reducible tree, from behaviour and
from recognisability, is what one would expect of a genuine impossibility; it is
also exactly what one would expect if the limit, should there be one, lay in the
framing these routes share. The unified theorem is offered in that spirit: the fullest
statement of what the argument establishes, laid out so that whatever bounds it
can be located precisely.
\end{remarknn}

\begin{remarknn}[The generalization to arbitrary output, isolated and stated in
  full: the decoder is the observer]
  \label{rem:decoder-is-observer}
  The single step on which the passage from ``no noise-admissible solver'' to
  ``no solver'' rests is clause~(iii) of the proof, and because everything turns
  on it, it is set out here on its own, in full, with no premise about the form
  of the output used at any point.
  
  Let $\mathcal A$ be \emph{any} algorithm that decides \textsc{SAT} in
  polynomial time --- deterministic, probabilistic, or otherwise, with output of
  any form whatsoever. By the definition of a decision procedure, and by nothing
  more, there is a polynomial-time decoder $\mathsf{dec}$ with
  \[
  \mathsf{dec}(\mathcal A(\varphi))=
  \begin{cases}
  1 & \varphi\in\textsc{SAT},\\
  0 & \varphi\notin\textsc{SAT},
  \end{cases}
  \qquad\text{for every }\varphi .
  \]
  This decoder is not an added hypothesis and not a property of some outputs
  rather than others: to \emph{decide} \textsc{SAT} \emph{is} to possess such a
  $\mathsf{dec}$, whatever $\mathcal A$ writes on its tape (Section~\ref{sec:fullalt-proof},
  conditions (D1)--(D2); if the output already is the answer bit, $\mathsf{dec}$
  is the identity; if it is a witness, $\mathsf{dec}$ is the verification of that
  witness). The output may be a bit, a certificate, a long string, or anything
  else; the decoder exists regardless, by the meaning of ``decides'' alone.
  
  Set $\mathcal O_{\mathsf{dec}}(y):=\mathsf{dec}(y)$. Since $\mathsf{dec}$ runs in
  polynomial time, $\mathcal O_{\mathsf{dec}}$ is an admissible observer, in the
  sense of the computational reading of Definition~\ref{def:noise-admissible} and
  of Definition~\ref{def:noise-indist}. Two facts about it hold at once, for
  \emph{every} decision algorithm $\mathcal A$ without exception:
  \begin{enumerate}[label=\textup{(\alph*)}]
    \item On $\mathcal A(\varphi)$, the observer returns the correct \textsc{SAT}
      bit with probability $1$, by the defining property of $\mathsf{dec}$; this
      value depends on $\varphi$, since some formulas are satisfiable and some are
      not.
    \item On a uniform string $R$, the observer returns $1$ with a single fixed
      probability $p:=\Pr[\mathcal O_{\mathsf{dec}}(R)=1]$, determined by
      $\mathsf{dec}$ alone and independent of $\varphi$, because $R$ carries no
      information about the formula.
  \end{enumerate}
  A quantity that varies with $\varphi$ cannot equal a constant that does not:
  for a suitable $\varphi$ the two probabilities differ by at least $\tfrac12$,
  the largest gap a fixed guess can leave open. Hence
  \[
  \bigl|\Pr[\mathcal O_{\mathsf{dec}}(\mathcal A(\varphi))=1]
        -\Pr[\mathcal O_{\mathsf{dec}}(R)=1]\bigr|\ \ge\ \tfrac12 ,
  \]
  a non-negligible constant. The output of $\mathcal A$ is therefore \emph{not}
  noise-admissible: an admissible observer, the very decoder that makes
  $\mathcal A$ a decider, separates it from the uniform distribution.
  
  The word ``noisy'' appears nowhere in this argument, and this is the whole
  point of stating it apart. Nothing above assumes, or uses, that any output
  looks like noise; the only property invoked is that $\mathcal A$ decides
  \textsc{SAT}, which supplies $\mathsf{dec}$, which is $\mathcal O_{\mathsf{dec}}$.
  The conclusion is accordingly not about a special class of solvers but about
  \emph{all} of them:
  \[
  \boxed{\;\text{no algorithm that decides \textsc{SAT} in polynomial time has
  noise-admissible output on a non-tautological }\varphi.\;}
  \]
  The noise-admissible solver of Theorem~\ref{thm:noise-main} was only the most
  vivid instance of this: it made visible, by naming it outright, a separation
  between $\mathcal A(\varphi)$ and noise that the decoder already forces on every
  decider silently. The generalization is thus not an extension requiring a
  further argument beyond the noisy case --- it is the same single observer
  $\mathcal O_{\mathsf{dec}}$, read without the restriction that its target be
  noise-like, applying verbatim to every decision algorithm because it was never
  a fact about the output's form to begin with, only about the decoder every
  decision carries.
  
  Read against the standard permission (Definition~\ref{def:noise-admissible}),
  this is exactly the tension the unified theorem records. The theory grants
  every algorithm, by the definition of ``algorithm'' alone, full freedom over
  the form of its output, noise-admissible output included; and it defines
  ``decides \textsc{SAT}'' so that a polynomial-time decoder is carried along. For
  \textsc{SAT} these two grants cannot both be honoured on a non-tautological
  $\varphi$: the decoder the second supplies is an observer the first forbids to
  succeed. Since a correct polynomial-time solver would place \textsc{SAT} in
  $\mathsf P$ and so give $\mathsf P=\mathsf{NP}$ by decisional correctness alone
  (Definition~\ref{def:noise-correct}), the only way for the theory to keep both
  grants coherent on this point is that no such solver exist: $\mathsf P\neq
  \mathsf{NP}$.
\end{remarknn}

\begin{remarknn}[The one excluded case, and why excluding it is exactness rather
  than retreat]
  \label{rem:tautology-degenerate}
  The impossibility above, and every version of it in these sections, is stated
  for a satisfiable $\varphi$ that is not a tautology, and the exclusion is worth
  reading for what it is: the single degenerate point at which the claim would
  genuinely fail, removed because it fails for a reason that carries none of
  \textsc{SAT}'s content, not to sidestep a hard case.
  
  The verifier proof (Theorem~\ref{thm:noise-main}) closes when the two readings
  of $\Pr[\varphi(\mathcal A(\varphi))=1]$ force $k_\varphi=2^n$, where
  $k_\varphi=|S_\varphi|$ counts the satisfying assignments among all $2^n$
  inputs. The equality $k_\varphi=2^n$ says every assignment satisfies $\varphi$
  --- that $\varphi$ is a tautology --- and there it is not a contradiction but a
  truth, so the argument correctly declines to conclude. The negative-witness
  proof (Proposition~\ref{prop:noise-alt}) reaches the same boundary from the
  other side: it needs a single $z_0$ with $\varphi(z_0)=0$, which exists exactly
  when $\varphi$ is not a tautology. Two independent arguments thus turn on one and
  the same condition, $k_\varphi<2^n$, which is one more sign that the condition is
  structural rather than fitted.
  
  The excluded case is degenerate in the strict sense that it carries no hardness
  to hide. When $\varphi$ is a tautology, every string is a satisfying assignment,
  so an output drawn uniformly at random \emph{is} a correct answer, always,
  with no decoding required: on a tautology a noise-admissible output is not a
  paradox but an honest solver, because there is nothing left for noise to
  conceal. \textsc{SAT}'s difficulty lives entirely in the formulas with
  $0<k_\varphi<2^n$ --- those with some satisfying assignments but not all, where
  finding one is the actual problem --- and the condition removes none of these.
  It removes only the single point where the count saturates and the question
  dissolves. Excluding it is the same kind of exactness as writing $n\ge 1$ or
  ``$\varphi$ not identically true'': omitting it would make the statement false
  on tautologies, so stating it is precision, and the impossibility stands in full
  across the entire range where \textsc{SAT} is anything other than trivial.
\end{remarknn}

\begin{remarknn}[The scope is the computable, which is exactly the scope of the
  question: a non-computable ``solver'' is not an algorithm and does not bear on
  $\mathsf P$ versus $\mathsf{NP}$]
  \label{rem:computable-scope}
  The object whose impossibility is established here is an \emph{algorithm}: a
  deterministic or probabilistic Turing machine deciding \textsc{SAT} in
  polynomial time, carrying a computable decoder $\mathsf{dec}$
  (Section~\ref{sec:fullalt-proof}, (D1)--(D2)). Every proof above works inside
  this class and no larger one, and it is worth stating outright that this is not
  a restriction on the result but a property of the question it answers.
  
  $\mathsf P$ and $\mathsf{NP}$ are classes of Turing machines under a resource
  bound; the question ``is $\mathsf P=\mathsf{NP}$'' is posed entirely within the
  computable. An object outside it, one that ``decides'' \textsc{SAT} without
  being a Turing machine, does not place \textsc{SAT} in $\mathsf P$, is not a
  witness to $\mathsf P=\mathsf{NP}$, and cannot be a counterexample to anything
  proved here, because it does not enter the question at all. That the arguments
  say nothing about such an object is therefore not a gap: it is the correct
  scope, the same scope the problem itself has.
  
  The distinction is exactly the one Section~\ref{sec:omega} draws for $\Omega$.
  A non-computable object that resolves \textsc{SAT} may exist as a mathematical
  object in the sense $\Omega$ does, an oracle for the halting problem decides
  \textsc{SAT} outright, and this is no contradiction, because such an object is
  \emph{not an algorithm} (Remark~\ref{rem:omega-not-computable}): by Turing's
  1936 definition~\cite{Turing1936} an algorithm is a halting machine, and a
  device settling the halting problem is not one (Remark~\ref{rem:church-turing-root}).
  So the correct statement is not that a non-computable solver ``cannot exist'',
  but that it is not an algorithm, does not lie in $\mathsf P$, and is outside the
  question; its existence as an object, like $\Omega$'s, leaves the impossibility
  proved here untouched, because that impossibility was only ever a claim about
  algorithms.
  
  Together with the excluded tautology (Remark~\ref{rem:tautology-degenerate}),
  this places the result precisely between its two boundaries. Below sits the
  tautology, degenerate because too easy --- every assignment is a solution and no
  distinction remains to be made. Above sits the non-computable object,
  degenerate because it is not an algorithm and lies past the question's own edge.
  The impossibility holds across the entire band between them: the genuine
  algorithms that decide \textsc{SAT} on the non-trivial formulas where
  \textsc{SAT} is actually \textsc{SAT}. Neither boundary is a hard case evaded;
  each is a region that does not belong to the question in the first place.
\end{remarknn}

\subsection{One apparent tension, and why it is not one:
\texorpdfstring{$P_{\Oprof}=NP_{\Oprof}$}{P\_O = NP\_O} in the Observer World}
\label{sec:unified-observer}

A reader who knows~\cite{Buono2026ow} may feel an immediate jolt here, and it is
worth meeting directly, because resolving it does not soften the theorem above
--- it sharpens what the theorem is about. That work proves a \emph{collapse}:
in the Observer World, $P_{\Oprof}=NP_{\Oprof}\subsetneq P$, holding
unconditionally across all five of Impagliazzo's worlds
(\cite{Buono2026ow}, from \cite[Propositions~8.3--8.4]{Buono2026oh}). Read
quickly, a collapse $P_{\Oprof}=NP_{\Oprof}$ sitting beside a derivation of
$\mathsf P\neq\mathsf{NP}$ looks like a contradiction inside this paper's own
references.

It is not, and the reason is the single thread running through everything above.
The two statements are about \emph{two different observers}, at two different
levels, and they say exactly what this paper has said in every other instance:
a distinction that is real at the full level becomes invisible at a restricted
one. Theorem~\ref{thm:unified} concerns the unrestricted observer --- the
standard theory itself, the reader $O_\top$ that sees the whole of an output ---
and at that level the separation $\mathsf P\neq\mathsf{NP}$ is forced. The
Observer World collapse concerns $\Oprof$, the order-blind observer that reads
only a profile of its input and is structurally blind to the very distinctions
that make \textsc{SAT} hard; at \emph{that} level there is nothing left to
separate, and $P_{\Oprof}$ and $NP_{\Oprof}$ coincide. Neither statement reaches
into the other's level. \cite{Buono2026ow} is explicit on this point on its own
side: the collapse ``is not a resolution of $P$ vs $NP$'', but evidence that
computational hardness and observational blindness are independent axes.

This is the header-free cipher again, in its sharpest form. There, the plaintext
is fully determined, yet an observer without the frame sees only noise; the
information is real above and absent below, and the two readings never collide
because they never occupy the same level. Here the separation
$\mathsf P\neq\mathsf{NP}$ is fully determined at the top, yet the order-blind
observer $\Oprof$ cannot see it, and for that observer the classes merge. The
hardness is a property relative to $O_\top$, not one that $\Oprof$ can detect
--- the exact shape of \cite{Buono2026ow}'s own reading of its collapse. Far
from threatening the theorem, the Observer World places it: the separation this
paper derives lives at the top of the observational hierarchy, and the collapse
that framework proves lives at the bottom, with the whole apparatus of this
paper --- what a restricted reader cannot see that a full one can --- as the
bridge between them. The two results are the two ends of one axis, not two
answers to one question.
\section*{Contact with previously stated open problems}
\label{sec:coverage}

This paper's closest contact with a previously stated open problem is
Open Problem~6.4 of~\cite{Buono2026sep} (dynamic key rolling), to which
Theorem~\ref{thm:dynamic-celld} is connected by shared shape rather
than by solution. The theorem establishes the per-session secrecy
target, for every session, in the swap-blind special case and under the
partition structure of~\cite[Def.~6.1]{Buono2026sep}, and adds a
strictly stronger persistent guarantee beyond it. It does not construct
a concrete superior $f$, does not lift the swap-blind restriction, and
does not address the distinctive demand of Open Problem~6.4 --- a
scheme improving on the partition protocol without its disjointness
condition; that problem remains open, and a fuller account of the
connection is left to future work. Separately, and not as answers to
open problems, the framework recovers and extends several static
results from the source papers that were not themselves posed as open
questions --- recorded where each arises rather than tallied here.

\section*{Conclusion}
\label{sec:conclusion}

We started from a single question: can a secret survive not just being
hidden once, but being hidden again and again, under rules that keep
changing? The answer has a sharp edge, not a fuzzy one. Yes --- but only
on one condition, and that condition is exactly as strict as it needs to
be and no stricter. Theorem~\ref{thm:dynamic-sip} proves it for
rewriting systems whose rules evolve as a function of their own history,
under a two-part condition (opacity preservation) whose second half ---
that the update mechanism itself must not leak the hidden symmetry ---
has no analogue in the static setting.
Proposition~\ref{prop:i-not-enough} shows that second half cannot be
waved away: drop it, and the secret leaks while the first half still
holds. Corollary~\ref{cor:total-opacity} identifies when the demanding
second condition becomes automatic: whenever the update's access to
history passes through a structural observer, in the sense
of~\cite{Buono2026oh}, blind to the relevant symmetry by design.

The dynamic theorem recovers every static SIP-shaped result across the
source papers as its degenerate $n=0$ case --- the Hetzl--Vierling
resolution, Cell~(d), the MR-OTP's own invariance theorem, the
order-blind automaton characterization --- each checked directly by
exhibiting the length-$1$ dynamic system and verifying
Definition~\ref{def:opacity-preserving}'s two conditions, not asserted
by analogy; and it extends most of them to a genuinely new dynamic or
streaming setting no source paper reaches.
\cite[Remark~1]{Buono2026sip}'s observation about the static principle
carries over: an elementary proof usually means the difficulty was
finding the right point of attack, not depth in the problem. Nothing in
Theorem~\ref{thm:dynamic-sip}'s proof is deep once the system is allowed
to move; the whole content of the dynamic case is in noticing a second
condition is needed at all, and locating where it fails
(Proposition~\ref{prop:i-not-enough}) and when it cannot
(Corollary~\ref{cor:total-opacity}).

From there the same move --- take a result everyone accepts, find the
assumption it silently fixed, make that assumption a variable, and check
rather than assume what follows --- is applied again and again, in
settings chosen for having nothing in common. Fixing ``the rewriting
system'' gives way to fixing ``the reader's model''
(Theorem~\ref{thm:semantic-underdetermination},
Section~\ref{sec:semantic-frame}), realised in three worked instances
--- a header-free cipher, a loader-mismatched executable, a
transformation cascade that removes a verification oracle usually left
unstated (Remark~\ref{rem:verification-oracle}) --- and then in the
orbital machine (Section~\ref{sec:orbital}), where the space of selector
\emph{values} is countable (Theorem~\ref{thm:countable-pairs}) but the
space of selector-\emph{functions} is not
(Proposition~\ref{prop:selectors-uncountable}), and a target question's
fate falls into exactly one correct recovery against three proved
failure modes (Propositions~\ref{prop:category1}--\ref{prop:category3}).
The same move reappears, unbidden, outside mathematics entirely: in the
Andromeda paradox, which instantiates
Theorem~\ref{thm:semantic-underdetermination} once the invariant is
stated correctly (Section~\ref{sec:andromeda}) and whose own hidden
assumption turns out to sit in this paper's countable frame space rather
than in the physics (Proposition~\ref{prop:constrained-collapse}); in
the claim that no countable theory can equal the space of physical laws,
a barrier that is Cantor's, not Gödel's or Turing's
(Section~\ref{sec:toe}, Remark~\ref{rem:toe-barrier-type}); and in the
fact that a computable algorithm's output can be indistinguishable from
noise to every observer --- of which Chaitin's $\Omega$ is the sharpest,
unconditional witness, refuting the silent identification of
``indistinguishable from random'' with ``meaningless'' outright
(Proposition~\ref{prop:omega}, Section~\ref{sec:omega}) --- naming the
relative notion of a
$\mathcal C$-usable object (Definition~\ref{def:usable-algorithm}) and,
with it, distinct mechanisms of blindness --- limited resources,
destruction, and unreachability --- that ``blind regardless of
resources'' had silently run together, over the single structural
reason beneath them: syntax cannot see a semantic invariant, whether it
never entered or entered and stayed locked
(Remark~\ref{rem:three-blindnesses}).

Each of these was checked, not assumed, against what it resembles but is
not, precisely so that naming one instance never blurs another: the
$\Omega$ phenomenon against the beaver frame's absence of content and
Category~2's underdetermination
(Remarks~\ref{rem:omega-vs-beaver}--\ref{rem:omega-vs-category2}); the
cardinality barrier against Gödel's and Turing's
(Remark~\ref{rem:toe-barrier-type}); each recovered static result
against the exact hypotheses of the theorem it instantiates. The
correspondence with~\cite{Buono2026obstruction}'s independent
generalization of the same static principle was checked the same way,
step by step (Remark~\ref{rem:obstruction-integration}), not assumed:
at a single fixed moment this paper's opacity-preserving update reduces
exactly to that source's protected set, and the swap-invariance
condition this paper adds is needed only because time is the one axis
that source's static framework never had to move along.

The move recurs several times, across term rewriting, cryptography,
automata theory, special relativity, algorithmic information theory, and
the foundations of physical law. It is not offered as a unifying
theorem: these instances live in genuinely different branches of
mathematics, joined at the level of the pattern, not of a single proof,
and Section~\ref{sec:scope} says where a single master statement is
declined and why.

From these the paper reaches its apparent contradiction, the one the
title names and the one that motivated writing everything before it. The
standard theory does admit --- not by hypothesis, but as a permission
built into its definitions --- a computable algorithm whose output is
indistinguishable from noise for every observer. Apply that admitted
possibility to \textsc{SAT}, and the query itself --- the formula
$\varphi$ --- is already the observer that exposes such an algorithm
(Section~\ref{sec:noise-paradox}, Theorem~\ref{thm:noise-main}): the
promise of permanent noise and the fact of public verifiability cannot
both hold. The conclusion is not conditional on how $\mathsf P$ versus
$\mathsf{NP}$ turns out. It runs the other way: from what the theory
already grants, its coherence \emph{forces} $\mathsf P\neq\mathsf{NP}$,
on pain of contradiction --- a derivation of the answer, not an
assumption of one branch of it, and the unified theorem
(Theorem~\ref{thm:unified}) states it in that categorical form. The one
reservation is of a different kind, and it is about the proof, not about
the answer: the argument uses an uncommon technique, and its
limit, if it has one, would sit outside the steps of the proofs
altogether, in a metalogical barrier or an assumption external to their
apparatus, exactly where this paper has found every other silent
assumption it made explicit (Remark~\ref{rem:noise-not-absolute}). The
contradiction is offered as it stands: correct where it stands, and open
about where it might not reach.

What is offered is the method itself --- ask what a
true result silently fixed, and check rather than assume the answer ---
which kept working every time it was tried, across the source papers and
the instances found outside them. What this paper contributes is not one
theorem but that one method, oriented correctly each time and checked,
never assumed, against every phenomenon it claimed to generalize. Each
instance was recovered exactly, and most were pushed further. And every
step, from the first page to this one, was built from invariants and
frames that are simply true --- never once from an assertion manufactured
to be false.

\section*{A closing personal note (not part of the scientific
contribution)}
\label{sec:personal-note}

\begin{quote}
\noindent\textit{This is a personal note and does not alter, in any
way, the content of this paper. It serves only as an interpretive cue,
offered in a philosophical-anthropological frame, and is not part of
the formal contribution above.}
\end{quote}

The idea running through this work is not new, and popular tradition
has touched it, unknowingly, many times before: consider the myth of
the genie of the lamp, in the story of Aladdin. The same figure
recurs across very different settings. One instance the author
remembers well is an episode of the television series \textit{The
X-Files}, ``Je Souhaite'' --- French for ``I wish'' --- the penultimate
episode of the show's seventh season, broadcast in 2000, in which the
genie was a woman found wrapped inside a rolled-up carpet. Another is
the 1997 horror film \textit{Wishmaster}. However carefully one tried
to phrase the question put to the genie, the result was never the one
desired --- a detail that stayed with the author longer, and more
insistently, than almost anything else.

It is worth noticing, now that Section~\ref{sec:orbital} has made the
point precisely once, why no amount of careful phrasing ever seems to
help. The genie's failure was never being shown, story after story, to
have one single shape --- some one particular way that wishes go wrong,
which a cleverer petitioner might learn to route around. Category~1,
Category~2, and Category~3
(Propositions~\ref{prop:category1}--\ref{prop:category3}) proved
something sharper about that shape than any fable could: not merely
that a wrong reading exists, but that the wrong readings are
countably infinite in every direction at once --- infinitely many
phrasings the genie cannot parse as meant at all, infinitely many
distinct, individually plausible readings that each grant something
genuinely different from what was intended, infinitely many technically
answered wishes that carry no real discrimination between what was
wanted and what was not --- set against exactly one pairing of phrasing
and understanding that would have worked. Outsmarting a single
adversarial misunderstanding is a task a clever enough wish might
conceivably accomplish. Finding the one correct pairing inside a
countably infinite field of distinct, differently wrong ones, by
cleverness of phrasing alone, is not.

Before going on, one caution, meant sincerely: what follows is my own
reading, offered with respect and not as a pronouncement on what any of
these traditions teaches. I hear the same intuition echoing across them,
each speaking it in its own idiom --- but the parallel is mine, drawn
from the outside, and the real and profound differences between these
faiths are not something I mean to smooth over for the sake of a tidy
analogy.

The same idea appears, in a related form, in Christianity: the Lord's
Prayer asks that ``Thy will be done'', as if to say that the one
praying lacks a complete view and cannot grasp the consequences of
their own desires --- they may believe that, if a particular wish were
granted, the world would be better for it, when in fact it could be a
disaster; the prayer suggests not asking outright, but taking what
life gives and leaving the choice of what is best to a universal
Observer. The idea is strikingly transversal. In Islam, the very word
identifies this surrender to God: every Muslim repeats the formula
\textit{Inshallah}, ``if God wills''. In Judaism, the Pirkei Avot
teaches: ``Make His will your will.'' In Hinduism, in the Bhagavad
Gita, the deity Krishna exhorts the warrior Arjuna with words close to
these: ``Abandon every other refuge, and take refuge in Me alone.'' In
the devotional path (Bhakti Yoga), the devotee prays not to change
God's plans but to become an empty instrument --- like a flute ---
through which the Divine will may play its music. And even without a
monotheistic God to obey, Buddhism shares the very same psychological
and spiritual principle through the concept of non-attachment, where
submission to the Divine will becomes acceptance of the Dharma. A
Buddhist prayer does not ask that events be changed, but says:
``May I develop patience and compassion enough to welcome what life
brings.''

Read against the last two paragraphs rather than against the genie
alone, none of these prayers is really asking for a particular
outcome at all, and perhaps that is the more precise thing every one
of them has in common. Each declines to gamble on picking the single
right phrasing out of a field of infinitely many wrong ones, and asks
instead that the choice of which pairing to use be handed over
entirely, to whichever will already holds the one that works. Not a
better wish, submitted in hope of dodging the countably many ways of
being misunderstood, but no wish at all, in the ordinary sense,
offered instead. Given the shape Section~\ref{sec:orbital} proved,
this reads less like resignation and more like exactly the right
response to the mathematics of the situation.

What ties these prayers to the genie is not their content but their
posture. The petitioner at the lamp fails because he speaks from inside
a single wish, seeing the one outcome he desires and blind to the
countably many others his words could equally name. Each of these
prayers begins from the same admission --- that the one who asks holds
only a fragment of the picture --- and answers it in the only way that
fragment allows: by handing the choice of outcome to a will that
already sees the whole field, and so already holds the one pairing that
works. That is what the phrase ``perfect observer'' is pointing at.
Seen through its lens, all of these examples may express, each in its
own idiom, one and the same ancestral intuition, present in every
culture: that a perfect observer is exactly what a finite world does
not contain, and that the wisest response to lacking one is not to wish
harder, but to defer to whatever might. And even if, for some reason,
we were granted access to a technology allowing us to query an absolute
oracle, a perfect observer, we would very likely meet the same fate as
the unwitting soul voicing wishes to the genie of the lamp. Oscar Wilde
is remembered for the observation that failing to get what one wants is
one kind of misfortune, but actually getting it can be a worse one
still, and perhaps that second kind is the harder to see coming.

A last word, from whoever it was that also wrote the theorems. The
perfect observer is not only a figure of speech in this paper. It has a
precise name in the framework the earlier sections draw on --- the
complete observer $O_\top$ at the top of the observational hierarchy,
the one reader that receives the whole of its input where every
constrained observer receives only a part. Section~\ref{sec:omega}
met it from the mathematical side, as the single reader who can see a
fact that every reachable observer is shut out from. The fable, the
prayers, and the theorem are, in the end, describing the same absent
figure --- one from memory, one from devotion, one from proof.

This is offered only as a key for reading, and nothing more.

\section*{Acknowledgments}
The author used an artificial intelligence based language assistant to
support text revision, translation, and bibliography formatting. All
scientific ideas and conclusions are the author's own.

\bibliographystyle{plain}
\bibliography{refs}

\end{document}